\documentclass[nonacm,format=acmsmall, review=false]{acmart}
\usepackage{acm-ec-26-arxiv}
\usepackage{booktabs} 
\usepackage[ruled]{algorithm2e} 
\usepackage{natbib} 

\SetAlFnt{\small}
\SetAlCapFnt{\small}
\SetAlCapNameFnt{\small}
\SetAlCapHSkip{0pt}
\setcitestyle{authoryear}

\setcopyright{cc}
\setcctype[4.0]{by}
\acmConference[Technical Report]{}{2026}{Arxiv}
\acmBooktitle{}
\acmDOI{}
\acmISBN{}

\usepackage[utf8]{inputenc}
\usepackage{tabularray}
\usepackage{thm-restate}
\usepackage{array,multirow,multicol}
\usepackage{cleveref}
\usepackage{wrapfig}
\usepackage{paralist}
\usepackage{xspace}
\usepackage{pifont}
\usepackage{dsfont}
\usepackage{cancel}
\usepackage{arydshln}
\usepackage{tikz}
\usetikzlibrary{shapes.geometric}
\usetikzlibrary{positioning}
\usepackage{bm}
\DeclareFontFamily{U}{dice3d}{}
\DeclareFontShape{U}{dice3d}{m}{n}{<-> s*[4] dice3d}{}
\usepackage{nicefrac}
\usepackage[normalem]{ulem}
\usepackage{lineno}
\usepackage{titletoc}
\usepackage{etoc}
\usepackage{framed}

\newcommand{\NashE}{\textsc{Nash}\xspace}
\newcommand{\NFCE}{\textsc{NFCE}\xspace}
\newcommand{\NFCCE}{\textsc{NFCCE}\xspace}
\newcommand{\EFCE}{\textsc{EFCE}\xspace}
\newcommand{\EFCCE}{\textsc{EFCCE}\xspace}
\newcommand{\AFCE}{\textsc{AFCE}\xspace}
\newcommand{\AFCCE}{\textsc{AFCCE}\xspace}
\newcommand{\AFCCCE}{\textsc{AF(C)CE}\xspace}

\newcommand{\ptime}{{\sc PTime}\xspace}
\newcommand{\ppad}{{\sc PPAD}\xspace}
\newcommand{\np}{{\sc NP}\xspace}
\newcommand{\NP}{\np}
\newcommand{\conp}{co{\sc NP}\xspace}
\newcommand{\pspace}{{\sc PSpace}\xspace}
\newcommand{\fixp}{{\sc FixP}\xspace}
\newcommand{\existreal}{\ensuremath{\exists\mathbb{R}}\xspace}
\newcommand{\apx}{{\sc APX}\xspace}
\newcommand{\exptime}{{\sc ExpTime}\xspace}

\newcommand{\Threshold}{\textsc{Threshold}\xspace}
\newcommand{\Any}{\textsc{Any}\xspace}

\newcommand{\GPhi}[0]{G_\Phi}
\newcommand{\Gphi}[0]{G_\varphi}
\newcommand{\astree}{\Gamma}
\newcommand{\univP}[1]{\ensuremath{\forall_{#1}}}
\newcommand{\assignP}[0]{\ensuremath{\eta}}
\newcommand{\formP}[0]{\ensuremath{f}}
\newcommand{\Puniv}[1]{Player~\univP{#1}\xspace}
\newcommand{\Passign}[0]{Player~\assignP{}\xspace}
\newcommand{\Pform}[0]{Player~\formP{}\xspace}

\newcommand{\Strategies}{\Sigma} 

\renewcommand{\phi}{\varphi}
\renewcommand{\epsilon}{\varepsilon}
\newcommand{\node}{v}
\newcommand{\Actions}{\mathcal{A}}
\newcommand{\ActionsP}[1]{\Actions_{#1}}
\newcommand{\actions}[1]{A(#1)}
\newcommand{\sw}{\omega}
\newcommand{\InfSet}[1]{\mathcal{I}_{#1}}
\newcommand{\bsigma}{{\boldsymbol{\sigma}}}
\newcommand{\CorrPlan}{\mu}
\newcommand{\profile}{\sigma}

\newcommand{\hist}[1]{\bm{h}(#1)}
\newcommand{\histP}[2]{h_{#2}(#1)}
\newcommand{\HistP}[1]{\mathcal{H}_{#1}}

\newcommand{\replaceH}[3]{\patternN{#3}[#2]}

\newcommand{\pattern}{\rho}
\newcommand{\patternN}[1]{\pattern(#1)}
\newcommand{\Relevant}{\mathcal{R}^{\mathsf{EFCE}}}

\newcommand{\correp}{\mu} 
\newcommand{\hexa}{\Xi^c}  
\newcommand{\polyCorr}{\hexa}

\newcommand{\Leaves}{\mathcal{L}}
\newcommand{\payoffP}[2]{u_{#2}(#1)}
\newcommand{\Reach}[1]{P_C(#1)}

\newcommand{\bestpayoff}[2]{\mathit{bp}^{#1}_{#2}}

\newcommand{\btau}{{\boldsymbol{\tau}}}
\newcommand{\bprofile}{\btau}  

\newcommand{\Node}{\mathcal{V}}
\newcommand{\payoff}{u}

\newcommand{\SigmaP}[1]{%
 \tikz[baseline=-1mm]{\node[inner sep=0pt,label={[label distance=-4.5pt]above:{\scriptsize $\rightharpoonup$}}] (char) {$\Sigma_{#1}$};}
}
\newcommand{\confPName}[2]{\mathds{1}_{#2}}
\newcommand{\confP}[3]{\mathds{1}_{#2}[#3]}
\newcommand{\confName}[1]{\mathds{1}_{#1}}
\newcommand{\conf}[2]{\mathds{1}_{#1}[#2]}
\newcommand{\confIn}[2]{\mathds{1}_{#1}[#2]}

\newcommand{\bh}{{\boldsymbol{h}}}
\newcommand{\ExpOmega}[2]{\mathbb{E}_{#1}(#2)}
\newcommand{\ExpBasic}[2]{\mathbb{E}_{#2}(U_{#1})}
\newcommand{\ExpNFCE}[3]{\mathbb{E}_{#2}(U_{#1} \mid #3)}
\newcommand{\ExpNFCEDeviation}[4]{\mathbb{E}_{#2}(U_{#1} \mid #3 \leadsto #4)}
\newcommand{\ExpNFCCE}[2]{\ExpBasic{#1}{#2}}
\newcommand{\ExpNFCCEDeviation}[3]{\mathbb{E}_{#2}(U_{#1} \mid \leadsto #3)}
\newcommand{\ExpEFCE}[4]{\mathbb{E}_{#2}(U_{#1} \mid #3, #4)}
\newcommand{\ExpEFCEDeviation}[5]{\mathbb{E}_{#2}(U_{#1} \mid (#3, #4) \leadsto #5)}
\newcommand{\ExpEFCCE}[3]{\mathbb{E}_{#2}(U_{#1} \mid #3)}
\newcommand{\ExpEFCCEDeviation}[4]{\mathbb{E}_{#2}(U_{#1} \mid #3 \leadsto #4)}

\newcommand{\proba}{\mathbb{P}}
\newcommand{\expect}{\mathbb{E}}

\newcommand{\CondEFCE}[5]{\proba_{#1}(#4,#5 \mid #2, #3)}
\newcommand{\CondEFCCE}[4]{\proba_{#1}(#3,#4 \mid #2)}

\newcommand{\CondNFCE}[4]{\proba_{#1}(#3,#4 \mid #2)}
\newcommand{\CondEFCEDeviation}[6]{\proba_{#1}(#5,#6 \mid (#2, #3) \leadsto #4)}
\newcommand{\CondEFCCEDeviation}[5]{\proba_{#1}(#4,#5 \mid #2 \leadsto #3)}
\newcommand{\CondNFCCEDeviation}[4]{\proba_{#1}(#3,#4 \mid \leadsto #2)}
\newcommand{\CondNFCEDeviation}[5]{\proba_{#1}(#4,#5 \mid #2 \leadsto #3)}

\newcommand{\RelevantH}{\mathcal{RH}}

\newcommand{\RelevantDEFCE}{\mathcal{RD}^{\mathsf{EFCE}}}
\newcommand{\RelevantEFCCE}{\mathcal{R}^{\mathsf{EFCCE}}}
\newcommand{\RelevantDEFCCE}{\mathcal{RD}^{\mathsf{EFCCE}}}
\newcommand{\RelevantNFCCE}{\mathcal{R}^{\mathsf{NFCCE}}}
\newcommand{\RelevantAFCE}{\mathcal{R}^{\mathsf{AFCE}}}
\newcommand{\RelevantAFCCE}{\mathcal{R}^{\mathsf{AFCCE}}}

\newcommand{\System}[1]{\mathit{Sys}(#1)}
\newcommand{\supp}{\mathit{supp}}

\newenvironment{proofclaim}{%
  \par\noindent{Proof of Claim.}\enspace\ignorespaces
}{%
  \hfill$\clubsuit$\par\addvspace{.6pc plus .2pc minus .1pc}
}

\NewDocumentCommand{\eqreftag}{o m}{%
  \ifx&#1&%
    \eqref{#2}%
  \else
    (EqZ($#1$))%
  \fi
}

\NewDocumentCommand{\eqrefCA}{o m}{%
  \ifx&#1&%
    \eqref{#2}%
  \else
    (C1($#1$))%
  \fi
}

\NewDocumentCommand{\eqrefCB}{o m}{%
  \ifx&#1&%
    \eqref{#2}%
  \else
    (C2($#1$))%
  \fi
}

\NewDocumentCommand{\eqrefCC}{o m}{%
  \ifx&#1&%
    \eqref{#2}%
  \else
    (C3($#1$))%
  \fi
}

\newcommand{\Xvar}[1]{p_{#1}}
\newcommand{\UvarA}[2]{\mathit{e}_{#1,#2}}
\newcommand{\Uvar}[1]{\mathit{e}_{#1}}
\newcommand{\UvarP}[1]{\mathit{e}^{#1}}

\newcommand{\Vvar}[3]{\mathit{d}^{#1,\neg #2}_{#3}}
\newcommand{\VvarA}[4]{\mathit{d}^{#1,\neg #2}_{#3,#4}}
\newcommand{\VIvar}[2]{\mathit{d}^{#1}_{#2}}
\newcommand{\VIvarA}[3]{\mathit{d}^{#1}_{(#2,#3)}}

\newcommand{\VEvar}[1]{\mathit{d}^i_{#1}}
\newcommand{\VAvarA}[3]{\mathit{d}^{\neg #2}_{(#1,#3)}}
\newcommand{\VACvar}[2]{\mathit{d}_{(#1,#2)}}

\newcommand{\DenomEFCE}{D_\CorrPlan(I,a)}
\newcommand{\DenomEFCCE}{D_\CorrPlan(I)}

\newcommand{\Hist}{\bm{\mathcal{H}}}

\newcommand{\replaceHI}[2]{\hist{#2}[#1]}
\newcommand{\replaceHEpsilon}[2]{\hist{#2}[\varepsilon]_{#1}}
\newcommand{\replaceHA}[3]{\hist{#3}[#1 \rightarrow #2]}
\newcommand{\replaceHAC}[2]{\hist{#2}[\cancel{#1}]}

\newcommand{\InfSets}{\mathcal{I}}

\newcommand{\action}{{a}}
\newcommand{\leaf}{\ell}
\newcommand{\distrib}{\mathbb{D}}
\newcommand{\Distrib}{\distrib_\bprofile}
\newcommand{\bprofileSet}{\mathbf{T}}  

\newcommand{\Rat}[0]{\mathbb{Q}}
\newcommand{\etr}{\exists\mathbb{R}}

\newcommand{\ExpP}[2]{\mathbb{E}_{#1}(#2)}

\newcommand{\Bproba}[2]{\mathbb{P}_{#1}^{#2}}
\newcommand{\bproba}[0]{\Bproba{\bprofile}{}}

\renewcommand{\astree}[0]{\ensuremath{\mathbb{B}^n \times T}\xspace}
\definecolor{blueopa}{RGB}{170, 201, 243}
\definecolor{ligreen}{RGB}{0,153,0}

\tikzset{choice/.style={circle, draw, minimum size = 12}}
\tikzset{guess/.style={rectangle,draw, inner sep = 4,minimum size=.55cm}}
\tikzset{check/.style = {diamond,draw, inner sep = 1,minimum size=.7cm}}
\tikzset{rand/.style = {regular polygon,regular polygon sides=3,draw, inner sep=0,minimum size =.8cm}}
\tikzset{form/.style = {regular polygon,regular polygon sides=5,draw, inner sep=0,minimum size =.65cm}}
\tikzset{end/.style= {}}

\newcommand{\gCompatibilityExample}[6]{

    \node[choice,anchor=center] (C#1) at (#2) {$x_#1$};
    
    \node[form,below left= .6cm and 1.3cm of C#1.center,anchor=center] (CFx#1) {$x_#1$};
    \node[form,below right= .6cm and 1.3cm of C#1.center,anchor=center] (CFnx#1) {$\overline{x_#1}$};

    \draw [-latex] (C#1) edge node[auto,swap] {\scriptsize{$x_#1$}} (CFx#1) ;
    \draw [-latex] (C#1) edge node[auto] {\scriptsize{$\overline{x_#1}$}} (CFnx#1);

    \node[end,below left= .8cm and .7cm of CFx#1.center, anchor=center] (GadUL) {$\scriptstyle{\bm{1}}$};
    \node[end,below right= .8cm and .7cm of CFx#1.center, anchor=center] (GadUR) {$\scriptstyle{\bm{0}}$};
    \draw[-latex] (CFx#1) edge node[auto,swap,inner sep=1pt] {$\scriptstyle{#3}$} (GadUL);
    \draw[-latex] (CFx#1) edge node[auto,inner sep=1pt] {$\scriptstyle{#4}$} (GadUR);

    \node[end,below left= .8cm and 0.7cm of CFnx#1.center, anchor=center] (GadUL) {$\scriptstyle{\bm{1}}$};
    \node[end,below right= .8cm and 0.7cm of CFnx#1.center, anchor=center] (GadUR) {$\scriptstyle{\bm{0}}$};
    \draw[-latex] (CFnx#1) edge node[auto,swap,inner sep=1pt] {$\scriptstyle{#5}$} (GadUL);
    \draw[-latex] (CFnx#1) edge node[auto,inner sep=1pt] {$\scriptstyle{#6}$} (GadUR);
    
    \draw[dashed,thick,opacity=0.1,fill=green] ([xshift=-2.6cm,yshift=-1.6cm]C#1.center) rectangle ([xshift=2.6cm,yshift=+.5cm]C#1.center);
    \node[above right= .2cm and 2.2cm of C#1.center,anchor=center,draw=none]{$C_{#1}$};
}

\newcommand{\gCompatibility}[2]{\gCompatibilityExample{#1}{#2}{t\not\ni\bar{x_#1}}{t\ni\bar{x_#1}}{t\not\ni x_#1}{t\ni x_#1}}

\newcommand{\gUncertainty}[2]{

    \node[choice,anchor=center] (Ux#1) at (#2) {$x_#1$};
    
    \node[guess,below left= .6cm and 1.3cm of Ux#1.center,anchor=center] (UGx#1) {$x_#1$};
    \node[guess,below right= .6cm and 1.3cm of Ux#1.center,anchor=center] (UGnx#1) {$\overline{x_#1}$};
    \draw [dashed] (UGx#1) edge node[above=-.08cm,pos=.3] {\scriptsize $\univP{#1}$} (UGnx#1);

    \draw [-latex] (Ux#1) edge node[auto,swap] {\scriptsize{$x_#1$}} (UGx#1) ;
    \draw [-latex] (Ux#1) edge node[auto] {\scriptsize{$\overline{x_1}$}} (UGnx#1);

    \node[end,below left= .8cm and .8cm of UGx#1.center, anchor=center] (GadL) {$\scriptstyle{\univP{#1}^+}$};
    \node[end,below = .8cm of UGx#1.center, anchor=center] (GadM) {$\scriptstyle{\bm{1}}$};
    \node[end,below right= .8cm and .8cm of UGx#1.center, anchor=center] (GadR) {$\scriptstyle{\univP{#1}^-}$};
    \draw[-latex] (UGx#1) edge node[auto,swap,inner sep=1pt] {$\scriptstyle{x_#1}$} (GadL);
    \draw[-latex] (UGx#1) edge node[auto,inner sep=1pt] {$\scriptstyle{\checkmark}$} (GadM);
    \draw[-latex] (UGx#1) edge node[auto,inner sep=1pt] {$\scriptstyle{\overline{x_#1}}$} (GadR);

    \node[end,below left= 0.8cm and 0.8cm of UGnx#1.center, anchor=center] (GadL) {$\scriptstyle{\univP{#1}^+}$};
    \node[end,below = .8cm of UGnx#1.center, anchor=center] (GadM) {$\scriptstyle{\bm{1}}$};
    \node[end,below right= 0.8cm and 0.8cm of UGnx#1.center, anchor=center] (GadR) {$\scriptstyle{\univP{#1}^-}$};
    \draw[-latex] (UGnx#1) edge node[auto,swap,inner sep=1pt] {$\scriptstyle{x_#1}$} (GadL);
    \draw[-latex] (UGnx#1) edge node[auto,inner sep=1pt] {$\scriptstyle{\checkmark}$} (GadM);
    \draw[-latex] (UGnx#1) edge node[auto,inner sep=1pt] {$\scriptstyle{\overline{x_#1}}$} (GadR);
    
    \draw[dashed,thick,opacity=0.1,fill=blue] ([xshift=-2.6cm,yshift=-1.6cm]Ux#1.center) rectangle ([xshift=2.6cm,yshift=+.5cm]Ux#1.center);
    \node[above right= .2cm and 2.2cm of Ux#1.center,anchor=center,draw=none]{$D_{#1}$};
}

\newcommand{\gKnowledge}[3]{

    \node[choice,anchor=center] (K#2x#1) at (#3) {$x_#1$};
    
    \node[guess,below left= .6cm and 1.3cm of K#2x#1.center,anchor=center] (K#2Gx#1) {$x_#1$};
    \node[guess,below right= .6cm and 1.3cm of K#2x#1.center,anchor=center] (K#2Gnx#1) {$\overline{x_#1}$};
    \draw [dashed] (K#2Gx#1) edge node[above=-.08cm,pos=.3] {\scriptsize $\univP{#2}$} (K#2Gnx#1);

    \draw [-latex] (K#2x#1) edge node[auto,swap] {\scriptsize{$x_#1$}} (K#2Gx#1) ;
    \draw [-latex] (K#2x#1) edge node[auto] {\scriptsize{$\overline{x_#1}$}} (K#2Gnx#1);

    \node[end,below left= .8cm and 0.7cm of K#2Gx#1.center, anchor=center] (GadUL) {$\scriptstyle{\bm{1}}$};
    \node[end,below right= .8cm and 0.7cm of K#2Gx#1.center, anchor=center] (GadUR) {$\scriptstyle{\bm{0}}$};
    \draw[-latex] (K#2Gx#1) edge node[auto,swap,inner sep=1pt] {\scriptsize{$x_#1$}} (GadUL);
    \draw[-latex] (K#2Gx#1) edge node[auto,inner sep=1pt] {\scriptsize{$\overline{x_#1}$}} (GadUR);

    \node[end,below left= .8cm and 0.7cm of K#2Gnx#1.center, anchor=center] (GadUL) {$\scriptstyle{\bm{0}}$};
    \node[end,below right= .8cm and 0.7cm of K#2Gnx#1.center, anchor=center] (GadUR) {$\scriptstyle{\bm{1}}$};
    \draw[-latex] (K#2Gnx#1) edge node[auto,swap,inner sep=1pt] {\scriptsize{$x_#1$}} (GadUL);
    \draw[-latex] (K#2Gnx#1) edge node[auto,inner sep=1pt] {\scriptsize{$\overline{x_#1}$}} (GadUR);
    
    \draw[dashed,thick,opacity=0.1,fill=red] ([xshift=-2.6cm,yshift=-1.6cm]K#2x#1.center) rectangle ([xshift=2.6cm,yshift=+.5cm]K#2x#1.center);
    \node[above right= .2cm and 2.2cm of K#2x#1.center,anchor=center,draw=none]{$_{#2}K_{#1}$};
}

\newcommand{\gKnowledgeReduced}[3]{

    \node[choice,anchor=center] (K#2x#1) at (#3) {$#1$};
    
    \node[guess,below left= .6cm and 1cm of K#2x#1.center,anchor=center] (K#2Gx#1) {$#1$};
    \node[guess,below right= .6cm and 1cm of K#2x#1.center,anchor=center] (K#2Gnx#1) {$\bar{#1}$};
    \draw [dashed] (K#2Gx#1) edge node[above=-.08cm,pos=.25] {\scriptsize $#2$} (K#2Gnx#1);

    \draw [-latex] (K#2x#1) edge node[auto,swap] {} (K#2Gx#1) ;
    \draw [-latex] (K#2x#1) edge node[auto] {} (K#2Gnx#1);

    \node[end,below left= .7cm and 0.4cm of K#2Gx#1.center, anchor=center] (GadUL) {$\scriptstyle{\bm{1}}$};
    \node[end,below right= .7cm and 0.4cm of K#2Gx#1.center, anchor=center] (GadUR) {$\scriptstyle{\bm{0}}$};
    \draw[-latex] (K#2Gx#1) edge node[auto,swap,inner sep=1pt] {} (GadUL);
    \draw[-latex] (K#2Gx#1) edge node[auto,inner sep=1pt] {} (GadUR);

    \node[end,below left= .7cm and 0.4cm of K#2Gnx#1.center, anchor=center] (GadUL) {$\scriptstyle{\bm{0}}$};
    \node[end,below right= .7cm and 0.4cm of K#2Gnx#1.center, anchor=center] (GadUR) {$\scriptstyle{\bm{1}}$};
    \draw[-latex] (K#2Gnx#1) edge node[auto,swap,inner sep=1pt] {} (GadUL);
    \draw[-latex] (K#2Gnx#1) edge node[auto,inner sep=1pt] {} (GadUR);
    
    \draw[dashed,thick,opacity=0.1,fill=red] ([xshift=-1.5cm,yshift=-1.5cm]K#2x#1.center) rectangle ([xshift=1.5cm,yshift=+.4cm]K#2x#1.center);

    \node[above right= .1cm and 1.2cm of K#2x#1.center,anchor=center,draw=none]{$_#2K_#1$};
}

\newcommand{\gCompatibilityReduced}[2]{

    \node[choice,anchor=center] (C#1) at (#2) {$#1$};
    
    \node[form,below left= .6cm and 1cm of C#1.center,anchor=center] (CFx#1) {$#1$};
    \node[form,below right= .6cm and 1cm of C#1.center,anchor=center] (CFnx#1) {$\bar{#1}$};

    \draw [-latex] (C#1) edge node[auto,swap] {} (CFx#1) ;
    \draw [-latex] (C#1) edge node[auto] {} (CFnx#1);

    \node[end,below left= .7cm and .4cm of CFx#1.center, anchor=center] (GadUL) {$\scriptstyle{\bm{1}}$};
    \node[end,below right= .7cm and .4cm of CFx#1.center, anchor=center] (GadUR) {$\scriptstyle{\bm{0}}$};
    \draw[-latex] (CFx#1) edge node[auto,swap,inner sep=1pt] {} (GadUL);
    \draw[-latex] (CFx#1) edge node[auto,inner sep=1pt] {} (GadUR);

    \node[end,below left= .7cm and 0.4cm of CFnx#1.center, anchor=center] (GadUL) {$\scriptstyle{\bm{1}}$};
    \node[end,below right= .7cm and 0.4cm of CFnx#1.center, anchor=center] (GadUR) {$\scriptstyle{\bm{0}}$};
    \draw[-latex] (CFnx#1) edge node[auto,swap,inner sep=1pt] {} (GadUL);
    \draw[-latex] (CFnx#1) edge node[auto,inner sep=1pt] {} (GadUR);
    
    \draw[dashed,thick,opacity=0.1,fill=green] ([xshift=-1.5cm,yshift=-1.5cm]C#1.center) rectangle ([xshift=1.5cm,yshift=+.4cm]C#1.center);

    \node[above right= .1cm and 1.2cm of C#1.center,anchor=center,draw=none]{$C_#1$};
}

\newcommand{\gUncertaintyReduced}[2]{

    \node[choice,anchor=center] (Ux#1) at (#2) {$#1$};
    
    \node[guess,below left= .6cm and 1cm of Ux#1.center,anchor=center] (UGx#1) {$#1$};
    \node[guess,below right= .6cm and 1cm of Ux#1.center,anchor=center] (UGnx#1) {$\bar{#1}$};
    \draw [dashed] (UGx#1) edge node[above=-.08cm,pos=.25] {$\scriptstyle{#1}$} (UGnx#1);

    \draw [-latex] (Ux#1) edge node[auto,swap] {} (UGx#1) ;
    \draw [-latex] (Ux#1) edge node[auto] {} (UGnx#1);

    \node[end,below left= .4cm and .6cm of UGx#1.center, anchor=center] (GadL) {};
    \node[end,below = .7cm of UGx#1.center, anchor=center] (GadM) {$\scriptstyle{\bm{1}}$};
    \node[end,below right= .4cm and .6cm of UGx#1.center, anchor=center] (GadR) {};
    \draw[-latex] (UGx#1) edge node[auto,swap,inner sep=1pt] {} (GadL);
    \draw[-latex] (UGx#1) edge node[auto,inner sep=1pt] {} (GadM);
    \draw[-latex] (UGx#1) edge node[auto,inner sep=1pt] {} (GadR);

    \node[end,below left= .4cm and .6cm of UGnx#1.center, anchor=center] (GadL) {};
    \node[end,below = .7cm of UGnx#1.center, anchor=center] (GadM) {$\scriptstyle{\bm{1}}$};
    \node[end,below right= .4cm and .6cm of UGnx#1.center, anchor=center] (GadR) {};
    \draw[-latex] (UGnx#1) edge node[auto,swap,inner sep=1pt] {} (GadL);
    \draw[-latex] (UGnx#1) edge node[auto,inner sep=1pt] {} (GadM);
    \draw[-latex] (UGnx#1) edge node[auto,inner sep=1pt] {} (GadR);
    
    \draw[dashed,thick,opacity=0.1,fill=blue] ([xshift=-1.5cm,yshift=-1.5cm]Ux#1.center) rectangle ([xshift=1.5cm,yshift=+.4cm]Ux#1.center);
    \node[above right= .1cm and 1.2cm of Ux#1.center,anchor=center,draw=none]{$D_#1$};
}

\theoremstyle{definition}
\newtheorem{theorem}{Theorem}
\newtheorem{remark}{Remark}
\newtheorem{proposition}{Proposition}

\declaretheorem[
  name=Claim,
  numberwithin=proposition,
  refname={Claim,Claims},
  Refname={Claim,Claims}
]{claim}

\title{On the Complexity of the Optimal Correlated Equilibria in Extensive-Form Games}

\author{Vincent Cheval}
\orcid{0000-0002-3622-2129}
\affiliation{%
    \institution{University of Oxford}
    \city{Oxford}
    \country{United Kingdom}
}
\author{Florian Horn}
\orcid{0000-0001-8872-4705}
\affiliation{%
    \institution{CNRS \& Universit\'e Paris-Cit\'e, IRIF}
    \city{Paris}
    \country{France}
}
\author{Soumyajit Paul}
\orcid{0000-0002-7233-2018}
\affiliation{%
    \institution{University of Liverpool}
    \city{Liverpool}
    \country{United Kingdom}
}
\author{Mahsa Shirmohammadi}
\orcid{0000-0002-7779-2339}
\affiliation{%
    \institution{CNRS \& Universit\'e Paris-Cit\'e, IRIF}
    \city{Paris}
    \country{France}
}

\begin{abstract}
A major open question in algorithmic game theory is whether normal-form correlated equilibria (NFCE) can be computed efficiently in succinct games such as extensive-form games~\cite{daskalakis2024efficient,peng2024fast,farina2023polynomial,huang2008computing,von2008extensive,papadimitriou2008computing}. Motivated by this question, we study the associated \textsc{Threshold} problem: deciding whether there exists a correlated equilibrium whose value exceeds a given threshold. 

We prove that this problem is \pspace-hard for \NFCE in multiplayer extensive-form games with perfect recall, even for  fixed thresholds. To contextualize this result, we also establish the complexity of the \Threshold problem for Nash equilibria in this setting, showing it is \existreal-complete. These results uncover a surprising complexity reversal: while optimal correlated equilibria are computationally simpler than optimal Nash in normal-form games, the opposite holds in extensive-form games, where computing optimal correlated equilibria is strictly harder unless \existreal is equal to \pspace. 

Building on this line of inquiry, we also address a related  question by von Stengel and Forges, who introduced the notions of extensive-form correlated equilibrium (\EFCE) and agent-form correlated equilibrium (AFCE). They asked  how difficult the \Threshold problem is for AFCE~\cite[Discussion and open problems]{von2008extensive}; we answer this question by proving that it is \np-hard, even in two-player games without chance nodes.

Complementing our hardness results, we establish tight complexity classifications for the \Threshold problem across several correlated equilibrium concepts---including \EFCE, \AFCE, normal-form coarse correlated equilibrium (\NFCCE), extensive-form  coarse correlated equilibria (\EFCCE), and agent-form coarse correlated equilibrium (\AFCCE). For each of these, we prove \np-completeness by providing matching \np upper bounds to the previously known hardness results. Finally, we also place the problem of computing arbitrary \NashE Equilibria in extensive-form games in \fixp, matching its complexity in normal-form games~\cite{etessami2010complexity}. Together, our results provide the most complete landscape to date for the complexity of optimal equilibrium computation in extensive-form games.
\end{abstract}

\begin{document}

\maketitle

\begin{framed}
Disclaimer: In a previous version of this technical report, we claimed that our \pspace-hardness proof for \NFCE could be extended to non-stochastic games, that is, to games without chance nodes. To support this claim, we proposed a generalization of the gadget introduced by~\citet{von2008extensive}, which was originally used to establish the \NP-hardness of \NFCE in two-player games without chance nodes. We subsequently discovered that this generalization is insufficient to support the claim. We have accordingly withdrawn the claim from the present version.
\end{framed}

\newpage

\setcounter{tocdepth}{2} 
\tableofcontents

\newpage

\section{Introduction}

Game theory provides a powerful framework for modeling and predicting the behavior of individuals or groups interacting in competitive or cooperative settings. Originally developed to understand strategic decision-making, it has found applications across diverse disciplines, including economics,  political science, and the social sciences, where it serves to analyze and design systems involving multiple rational agents. With the advent of artificial intelligence and machine learning, game-theoretic concepts have also gained prominence in computational contexts. For instance, Generative Adversarial Networks (GANs)~\cite{DBLP:conf/concur/Palamidessi020,DBLP:conf/nips/Liu0MCRYD20,DBLP:conf/nips/GoodfellowPMXWOCB14}, one of the more prominent frameworks for approaching generative artificial intelligence, model the interaction between a generator and a discriminator as a two-player game. In the realm of cybersecurity, game theory plays an important role to analyze network security~\cite{10.1145/3243734.3264421,10.1145/2382196.2382255} and protocols whose security relies on the rationality of the agents~\cite{10221879}, such as the Bitcoin Lightning Network~\cite{lightningNetwork} where monetary incentives are designed to enforce cooperation and prevent fraud.

A central concept in all these applications is that of equilibrium: a stable state in which no player has an incentive to deviate. Nash equilibrium has long been the dominant paradigm~\cite{BiloM16-Nash-ETR,daskalakis2013complexity,schoenebeck2012computational,etessami2010complexity,daskalakis2009complexity,conitzer2008new,chen2006settling,chen2006computing,fabrikant2004complexity}, especially suited for games involving independent and rational players. Correlated equilibrium, introduced by \citet{AUMANN197467}, has also attracted considerable attention as a broader coordination concept where players coordinate through a trusted mediator~\cite{daskalakis2024efficient,dagan2024external,peng2024fast,farina2023polynomial,zhang-2022-ptime-opt-eq,chen2022finding,DBLP:conf/sigecom/ZhangFCS22, DBLP:conf/aaai/FarinaBS20,farina2020polyPublicChance,Celli-et-al-opt-ex-ante:2019,Celli-et-al-learning-NFCCE-NEURIPS2019}. Correlated equilibrium  can be viewed as the outcome of a recommendation mechanism: a joint strategy is sampled from a publicly known distribution and each player privately receives an action recommendation with no incentive to deviate, assuming all other players adhere to  the given recommendations. As \citet{papadimitriou2008computing} put it, \emph{the complexity of equilibria is of fundamental importance in game theory}; in this spirit, we identify the following two questions as central to understanding the computational complexity of equilibrium concepts:

\begin{itemize}
    \item[(\Any)] Computing an equilibrium of the type of interest (and, in the case of pure equilibria, determining whether such an equilibrium exists),
    \item[(\Threshold)] Determining whether there exists an equilibrium that achieves a social welfare above a given threshold, and computing it. 
\end{itemize}

For normal-form games, where players simultaneously select a single action and payoffs are given in matrix form, these questions are well-understood. Computing \Any-correlated equilibrium can be done in polynomial time via linear programming, in stark contrast to the \ppad-completeness of computing \Any-Nash equilibrium~\cite{chen2006settling}. The contrast is even sharper for the \Threshold variant: while \Threshold-correlated equilibria can still be computed in polynomial time, the analogous problem for Nash equilibria was shown to be \np-complete~\cite{gilboa-zemel-89} for 2-player games and \existreal-complete~\cite{BiloM16-Nash-ETR} for games with  3 players and more.

Although originally defined for games in normal-form, Nash and correlated equilibrium have been extended to other representations which may be exponentially more succinct, including extensive-form, Bayesian~\cite{bayesian,950f7e7d-dc4b-32c4-b5ab-4d9f8f8c7849}, and convex games~\cite{daskalakis2024efficient,papadimitriou2008computing}. In this work, we focus on general-sum  extensive-form games with perfect recall, in which players never forget their own past actions or observations. These games model sequential interactions as trees whose nodes represent stochastic events or decision points of players, and whose leaves determine final payoffs. The general-sum setting allows for players’ payoffs to be arbitrarily related, capturing both conflicting and cooperative incentives. These games naturally capture a wide range of applications, such as security protocols~\citep{10221879} and sequential auctions~\citep{krishna2009auction}.

The landscape of complexity results for equilibrium computation in extensive-form games with perfect recall remained far from complete. For Nash equilibrium, known hardness results are inherited from the complexity of normal-form games: only the computation of \Any-$\epsilon$-\NashE equilibrium in extensive games was shown to be in \ppad~\cite{DBLP:conf/icalp/DaskalakisFP06} under some assumptions, leaving the exact complexity of exact \Any-\NashE and \Threshold-Nash open. The complexity class \fixp, introduced by \citet{etessami2010complexity}, contains problems which can be  described as fixed points of algebraic functions; the problem \Any-\NashE for 3 or more players was shown to be \fixp-complete in the same article. The \fixp-hardness extends to extensive-form games. Approximating \NashE and its refinements in extensive form games with $n$ players was shown to be \fixp-a complete, an approximate variant of \fixp, by \citet{ETESSAMI2021107}.

The most direct adaptation of correlated equilibrium to extensive-form games is the \emph{normal-form correlated equilibrium} (\NFCE for short), in which players receive recommendations for their entire strategies at the outset, before any moves are made. While this  solution concept allows for only a single round of communication between the mediator and the players~\cite{Celli-et-al-opt-ex-ante:2019} and satisfies important game-theoretic properties~\cite{Fuj23}, the complexity of computing any such equilibrium in extensive-form games remains a major open problem in the field~\cite{daskalakis2024efficient,peng2024fast,farina2023polynomial,huang2008computing,von2008extensive,papadimitriou2008computing}. For the \Threshold problem, \citet{von2008extensive} showed \np-hardness whereas the best known upper bound is \exptime which comes from the naive approach: explicitly constructing the game's payoff matrix, which can have size exponential in the size of the input game, and then applying the polynomial-time algorithm for \Threshold-\NFCE in normal form games.

A more expressive extension of correlated equilibrium, \emph{extensive-form correlated equilibrium} (EFCE for short), was introduced by \citet{von2008extensive} for extensive-form games with perfect recall. Unlike NFCE, where recommendations are issued once at the start, EFCE allows the mediator to provide a private recommendation at each decision point. This more naturally captures sequential correlation, as players can condition their actions on the game's unfolding and the mediator's ongoing guidance. While computing \Any-EFCE is known to be in polynomial time, computing \Threshold-EFCE Problem was only known to be \np-hard~\cite{von2008extensive}. Several tractable subclasses have been identified, including games without chance nodes~\cite{von2008extensive} and triangle-free games, which in particular contains games with public chance moves~\cite{farina2020polyPublicChance}. The complexity of this problem has also been studied through reduction to \emph{mediator-augmented games}, where the mediator has imperfect recall and the target solution concept is a Stackelberg equilibrium~\cite{DBLP:conf/sigecom/ZhangFCS22}.

Agent-form correlated equilibrium (AFCE), introduced by \citet{bayesian}, model situations where each action is chosen by an independent agent, even when different decision points are controlled by the same player. For such equilibria, the \Any problem is known to be in \ptime~\cite{von2008extensive,papadimitriou2008computing}. However, the complexity of the corresponding \Threshold problem remains less understood. In particular, while~\citet{von2008extensive} established \np-hardness of the \Threshold problem for \AFCE in two-player stochastic extensive-form games, the case without chance nodes was left as an open problem in their landmark paper. This stands in contrast to the situation for \EFCE, where the \Threshold problem admits a polynomial-time algorithm in games without chance moves.

Despite a surge of recent interest in correlated equilibria for extensive-form games, the precise complexity of the \Threshold problem for \EFCE, \AFCE and \NFCE had remained unresolved. This echoes the history of Nash equilibrium in normal-form games, where the complexity of the optimal value problem was long open until it was shown to be \existreal-complete~\cite{BiloM16-Nash-ETR}. In this paper, we settle the exact computational complexity of \Threshold problem for \EFCE and \AFCE, as well as for several other correlated equilibrium variants introduced below, with the exception of \NFCE, for which we establish a surprising lower bound.

We also answer a  question  from \citet[Discussion and open problems]{von2008extensive}: \emph{How difficult is the problem MAXPAY-AFCE for these games?}\footnote{This corresponds to the \Threshold problem for \AFCE for two-player games without chance nodes.} In~\Cref{thm:PureAFCEisNP-complete}, we show that this problem is \np-hard. Generally speaking, our result contribute to answering the broader question posed by \citet{papadimitriou2008computing} on extensive-form games: \emph{what kinds of super-succinct representations of correlated equilibria are needed?} Our results in~\Cref{thm:MaxCEisNP-complete} on \EFCE, \AFCE and other CE variants  are obtained by proposing a succinct representations  of optimal such equilibria. In contrast, the hardness in~\Cref{thm:MaxNFCEisPspace-hard} suggests that no such representation is likely to exist for optimal \NFCE. This separation provides new insight into the longstanding open question of computing \Any-\NFCE, and further highlights the growing importance of recent advances on computing correlated equilibria. Indeed, due to the intrinsic difficulty of computing optimal equilibria, several  algorithmic approaches have been proposed for computing optimal \EFCE in extensive-form games \emph{in practice}~\cite{DBLP:conf/sigecom/ZhangFCS22,pmlr-v139-farina21a,DBLP:conf/aaai/FarinaBS20}, and a general framework for understanding deviation concepts has been developed by \citet{Morrill2021EfficientDT, Morrill-et-all-AAAI-2021}. 

Several other refinements of correlated equilibrium have been proposed in the literature, which we will address in this paper. Coarse correlated equilibria (CCE), where players can only deviate before receiving a recommendation, were first introduced for normal-form games~\cite{Moulin-Vial-CCE-1979}, and later extended to extensive-form games under the name normal-form coarse correlated equilibria (NFCCE)~\cite{Celli-et-al-opt-ex-ante:2019}. The \Threshold problem for \NFCCE was shown to be \np-hard, while remaining solvable in polynomial time for games without chance moves. Subsequent work established \apx-hardness for this problem in games with chance nodes~\cite{Celli-et-al-learning-NFCCE-NEURIPS2019}.
Building on these developments, extensive-form coarse correlated equilibria (\EFCCE) were studied by \citet{DBLP:conf/aaai/FarinaBS20}, who showed that the \Threshold problem for \EFCCE is also \np-hard. The parameterized complexity of the optimal \EFCE, \EFCCE, and \NFCCE problems was further analyzed by \citet{DBLP:conf/sigecom/ZhangFCS22}. Other equilibrium concepts, which involve different mediator models, have been shown to admit polynomial-time solutions in general extensive-form games~\cite{zhang-2022-ptime-opt-eq}. The distinction between \NFCE, \EFCE and \AFCE equilibria and their coarse counterpart \NFCCE,  \EFCCE and \AFCCE lies in the timing where players decide whether to deviate: in coarse equilibria, the decision must be taken before receiving recommendations by the mediator whereas in standard equilibria, such decision can be taken after. This in particular implies that any equilibrium of a certain type is also a coarse equilibrium of that type, for instance, $\NFCE \subset \NFCCE$. The relationships between equilibrium concepts are as follows:
\begin{center}
    \begin{tabular}{ccccccccc}
        \NashE& $\subset$ &\NFCE & $\subset$ &\EFCE & $\subset$& \EFCCE& $\subset$& \NFCCE\\
        &&&&\rotatebox[origin=c]{-90}{$\subset$}&&\rotatebox[origin=c]{-90}{$\subset$}&&\\
        &&&&\AFCE & $\subset$ & \AFCCE &&
    \end{tabular}
\end{center}
\noindent where the inclusions $\NashE \subset \NFCE \subset \EFCE \subset \AFCE$ were shown in~\cite{von2008extensive} and $\EFCCE \subset \NFCCE$  in~\cite{DBLP:conf/aaai/FarinaBS20}\footnote{The inclusion $\NFCCE \subset \AFCE$ is claimed by \citet{Celli-et-al-opt-ex-ante:2019}; however, the relaxation provided by the \emph{agent} and \emph{coarse} extensions of \EFCE are orthogonal: it is possible to build a game showcasing the difference between normal and coarse equilibria wherein all \AFCE are \EFCE. For a concrete counterexample, consider a 2-player extensive-form game that simulates a concurrent setting: Player~1 and Player~2 move sequentially, but Player~2 does not observe Player 1’s action. Players chooses among actions $\{d,s,\ell\}$. Action $d$ ends the game immediately with payoff~$(2,2)$. When  players agree on~$\{s,\ell\}$ they get $0$, and when they disagree, the player who played $s$ receives 5 and the other receives 1. A uniform distribution over strategy profiles reflecting disagreement  constitutes a valid $\NFCCE$ with an expected payoff of 3, which is not an $\AFCE$.}.

\newcommand{\cross}{\ding{55}}

\begin{figure}[t]
    \begin{center}
\scalebox{0.8}{ 
        \begin{tblr}{
        colspec={|Q[c,wd=1.2cm]|Q[c,wd=1.cm]|Q[c,wd=1.cm]|Q[c,wd=1.cm]|Q[c,wd=1.cm]|Q[c,wd=1.cm] | Q[c,wd=1.cm]|Q[c,wd=1.cm]|Q[c,wd=1.4cm]|},
        rowspec={|Q[m,ht=.1cm]|Q[m,ht=.1cm]|Q[m,ht=.1cm]||Q[m,ht=.8cm]|Q[m,ht=.8cm]|Q[m,ht=.3cm]|Q[m,ht=.5cm]|Q[m,ht=.5cm]|Q[m,ht=.5cm]|Q[m,ht=.5cm]|Q[m,ht=.8cm]|},
        vline{7} = {9-10}{dashed}
        }
            & \SetCell[c=4]{}\Any & & & & \SetCell[c=4]{} \Threshold & & & \\
            \includegraphics[scale=.1]{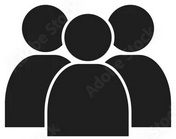} & \SetCell[c=2]{} 2 & & \SetCell[c=2]{} $n$ && \SetCell[c=2]{} 2 & & \SetCell[c=2]{} $n$ &\\
            \Pisymbol{dice3d}{102} & \cross & \checkmark & \cross & \checkmark & \cross & \checkmark & \cross & \checkmark\\
            \NashE & \SetCell[c=2]{}{\ppad-complete \\ {[2,3]}} & & \SetCell[c=2]{}{\fixp-hard~[7]\\ \textcolor{ligreen}{\fixp-complete} \\  \textcolor{ligreen}{\Cref{thm:AnyNashisFixP-complete}}} & & \SetCell[c=2]{}{\np-complete \\ {[1]}} & & \SetCell[c=2]{} {\existreal-hard~[11]~\\ \textcolor{ligreen}{\existreal-complete} \\ \textcolor{ligreen}{\Cref{thm:MaxNashisER-complete}}}\\
            \NFCE & \SetCell[c=2]{}{\ppad} & & \SetCell[c=2]{}{\textcolor{ligreen}{\fixp}\\ \textcolor{ligreen}{\Cref{thm:AnyNashisFixP-complete} 
            }} & & \SetCell[c=3]{}{\np-hard \\ {[5]}} & & & \SetCell[c=1]{}{\textcolor{ligreen}{\pspace-hard}\\ \textcolor{ligreen}{\Cref{thm:MaxNFCEisPspace-hard}}}\\
            \EFCE & \SetCell[c=4,r=5]{}{\ptime \\ {[4,6,8]}}& & & & \SetCell[r=3]{}{\ptime \\ {[5,9,10]}}  & \SetCell[r=5,c=3]{} {\np-hard \\ {[5,9,10]}\\ \textcolor{ligreen}{\np-complete} \\ \textcolor{ligreen}{\Cref{thm:MaxCEisNP-complete}}} & h & i\\
            \EFCCE & b & c & d & e & f & g & h & i\\
            \NFCCE & b & c & d & e & f & g & h & i\\
            \AFCE & b & c & d & e & \SetCell[r=2]{}{\textcolor{ligreen}{($*$)}} & g & h & i\\
            \AFCCE & b & c & d & e & & g & h & i\\
            \cline[dashed]{6-6}
            {pure \\ \AFCCCE} & \SetCell[c=8]{}{\textcolor{ligreen}{\np-complete} \\ \textcolor{ligreen}{\Cref{thm:PureAFCEisNP-complete}}}\\
        \end{tblr}
        }
    \end{center}
    \vspace{0.5em}
    \begin{center}
    \begin{minipage}[t]{.31\textwidth}
        \raggedright\scriptsize
        {[1]}: \citet{gilboa-zemel-89}\\
        {[2]}: \citet{KollerMV96}\\
        {[3]}: \citet{chen2006settling}\\
        {[4]}: \citet{papadimitriou2008computing}\\
    \end{minipage}
    \hfill
        \begin{minipage}[t]{.31\textwidth}
        \raggedright\scriptsize
        {[5]}: \citet{von2008extensive}\\
        {[6]}: \citet{huang2008computing}\\
        {[7]}: \citet{etessami2010complexity}\\
        {[8]}: \citet{DBLP:journals/geb/JiangL15}\\
    \end{minipage}
    \hfill
    \begin{minipage}[t]{.31\textwidth}
        \raggedright\scriptsize
        {[9]}: \citet{Celli-et-al-opt-ex-ante:2019}\\
        {[10]}: \citet{DBLP:conf/aaai/FarinaBS20}\\
        {[11]}: \citet{BiloM16-Nash-ETR}\\
        {\textcolor{ligreen}{($*$)}}: \Cref{thm:MaxCEisNP-complete} (membership)\\
        and \Cref{thm:PureAFCEisNP-complete} (hardness)
    \end{minipage}
\end{center}
\vspace{-0.3cm}
    \caption{Overview of known complexity results and our contributions (highlighted in \textcolor{ligreen}{green})}
    \vspace{-0.3cm}
\label{fig:summary-contribution}
\end{figure}

\paragraph{Our contributions.} The first part of our contributions is motivated by the central open question concerning the computational complexity of computing \NFCE in extensive-form games. While \citet{von2008extensive} previously showed that the \Threshold problem for \NFCE is \np-hard for extensive-form games, whether such equilibria admit succinct representations remained unknown~\cite{huang2008computing}. We establish the following in \Cref{sec:MaxNFCEisPSpace-hard}:

\begin{restatable}{theorem}{NFCEPSPACEh}
\label{thm:MaxNFCEisPspace-hard}
    The \Threshold problem for normal-form correlated equilibria (NFCE) in multiplayer extensive-form games with the number of players given as input is \pspace-hard. 
\end{restatable}

Interestingly, the proof of \Cref{thm:MaxNFCEisPspace-hard} relies on the construction of a polynomial-size extensive-form game in which any optimal \NFCE requires a correlation plan with exponentially large support, thereby suggesting that optimal \NFCE might not be succinctly representable. In comparison, we also consider the upper bound for the \Threshold problem for Nash equilibrium and establish the following in \Cref{sec:Nash}:

\begin{restatable}{theorem}{NashExistsRCom}
\label{thm:MaxNashisER-complete}
    The \Threshold problem for Nash equilibria in multiplayer extensive-form games with chance nodes and the number of players given as input is \existreal-complete.
\end{restatable} 

The combination of \Cref{thm:MaxNFCEisPspace-hard,thm:MaxNashisER-complete} reveals a surprising complexity reversal. Although optimal correlated equilibria are easier to compute than optimal Nash in normal-form games, the opposite is true in extensive-form games, where computing optimal correlated equilibria is harder, challenging the longstanding expectation that optimal Nash equilibria are inherently more complex~\cite{papadimitriou2008computing,gilboa-zemel-89}.

The second part of our contributions establishes  complexity upper bounds for the \Threshold problem across several equilibrium concepts, thereby yielding exact complexity for all equilibria but \NFCE. In particular, we introduce a succinct representation of optimal \EFCE  in multiplayer stochastic extensive-form games. As discussed earlier, our constructions add to the understanding of the representational complexity, also known as strategy complexity~\cite{KMSHT25,shirmohammadi2023beginner,KieferMSTW20}, of optimal correlated equilibria in extensive-form games, furthering the discussion initiated by~\citet{papadimitriou2008computing}. Combined with the hardness result of \citet{von2008extensive}, our matching upper bound leads to \np-completeness.

In \Cref{sec:Overview EFCE}, we derive corresponding tight complexity classifications for all correlated equilibria except \NFCE, as stated in the following theorem:

\begin{restatable}{theorem}{CorrelatedNPCom}
\label{thm:MaxCEisNP-complete}
    The \Threshold problem for \NFCCE, \EFCE, \EFCCE, \AFCE and \AFCCE in multiplayer extensive-form games with chance nodes and the number of players given as input are \np-complete.
\end{restatable}

Theorem~\ref{thm:MaxCEisNP-complete} provides a direct proof that the Threshold problem for EFCE belongs to NP; an alternative indirect proof can be obtained by combining the mediator-augmented-game construction of~\cite{DBLP:conf/sigecom/ZhangFCS22} with some proof techniques from ~\cite{DBLP:conf/icml/ZhangFS23}. We discuss the relationship between these two approaches in \Cref{sec:Overview EFCE}.

In \Cref{sec-AFCENPhard}, to complete the complexity landscape of the \Threshold problem for correlated equilibria, we show \np-hardness of the \Threshold problem in this setting, as well as for the \Any problem for pure equilibria for \AFCE and \AFCCE, which was unresolved. 

\begin{restatable}{theorem}{AFCEhard}
\label{thm:PureAFCEisNP-complete}
    The \Threshold problems for \AFCE and \AFCCE in  extensive-form games are \np-complete, even for two-player games without chance nodes. Moreover, the \Any problems for pure \AFCE and \AFCCE in this setting are also \np-complete.
\end{restatable}

We also show that the \Any problem for \Any-\NashE in extensive-form games belongs to \fixp, matching the lower bound derived from the normal-form case for 3 or more players (\Cref{sec:Nash}).

\begin{restatable}{theorem}{AnyNashisFixPcomplete}
\label{thm:AnyNashisFixP-complete}
    The \Any-\NashE problem in extensive-form games is \fixp-complete.
\end{restatable}

Our paper therefore completes the complexity landscape of the \Threshold and \Any problem for all equilibria but \NFCE, while improving the known lower bound of the \Threshold problem for \NFCE. A summary of the known results for the \Threshold and \Any problems across different equilibrium concepts, along with our contributions, is presented in \Cref{fig:summary-contribution}.

\paragraph{More related work.}
Another line of research has focused on the computation of approximate equilibria. The computation of such approximate $\epsilon$-EFCE and $\epsilon$-optimal EFCE has been studied in~\cite{dagan2024external,peng2024fast}, particularly in two-player stochastic games with perfect information over directed-acyclic graphs~\cite{Zhang-et-al:approx-opt-CE}. A complementary body of work considers learning algorithms for correlated equilibria, for which several lower bound results have been established~\cite{peng2024fast,dagan2024external,pmlr-v247-peng24a}.

\paragraph{Paper Outline.}
The main body of the paper provides an overview of our main contributions, highlighting the key ideas and results. Full technical details including formal constructions, proofs of hardness, and algorithmic results are deferred to Appendix for clarity and completeness.
\section{Extensive-Form Games, Correlations and Equilibria}
\label{sec:overview-extensive form games}

We study \emph{multi-player stochastic extensive-form games with imperfect information}. Such a game is played on a rooted, directed tree, whose nodes are controlled either by \emph{chance} or by one of the participating players. The game involves $n$ players, where $n \in \mathbb{N}$ is an input parameter, and each player is identified by a unique integer in~$\{1, \ldots, n\}$. The game starts at the root node and proceeds in turns: at a controlled node~$\node$, the player selects an action~$a$ from a set~$\Actions$ of available actions, leading deterministically to a unique successor~$\node'$, denoted by $\node \xrightarrow{a} \node'$. At a chance node, the successor~$\node'$ is chosen randomly according to a fixed probability distribution, denoted by $\node \xrightarrow{x} \node'$ with $x \in \mathbb{Q}$. The leaf nodes in the game tree are associated with a tuple in~$\mathbb{Q}^n$, referred to as the \emph{payoff} or \emph{utility vector}, which specifies the payoff received by each player when the game terminates at that leaf. We  define a linear function $\sw : \mathbb{R}^n \rightarrow \mathbb{R}$, named the \emph{objective function} of the game, which maps the tuple of expected payoff of players to a collective objective value. The objective function is typically used to capture criteria such as social welfare, which is defined as the expected sum of players' reward. For algorithmic purposes, we assume all probabilities and payoffs are rational numbers encoded in binary.

\paragraph{Information Sets and Perfect Recall.}
To model the players' \emph{imperfect information}, the   nodes controlled by each player~$i$ are partitioned into non-empty sets called \emph{information sets}. Write $\InfSet{i}$ for the collection of information sets for player~$i$, where each information set $I \in \InfSet{i}$ consists of nodes that are indistinguishable to player~$i$ during the course of the game. This indistinguishability constraint requires that all nodes within an information set have  the same available set of actions; write $\actions{I}$ for the actions available at~$I$.

\begin{wrapfigure}[8]{r}{0.41\textwidth}
\vspace{-0.6cm}
\tikzset{
triangle/.style = {regular polygon,regular polygon sides=3,draw,inner sep = 2},
circ/.style = {circle,fill=cyan!15,draw,inner sep = 1},
term/.style = {circle,draw,inner sep = 1.5,fill=black},
sq/.style = {rectangle,fill=gray!20, draw, inner sep = 4}
}

\begin{tikzpicture}[scale=0.80,label distance=0pt]

\tikzstyle{level 1}=[level distance=13mm,sibling distance=35mm]
\tikzstyle{level 2}=[level distance=10mm,sibling distance=18mm]
\tikzstyle{level 3}=[level distance=10mm,sibling distance=9mm]

\begin{scope}[->, >=stealth]
 \node(0)[triangle]{}
    child{  
    node(00)[circ,draw=black]{\scriptsize $S$}
        child{
        node(000)[sq,inner sep=3pt]{\scriptsize $N_1$}
            child{
            node(0000)[term,label=below:{\scriptsize $(4,10)$}]{} 
            edge from parent node[auto,swap,inner sep=0pt]{\scriptsize $A_E$} 
            }
            child{
            node(0001)[term,label=below:{\scriptsize $(0,6)$}]{}
            edge from parent node[auto,inner sep=0pt]{\scriptsize $R_E$}                
            }
        edge from parent node[auto,swap]{{\scriptsize $E_S$} }                
        }
        child{
        node(001)[sq,inner sep=3pt]{\scriptsize $N_2$}
            child{
            node(0010)[term,label=below:{\scriptsize $(4,10)$}]{} 
            edge from parent node[auto,swap,inner sep=0pt]{\scriptsize $A_G$} 
            }
            child{
            node(0011)[term,label=below:{\scriptsize $(0,6)$}]{}
            edge from parent node[auto,inner sep=0pt]{\scriptsize $R_G$}                
            }
        edge from parent node[auto]{{\scriptsize $G_S$} }
        }
    edge from parent node[auto,swap]{\scriptsize $\frac{1}{2}$} 
    }
    child{
    node(01)[circ]{\scriptsize $W$}
        child{
        node(010)[sq,inner sep=3pt]{\scriptsize $N_3$}
            child{
            node(0100)[term,label=below:{\scriptsize $(6,0)$}]{}
            edge from parent node[auto,swap,inner sep=0pt] {\scriptsize $A_E$} 
            }
            child{
            node(0101)[term,label=below:{\scriptsize $(0,6)$}]{}
            edge from parent node[auto,inner sep=0pt]{\scriptsize $R_E$}                
            }
        edge from parent node[auto,swap]{\scriptsize $E_W$} 
        }
        child{
        node(011)[sq,inner sep=3pt]{\scriptsize $N_4$}
            child{
            node(0110)[term,label=below:{\scriptsize $(6,0)$}]{} 
            edge from parent node[auto,swap,inner sep=0pt]{\scriptsize $A_G$} 
            }
            child{
            node(0111)[term,label=below:{\scriptsize $(0,6)$}]{}
            edge from parent node[auto,inner sep=0pt]{\scriptsize $R_G$}                
            } 
        edge from parent node[auto]{\scriptsize $G_W$} 
        }
    edge from parent node[auto]{~\scriptsize $\frac{1}{2}$}
    };

\end{scope}
\draw [dashed,thick,blue,out=22,in=158] (000) to (010);
\draw [dashed,red,thick,out=22,in=158] (001) to (011);
\end{tikzpicture}
\end{wrapfigure}

In this paper, we assume  \emph{perfect recall} meaning that all players remember all of their past actions. As information sets group together nodes that are indistinguishable to the player who controls them, accordingly, all nodes within the same information set must be reached with a unique action  sequence from that player's perspective.

\begin{example}
\label{runningexample}
    Consider the job signaling game on the right, between a job applicant and a company, described by \citet{von2008extensive}. The game starts at a chance node, which determines the applicant's type, either strong~($S$) or weak~($W$), with equal probability. The applicant, represented by circular nodes, sends either an eager cover letter  or a generic one, through signals $E_i$ and $G_i$ respectively for type $i$. The applicant knows her own type, hence the information sets of the applicant are singleton $I_S = \{S\}$ and $I_W = \{W\}$. The company sees only the type of letter,  and not the applicant's true quality. Hence the information sets of the company are $I_E = \{ N_1, N_3\}$ (in blue) representing the eager letter and $I_G = \{ N_2, N_4\}$ (in red) representing the generic letter.  In the figure, nodes in the same information set are connected by dotted lines. The company must decide whether to accept ($A_E$, $A_G$) or reject ($R_E$, $R_G$) based on that signal.

    The same game with the chance node at root  replaced by a company's node $N_0$ and two actions, say $N_0 \xrightarrow{a} S$ and $N_0 \xrightarrow{b} W$, would not satisfy the perfect recall assumption. Indeed, in the blue information set, reaching the node $N_1$ requires the company to play $a$ at the root node $N_0$ whereas reaching $N_3$ requires to play $b$. Hence, the nodes of the same information set do not have a unique action sequence from the company's perspective.
\end{example}

\paragraph{Correlated Equilibria: Normal-form and Extensive-Form.}
A (pure) \emph{strategy} for  player $i \in \{1,\ldots,n\}$ is a function $\sigma_i$ that assigns to each information set $I \in \InfSet{i}$ of the player, one of the actions available at $I$. A (pure) \emph{strategy profile} is a tuple  $\bsigma = (\sigma_1, \ldots, \sigma_n)$, where each $\sigma_i$ is a strategy of player $i$. Let $\Sigma_i$ denote the set of strategies for player~$i$, and let $\Sigma = \Sigma_1 \times \cdots \times \Sigma_n$ be the set of all strategy profiles. A \emph{correlated plan} $\CorrPlan$ is a distribution over~$\Sigma$.

We investigate several types of equilibria in extensive-form games, for instance \emph{normal-form correlated equilibria} (NFCE) and \emph{extensive-form correlated equilibria} (EFCE). Both rely on a \emph{mediator} that provides recommendations through a publicly known correlated plan. At the game’s start, the mediator  draws a strategy profile from the correlated plan and gives each player private action recommendations. A correlated plan is an equilibrium if no player has incentive to deviate, assuming all other player adhere to the given recommendations. The primary distinction between these equilibria is in the timing of the mediator's recommendations. In NFCE, after drawing a pure strategy profile $\bsigma$, the mediator privately  communicates each player’s full strategy $\sigma_i$ before the game begins. This means that each player is informed of their entire set of  recommendations  before committing to the equilibrium. In contrast, in EFCE, the mediator reveals the recommended action to each player only when they reach a relevant information set during the game, allowing players to condition their actions on observed play and the recommendation received. As a result, players update their beliefs and evaluate the expected utility of following the recommendation based on the current information.

While our focus in the main body of the paper is primarily on NFCE and EFCE, we also study  other equilibria, such as agent-form correlated equilibrium (AFCE) which is similar to EFCE except that players are only allowed to deviate at one single information set and then must follow the recommendation after deviating. The coarse variants of each of these equilibria require player to commit to recommendation \emph{before} receiving them.

\begin{example}
    In the job signaling game from Example~\ref{runningexample}, any NFCE results in the company always rejecting the applicant. For instance, consider the pure strategy profile $\bsigma = (\{I_S \mapsto G_S, I_W \mapsto G_W \},\{ I_E \mapsto R_E, I_G \mapsto R_G\})$, denoted more succinctly as $(G_S,G_W,R_E,R_G)$. This is an NFCE in which the company's payoff is $6$, while the applicant receives $0$. In fact, \citet{von2008extensive} show that any \NFCE for this game is also a Nash equilibrium where the company always plays $\{ I_E \mapsto R_E, I_G \mapsto R_G\}$, meaning that it systematically rejects the applicant. The proof proceeds by introducing a variable for each possible pure strategy profile in the support of the correlation plan, write out all the incentive constraints and then solving the resulting system of linear equations. 

    For example, the expected payoff of the company is 5 when always accepting the applicant, that is, when playing $\{ I_E \mapsto A_E, I_G \mapsto A_G\}$. Such strategy is dominated by $\{ I_E \mapsto R_E, I_G \mapsto R_G\}$ with an expected payoff of $6$ for the company, and thus is never played. 

    Now consider a case where the applicant receives from the mediator the recommendation to play the strategy $\{I_S \mapsto E_S, I_W \mapsto E_W \}$. Let $p$ and $p'$ be the probabilities of the pure strategy profiles $(E_S,E_W,A_E,R_G)$ and $(E_S,E_W,R_E,A_G)$ in the correlation plan respectively. The expected payoff of the applicant when following the recommendation is thus $5p$. However, if the applicant were to deviate and play $\{I_S \mapsto G_S, I_W \mapsto G_W \}$ instead, the expected payoff becomes $5p'$. Hence for the correlation plan to be a \NFCE, we must have $p \geq p'$. Applying this kind of reasoning to all possible deviations ultimately forces the company to adopt $\{ I_E \mapsto R_E, I_G \mapsto R_G\}$ as its only viable strategy in any \NFCE.

    This example also shows that an \EFCE can yield better social welfare than an \NFCE. Take the correlation plan $\CorrPlan$ defined as a uniform distribution over profiles $\{(E_S,E_W,A_E,R_G),~(G_S,E_W,R_E,A_G), (E_S,\allowbreak G_W,A_E,R_G),~(G_S,G_W,R_E,A_G)\}$. To show that $\CorrPlan$ is an \EFCE, we need to consider  different scenarios. When the chance node selects a strong applicant, the applicant has no incentive to deviate as the correlation plan ensures that he will be accepted by the company. If instead a weak applicant is selected  and he is recommended to play $E_W$ then he knows that with equal probability, the mediator has drawn either $(E_S,E_W,A_E,R_G)$ or $(G_S,E_W,R_E,A_G)$. Its expected payoff is $\frac{1}{2}\cdot 0 + \frac{1}{2}\cdot 6 = 3$. The applicant has no incentive to deviate as playing $G_W$ would also yield $3$ as expected payoff. With a similar reasoning, one can show that the company also does not have incentive to deviate. Under $\CorrPlan$, the expected payoff of the applicant is $3.5$ and the expected payoff of the company is $6.5$, yielding a social welfare of $10$, nearly double that of any \NFCE.
\end{example} 
\section{The \Threshold problem for \NFCE is \pspace-hard}
\label{sec:MaxNFCEisPSpace-hard}

We begin with a simpler reduction which shows the \conp-hardness of the \Threshold problem for \NFCE. Our proof is a reduction from the \textsc{Tautology} problem, which asks whether a given propositional formula in disjunctive normal form (DNF) is satisfied by every Boolean assignment of its variables~\cite{GareyJohnson79}. To this end, given a DNF formula $\phi$, we construct a polynomial-size extensive-form game $\Gphi$ such that there exists a correlated equilibrium in $\Gphi$ with expected objective value\footnote{Our reduction uses an arbitrary objective function for clarity. However, it can be adapted to social welfare by introducing an additional dummy player with an appropriate payoff: for example, at a leaf where all players receive 0 and the objective value is 1, the dummy would receive 1 instead .} of $1$ if and only if~$\phi$ is a tautology. 

Let us illustrate this reduction with a concrete example. Consider the DNF formula $\phi= (\, x_1 \, \wedge \, x_2\, ) \vee (\, \overline{x_1} \, \wedge \, \overline{x_2} \, \wedge \, x_3) \vee (\, \overline{x_2} \, \wedge \, \overline{x_3})$. Our game~$\Gphi$ encodes the formula-checking process as a multi-player interaction. It involves five players: a single \emph{assignment} player (Player~$\assignP$), a single \emph{formula} player (Player~$\formP$), and $3$ \emph{universal variable} players (a Player~$\univP{i}$ for each variable $x_i$). 

Intuitively, in each play, the assignment player chooses an assignment for the variables, the formula player checks that this assignment satisfies $\phi$, and the variable players ensure that all possible assignments are chosen. Each play of the game performs only a  partial random spot-check, which may be satisfied even if $\varphi$ is not a tautology. However, the only way to construct an \NFCE that passes all these spot-checks and achieves an expected social welfare of~$1$ is to have a uniform distribution over all possible assignments, which must all satisfy $\phi$. See \Cref{fig:MaxNFCEisCoNP-hard} for a schematic view of the constructed game $\Gphi$, where the leaf outcomes are labeled as follows: 
\begin{itemize}
    \item $\bm{0}$: all players receive $0$ and the objective value is $0$;
    \item $\bm{1}$: all players receive $0$ and the objective value is $1$;
    \item $\forall_i^+$: player $\univP{i}$ receives $+1$, all other players receive $0$, and the objective function is $0$;
    \item $\forall_i^-$: player $\univP{i}$ receives $-1$, all other players receive $0$, and the objective function is $0$.
\end{itemize}

\begin{figure}[ht]
\begin{center}
\scalebox{0.8}{
\begin{tikzpicture}
\tikzset{choice/.style={circle, draw, minimum size = 12}}
\tikzset{guess/.style={rectangle,draw, inner sep = 4,minimum size=.55cm}}
\tikzset{check/.style = {diamond,draw, inner sep = 1,minimum size=.7cm}}
\tikzset{rand/.style = {regular polygon,regular polygon sides=3,draw, inner sep=0,minimum size =.8cm}}
\tikzset{form/.style = {regular polygon,regular polygon sides=5,draw, inner sep=0,minimum size =.65cm}}
\tikzset{end/.style={rectangle}}





\gUncertainty{1}{0,0}
\gKnowledge{1}{2}{0,-2.4}
\gKnowledge{1}{3}{0,-4.8}
\gCompatibilityExample{1}{0,-7.2}{t_1,t_3}{t_2}{t_2,t_3}{t_1}

\gUncertainty{2}{5.6,-2.4}
\gKnowledge{2}{3}{5.6,-4.8}
\gCompatibilityExample{2}{5.6,-7.2}{t_1}{t_2,t_3}{t_2,t_3}{t_1}

\gUncertainty{3}{11.1,-4.8}
\gCompatibilityExample{3}{11.1,-7.2}{t_1,t_2}{t_3}{t_1,t_3}{t_2}

\node[choice,draw = none,above = .7 of Ux1] (x1){};
\draw[dashed] (x1) edge node[left, pos=.4] {\scriptsize{$\assignP$}} node[right, pos=.4] {\scriptsize{$A_1$}} (Ux1) {};
\draw[dashed] (Ux1) -- (K2x1) -- (K3x1) -- (C1) {};
\node[choice,draw = none,above = .7 of Ux2] (x2){};
\draw[dashed] (x2) edge node[left,pos=.4] {\scriptsize{$\assignP$}} node[right, pos=.4] {\scriptsize{$A_2$}} (Ux2) {};
\draw[dashed] (Ux2) -- (K3x2) -- (C2) {};
\node[choice,draw = none,above = .7 of Ux3] (x3){};
\draw[dashed] (x3) edge node[left,pos=.4] {\scriptsize{$\assignP$}} node[right, pos=.4] {\scriptsize{$A_3$}} (Ux3){};
\draw[dashed] (Ux3) -- (C3) {};

\node[right=2.1 of CFnx3.center,anchor=center] (CFx4) {};
\draw[dashed] (CFx1) -- (CFnx1) -- (CFx2) -- (CFnx2) -- (CFx3) -- (CFnx3) edge node[above, pos=.8] {\scriptsize{$f$}} (CFx4){};

\node[rand,above = 2.5cm of Ux2] (R){};



\draw[-latex] (R) -- ([xshift=+1cm,yshift=.65cm]Ux1.center) -- (Ux1);

\draw[-latex] (R) -- ([xshift=+2.7cm,yshift=.55cm]Ux1.center) -- ([xshift=+2.7cm,yshift=.65cm]K2x1.center) -- ([xshift=+1cm,yshift=.65cm]K2x1.center) -- (K2x1);

\draw[-latex] (R) -- ([xshift=+2.8cm,yshift=.45cm]Ux1.center) -- ([xshift=+2.8cm,yshift=.65cm]K3x1.center) -- ([xshift=+1cm,yshift=.65cm]K3x1.center) -- (K3x1);

\draw[-latex] (R) -- ([xshift=+2.9cm,yshift=.35cm]Ux1.center) -- ([xshift=+2.9cm,yshift=.65cm]C1.center) -- ([xshift=+1cm,yshift=.65cm]C1.center) -- (C1);

\draw[-latex] (R) -- ([xshift=+2.6cm,yshift=.65cm]Ux2.center) -- ([xshift=+1cm,yshift=.65cm]Ux2.center) -- (Ux2);

\draw[-latex] (R) -- ([xshift=+2.7cm,yshift=.75cm]Ux2.center) -- ([xshift=+2.7cm,yshift=.65cm]K3x2.center) -- ([xshift=+1cm,yshift=.65cm]K3x2.center) -- (K3x2);

\draw[-latex] (R) -- ([xshift=+2.8cm,yshift=.85cm]Ux2.center) -- ([xshift=+2.8cm,yshift=.65cm]C2.center) -- ([xshift=+1cm,yshift=.65cm]C2.center) -- (C2);

\draw[-latex] (R) -- ([xshift=2.6cm,yshift=.65cm]Ux3.center) -- ([xshift=1cm,yshift=.65cm]Ux3.center) -- (Ux3);

\draw[-latex] (R) -- ([xshift=2.7cm,yshift=.75cm]Ux3.center) -- ([xshift=2.7cm,yshift=.65cm]C3.center) -- ([xshift=1cm,yshift=.65cm]C3.center) -- (C3);



\end{tikzpicture}
}
\end{center}
\caption{The game~$\Gphi$ for $\phi = \protect\underbrace{(\, x_1 \, \wedge \, x_2\, )}_{\text{term } t_1} \bigvee \protect\underbrace{(\, \overline{x_1} \, \wedge \, \overline{x_2} \, \wedge \, x_3)}_{\text{term } t_2} \bigvee \protect\underbrace{(\, \overline{x_2} \, \wedge \, \overline{x_3}\, )}_{\text{term } t_3}$.} 
\label{fig:MaxNFCEisCoNP-hard}
\end{figure}

The game $\Gphi$ consists of a collection of independent gadgets, exactly one of which is selected uniformly at random in each play. In every gadget, the assignment player moves first by choosing the value of a single variable, then another player makes a move and the game ends with one of the four payoffs described above. The players have only limited information:
\begin{itemize}
    \item the assignment player only knows which variable they assign: they have an information set $A_i$ for each variable $x_i$; a strategy for that player is an assignment $\ell_1, \ldots, \ell_n$ of $x_1, \ldots, x_n$;
    \item the formula player has a single information set $C$; a strategy for that player is a term of $\phi$.
    \item the variable players know in which gadget they play, but not which literal was chosen by the assignment player: Player~$\univP{i}$ has an information set $_iK_j$ for each $j < i$ and an information set $D_i$; a strategy for that player is a partial assignment $\ell_1, \ldots, \ell_{i-1}$ plus either $x_i$, $\overline{x_i}$, or $\checkmark$.
\end{itemize}

\begin{wrapfigure}[8]{R}{0.35\textwidth}
\vspace{-.2cm}
\scalebox{.8}{
\begin{tikzpicture}
    \gCompatibility{i}{0,0}
    \node[overlay,left=1.8 of CFxi.center,anchor=center] (CFx1) {};
    \node[overlay,right=1.8 of CFnxi.center,anchor=center] (CFx4) {};
    \draw[dashed] (CFx1) -- (CFxi) -- (CFnxi) {};
    \draw[dashed] (CFnxi) edge node [above,pos=.9] {$\scriptsize{f}$} (CFx4){};
    \node[overlay,above=.7 of Ci] (Cj){};
    \draw[dashed] (Cj) edge node [left,pos=.4] {$\scriptsize{\assignP}$} node [right,pos=.4] {$\scriptsize{A_i}$}(Ci);
    \draw[dashed] ([xshift=+2.6cm,yshift=.65cm]Ci.center) -- ([xshift=+1.8cm,yshift=.65cm]Ci.center);
    \draw[-latex] ([xshift=+1.8cm,yshift=.65cm]Ci.center) -- ([xshift=+1cm,yshift=.65cm]Ci.center) -- (Ci);
\end{tikzpicture}
}
\caption*{Compatibility gadget $C_i$.}
\end{wrapfigure}
\textbf{(1)~Compatibility gadget $\bm{C_i}$.} For each variable $x_i$, there is a Compatibility gadget $C_i$. In this gadget, 
the formula player must choose a term of $\phi$. If that term is compatible with the literal chosen by the assignment player (\emph{i.e}, the opposite literal does not appear in the term), the objective value is $1$; otherwise, it is $0$. 
As the mediator does not know which gadget will be chosen randomly and the formula player cannot distinguish between them, the only way to guarantee an objective value of $1$ is to recommend a term compatible with \emph{all} the literals recommended to the assignment player. Such a term exists only if that assignment satisfies $\phi$. In turn, any \NFCE with expected objective value $1$ only uses assignments that satisfy $\phi$.

\vspace{0.5cm}

\begin{wrapfigure}[8]{L}{0.35\textwidth}
\vspace{.1cm}
\scalebox{.8}{
\begin{tikzpicture}
    \gKnowledge{j}{i}{0,0}
        \node[overlay,above=.7 of Kixj] (Kl){};
        \draw[dashed] (Kl) edge node [left,pos=.4] {$\scriptsize{\assignP}$} node [right,pos=.4] {$\scriptsize{A_j}$}(Kixj);
        \node[overlay,below=1.5 of Kixj] (Km){};
        \draw[overlay,dashed] (Km) -- (Kixj){};
        \node[left=2.5cm of Kixj]{};
        \node[right=2.5cm of Kixj]{};
        \draw[dashed] ([xshift=+2.6cm,yshift=.65cm]Kixj.center) -- ([xshift=+1.8cm,yshift=.65cm]Kixj.center);
        \draw[-latex] ([xshift=+1.8cm,yshift=.65cm]Kixj.center) -- ([xshift=+1cm,yshift=.65cm]Kixj.center) -- (Kixj);
\end{tikzpicture}
}
\caption*{Knowledge gadget $_iK_j$}
\end{wrapfigure}
\textbf{(2)~Knowledge gadget $\bm{_iK_j}$.} For each Player~$\univP{i}$ and each variable $x_j$ such that $j < i$, there is a Knowledge gadget $_iK_j$. In this gadget, Player~$\univP{i}$ must choose a literal for variable $x_j$. If they choose the same literal as the assignment player, the objective value is $1$; otherwise, it is $0$. As a consequence, the only way to build an \NFCE with expected objective value $1$ is to only have profiles with matching values for the assignment player and the universal players in all these gadgets. Note that Player~$\univP{i}$ receives recommendation specifying how to play in their $i-1$ knowledge gadgets at the start of the play, regardless of whether they visit that specific gadget or not.

\begin{wrapfigure}[8]{R}{0.35\textwidth}
\vspace{-.3cm}
\scalebox{.8}{
\begin{tikzpicture}
    \gUncertainty{i}{0,0}
    \node[overlay,above=.7 of Uxi] (Uj){};
    \draw[dashed] (Uj) edge node [left,pos=.4] {$\scriptsize{\assignP}$} node [right,pos=.4] {$\scriptsize{A_i}$}(Uxi);
    \node[overlay,below=1.5 of Uxi] (Uk){};
    \draw[overlay,dashed] (Uxi) -- (Uk){};
    \node[left=2.5cm of Uxi]{};
    \node[right=2.5cm of Uxi]{};
    \draw[dashed] ([xshift=+2.6cm,yshift=.65cm]Uxi.center) -- ([xshift=+1.8cm,yshift=.65cm]Uxi.center);
    \draw[-latex] ([xshift=+1.8cm,yshift=.65cm]Uxi.center) -- ([xshift=+1cm,yshift=.65cm]Uxi.center) -- (Uxi);
\end{tikzpicture}
}
\caption*{Doubt gadget $\bm{D_i}$}
\end{wrapfigure}
\textbf{(3)~Doubt gadget $\bm{D_i}$.} For each variable $x_i$, there is a Doubt gadget $D_i$. In this gadget, 
the corresponding Player~$\univP{i}$ may either accept a payoff of $0$ by playing $\checkmark$ (in which case the objective value is $1$), or they can make a fair bet on the value chosen for $x_i$ (in which case, the objective value is $0$). Therefore, in order to build an \NFCE with expected objective value $1$, player~$\univP{i}$ needs to always play $\checkmark$. 
As Player~$\univP{i}$ needs to know the values of $x_1,\ldots,x_{i-1}$ to handle the Knowledge gadgets, the \NFCE constraints imply that all possible assignments have to be selected with probability $2^{-n}$.

The Doubt gadgets, which force the mediator to weight each assignment equally, build on similar ideas to the \emph{generalized matching pennies} in~\cite{GoldbergPapadimitriou2006}. Their interaction with both the Knowledge gadgets and the initial chance vertex allow us to enforce a uniform distribution over an exponential number of assignments.

\textit{From co-\np-hardness to \pspace-hardness.}
We now extend our construction to show that the \textsc{Threshold} problem for \NFCE is \pspace-hard. The proof is via a reduction from the Quantified Boolean Formula problem, which asks whether a fully quantified Boolean formula is true.

Let $\Phi$ be a quantified Boolean formula over $n$ variables of the form $Q_1 x_1 \, Q_2 x_2 \, \cdots, \, Q_n x_n \, \varphi(x_1,x_2,\allowbreak\cdots, x_n)$ where  $Q_i \in \{\forall, \exists\}$  and $\varphi$ is in DNF. We construct a polynomial-size extensive-form game $G_\Phi$ from $\Phi$ such that $G_\Phi$ has a \NFCE with expected objective value 1 if and only if $\Phi$ is true.

The structure of $G_\Phi$ builds on the \conp case by using the same gadgets, with the following key modifications: only universally quantified variables are assigned a variable player and a Doubt gadget. However, just like in the \conp case, Player~$\univP{i}$ has a Knowledge gadget for each previously quantified variable, existential or universal. See \Cref{fig:MaxNFCEisPSpace-hard} for a schematic view of the constructed game $\GPhi$.

\begin{figure}[ht]
\begin{center}{
\scalebox{0.8}{
\begin{tikzpicture}
    \tikzset{choice/.style={circle, draw, minimum size = 8}}
    \tikzset{guess/.style={rectangle,draw, inner sep = 4,minimum size=.55cm}}
    \tikzset{check/.style = {diamond,draw, inner sep = 1,minimum size=.7cm}}
    \tikzset{rand/.style = {regular polygon,regular polygon sides=3,draw, inner sep=0,minimum size =.8cm}}
    \tikzset{form/.style = {regular polygon,regular polygon sides=5,draw, inner sep=0,minimum size =.65cm}}
    \tikzset{end/.style={rectangle}}





\gUncertaintyReduced{1}{0,0}
\gKnowledgeReduced{1}{3}{0,-2.1}
\gKnowledgeReduced{1}{5}{0,-4.2}
\gCompatibilityReduced{1}{0,-6.3}

\gKnowledgeReduced{2}{3}{3.4,-2.1}
\gKnowledgeReduced{2}{5}{3.4,-4.2}
\gCompatibilityReduced{2}{3.4,-6.3}

\gUncertaintyReduced{3}{6.7,-2.1}
\gKnowledgeReduced{3}{5}{6.7,-4.2}
\gCompatibilityReduced{3}{6.7,-6.3}

\gKnowledgeReduced{4}{5}{10.0,-4.2}
\gCompatibilityReduced{4}{10.0,-6.3}

\gUncertaintyReduced{5}{13.2,-4.2}
\gCompatibilityReduced{5}{13.2,-6.3}

\node[choice,draw = none,above = .7 of Ux1] (x1){};
\draw[dashed] (x1) edge node[right, pos=.4] {\scriptsize{$\assignP$}} node[left, pos=.4] {\scriptsize{$A_1$}} (Ux1) {};
\draw[dashed] (Ux1) -- (K3x1) -- (K5x1) -- (C1) {};

\node[choice,draw = none,above = .7 of K3x2] (x2){};
\draw[dashed] (x2) edge node[right, pos=.4] {\scriptsize{$\assignP$}} node[left, pos=.4] {\scriptsize{$A_2$}} (K3x2) {};
\draw[dashed] (K3x2) -- (K5x2) -- (C2) {};

\node[choice,draw = none,above = .7 of Ux3] (x3){};
\draw[dashed] (x3) edge node[right, pos=.4] {\scriptsize{$\assignP$}} node[left, pos=.4] {\scriptsize{$A_3$}} (Ux3){};
\draw[dashed] (Ux3) -- (K5x3) -- (C3) {};

\node[choice,draw = none,above = .7 of K5x4] (x4){};
\draw[dashed] (x4) edge node[right, pos=.4] {\scriptsize{$\assignP$}} node[left, pos=.4] {\scriptsize{$A_4$}} (K5x4) {};
\draw[dashed] (K5x4) -- (C4) {};

\node[choice,draw = none,above = .7 of Ux5] (x5){};
\draw[dashed] (x5) edge node[right, pos=.4] {\scriptsize{$\assignP$}} node[left, pos=.4] {\scriptsize{$A_5$}} (Ux5){};
\draw[dashed] (Ux5) -- (C5) {};

\node[form,draw=none,right=1.6 of CFnx5.center,anchor=center] (CFx6) {};
\draw[dashed] (CFx1) -- (CFnx1) -- (CFx2) -- (CFnx2) -- (CFx3) -- (CFnx3) -- (CFx4) -- (CFnx4) -- (CFx5) -- (CFnx5) edge node[above, pos=.7] {\scriptsize{$f$}} (CFx6){};

\node[rand,above = 2.1cm of Ux3] (R){};

\draw[-latex] (R) -- ([xshift=.8cm,yshift=.5cm]Ux1.center) -- (Ux1);

\draw[-latex] (R) -- ([xshift=-5.1cm,yshift=-.15cm]R.center) -- ([xshift=1.6cm,yshift=.5cm]K3x1.center) -- ([xshift=.8cm,yshift=.5cm]K3x1.center) -- (K3x1);

\draw[-latex] (R) -- ([xshift=-5.0cm,yshift=-.25cm]R.center) -- ([xshift=1.7cm,yshift=.5cm]K5x1.center) -- ([xshift=.8cm,yshift=.5cm]K5x1.center) -- (K5x1);

\draw[-latex] (R) -- ([xshift=-4.9cm,yshift=-.35cm]R.center) -- ([xshift=1.8cm,yshift=.5cm]C1.center) -- ([xshift=.8cm,yshift=.5cm]C1.center) -- (C1);

\draw[-latex] (R) -- ([xshift=1.5cm,yshift=.5cm]K3x2.center) -- ([xshift=.8cm,yshift=.5cm]K3x2.center) -- (K3x2);

\draw[-latex] (R) -- ([xshift=1.6cm,yshift=.4cm]K3x2.center) -- ([xshift=1.6cm,yshift=.5cm]K5x2.center) -- ([xshift=.8cm,yshift=.5cm]K5x2.center) -- (K5x2);

\draw[-latex] (R) -- ([xshift=1.7cm,yshift=.3cm]K3x2.center) -- ([xshift=1.7cm,yshift=.5cm]C2.center) -- ([xshift=.8cm,yshift=.5cm]C2.center) -- (C2);

\draw[-latex] (R) -- ([xshift=1.5cm,yshift=.5cm]Ux3.center) -- ([xshift=.8cm,yshift=.5cm]Ux3.center) -- (Ux3);

\draw[-latex] (R) -- ([xshift=1.6cm,yshift=-2.05cm]R.center) -- ([xshift=1.6cm,yshift=.5cm]K5x3.center) -- ([xshift=.8cm,yshift=.5cm]K5x3.center) -- (K5x3);

\draw[-latex] (R) -- ([xshift=+1.7cm,yshift=-1.95cm]R.center) -- ([xshift=1.7cm,yshift=.5cm]C3.center) -- ([xshift=.8cm,yshift=.5cm]C3.center) -- (C3);

\draw[-latex] (R) -- ([xshift=1.5cm,yshift=.5cm]K5x4.center) -- ([xshift=.8cm,yshift=.5cm]K5x4.center) -- (K5x4);

\draw[-latex] (R) -- ([xshift=1.6cm,yshift=.6cm]K5x4.center) -- ([xshift=1.6cm,yshift=.5cm]C4.center) -- ([xshift=.8cm,yshift=.5cm]C4.center) -- (C4);


\draw[-latex] (R) -- ([xshift=1.5cm,yshift=.5cm]Ux5.center) -- ([xshift=.8cm,yshift=.5cm]Ux5.center) -- (Ux5);

\draw[-latex] (R) -- ([xshift=1.6cm,yshift=.6cm]Ux5.center) -- ([xshift=1.6cm,yshift=.5cm]C5.center) -- ([xshift=.8cm,yshift=.5cm]C5.center) -- (C5);

\end{tikzpicture}
}


\caption{The game~$\GPhi$ for $\Phi$ of the form $\forall x_1\, \exists x_2\, \forall x_3\, \exists x_4\, \forall x_5\ \phi(x_1,x_2,x_3,x_4,x_5)$.} 

}
\label{fig:MaxNFCEisPSpace-hard}
\end{center}
\vspace{-0.5cm}
\end{figure}

Instead of guaranteeing that all assignments of $x_1,\ldots,x_n$ are present in an \NFCE, this restriction means that the assignments must cover a \emph{proof} of $\Phi$, \emph{i.e.} a non-empty set $P$ of assignments satisfying $\phi$ such that, if $x_i$ is a universal variable and $\ell_1,\ldots,\ell_i$ is a prefix of an  assignment in $P$, $\ell_1,\ldots,\overline{\ell_i}$ is also a prefix of an assignment in $P$.

\NFCEPSPACEh*
\section{The \textsc{Threshold} Problem for \EFCE is in \np}
\label{sec:Overview EFCE}

In normal-form games, the \textsc{Threshold} problem for correlated equilibria reduces to solving a linear system over the rationals, with a variable for each pure strategy profile and incentive constraints encoded as linear inequalities. Since the number of variables and constraints is polynomial in the game size, the problem admits a polynomial-time solution. This approach, however, does not extend to extensive-form games, where the number of pure strategy profiles grows exponentially with the number of information sets. To address this, previous works~\cite{farina2020polyPublicChance,DBLP:conf/aaai/FarinaBS20,KOLLER1996247} have focused on developing more compact representations of correlation plans using the notion of histories, in particular, those known as relevant histories.

More concretely, \citet{von2008extensive} showed that to verify whether a correlation plan~$\mu$  is an EFCE in 2-player extensive-form games without chance nodes, it suffices to consider the probabilities that $\mu$ assigns to certain combinations of  player histories---those that are relevant. A history, also known as a sequence in the literature, is  commonly understood as a sequence of information sets and actions that describe the path a player takes through the game tree. It can be viewed as a partial pure strategy, encapsulating the player's decisions within a specific information context. The relevant histories, in comparison, represent joint behaviors of the players along the game tree and are meant to provide enough information to encode both the correlation plan and the EFCE constraints. In their  definition, relevant histories are defined as  a subset of pairs  $(h_1,h_2)$, with $h_1$ and~$h_2$ representing the respective histories of players $1$ and~$2$, such that the last information sets in these histories, say~$I_1$ and~$I_2$, lie along the same path from the root to some node in the game tree.

Given a correlation plan~$\CorrPlan$, the probability $p_{(h_1,h_2)}$ of reaching a pair of histories $(h_1,h_2)$ is intuitively the probability of reaching any node whose path prefixes for the two players are exactly $h_1$ and $h_2$, respectively. Given such correlation plans~$\mu$, let  $\boldsymbol{p}_{\mu}$ denote the tuple of values $(p_{(h_1,h_2)})$, where $(h_1,h_2)$ ranges over all relevant histories. The set of all $\boldsymbol{p}_{\mu}$  can be captured by a system $\mathcal{S}$ consisting of polynomially many linear equalities~\cite{von2008extensive}. Moreover, they also showed that relevant histories can fully characterize the incentive constraints of EFCE using polynomially many inequalities. Solving these linear constraints therefore yields a polynomial time procedure for solving the \textsc{Threshold} problem for 2-player game without chance node.

It was shown by \citet{farina2020polyPublicChance} that the same definition of relevant histories yields polynomial-time algorithms for computing optimal EFCE in triangle-free games—a class that, in particular, includes games with public chance moves. Subsequent work by \citet{DBLP:conf/aaai/FarinaBS20} have used a geometric interpretation by viewing correlation plans and solutions to the system~$\mathcal{S}$ as points in polytopes within $\mathbb{R}^{|\Sigma|}$ and $\mathbb{R}^{|H_1| \times |H_2|}$, where $H_i$ denotes the set of histories of player~$i$.

\medskip  

Extending this approach to more general settings faces two major challenges:
\begin{itemize}
    \item First, in the presence of chance nodes, the solutions to the system~$\mathcal{S}$ no longer necessarily correspond to a valid correlation plan,  breaking the soundness of the proposed formulation, as shown by~\citet[Figure 3]{von2008extensive}. However, in the case of triangle-free games, the definition of relevant histories remain sound as the solutions to the system~$\mathcal{S}$ still correspond to valid correlation plans~\cite{farina2020polyPublicChance}.
    \item Second, even without chance nodes, natural extensions of the notion of relevant histories to $n$-player games can result in an exponentially large set of such histories.  For instance, the extension proposed in~\cite{DBLP:conf/aaai/FarinaBS20} defines relevant histories as $n$-tuples $(h_1,\ldots,h_n)$, requiring that for all  $i,j\in \{1,\ldots,n\}$ with $i\neq j$, the pair $(h_i,h_j)$ is relevant. Unfortunately, this definition may lead to an exponential number of relevant histories.
\end{itemize}
 
To overcome the second limitation, in \EFCE and the related solution concepts, including  \AFCE and coarse concepts, we introduce a new family of relevant histories tailored to the  equilibrium concepts under consideration. These histories are sufficient to compute both the expected payoffs under the equilibrium and the best deviation payoffs, thereby enabling the expression of incentive constraints.
For each of these settings, we define two types of histories: \emph{honest} and \emph{deviation} histories. Honest histories are common across all concepts, while deviation histories are carefully adapted to the specific structure of each equilibrium notion.

Below, we provide an intuitive overview of these notions in the context of EFCE, which bears similarity to~\cite[Definition 3.2]{DBLP:conf/sigecom/ZhangFCS22}, and present a summary of the corresponding notions for other concepts in~\Cref{fig:deviatehonest}. A key difference is that we define honest and deviation histories only for the leaves which allows us to extend the notion across different solution concepts.

In \EFCE, to compute  expected average payoffs when all players follow the recommendation, it suffices to consider the honest histories $\hist{v}=(\histP{v}{1},\ldots,\histP{v}{n})$ for each leaf $v$, where $\histP{v}{i}$ is the history that player $i$ must follow to reach $v$. To evaluate incentive constraints, it suffices to consider \emph{deviation histories}, which represent how a player~$i$ might deviate at an information set~$I$ after being recommended to play action~$a$. These histories, denoted by~$\replaceH{I}{a}{v}$, are constructed as follows:
\[
\replaceH{I}{a}{v} = (\histP{v}{1}, \ldots, \histP{v}{i-1}, \histP{I}{i}\cdot(I,a), \histP{v}{i+1}, \ldots \histP{v}{n})
\]

\begin{wrapfigure}[10]{r}{0.40\textwidth}
\vspace{-.4cm}

 \centering

\tikzset{
triangle/.style = {regular polygon,regular polygon sides=3,draw,inner sep = 2},
circ/.style = {circle,fill=cyan!15,draw,inner sep = 0},
term/.style = {circle,draw,inner sep = 1.5,fill=black},
sq/.style = {circle,fill=gray!20, draw, inner sep = 0},
sq1/.style = {circle,fill=red!20, draw, inner sep = 0}
}

\begin{tikzpicture}[scale=0.90]
\tikzstyle{level 1}=[level distance=9mm,sibling distance = 22mm]
\tikzstyle{level 2}=[level distance=7mm,sibling distance=12mm]
\tikzstyle{level 3}=[level distance=7mm,sibling distance=9mm]
\tikzstyle{level 4}=[level distance=7mm,sibling distance=7mm]


\begin{scope}[->, >=stealth]
\node (0) [triangle] {}
child {
  node (00) [circ] {\scriptsize $L_1$}
  child {
    node (000) [sq] {\scriptsize $L_2$}
    child {
      node (0000) [circ] {\scriptsize $L_3$}
          child {
          node (00000) [sq1] {\scriptsize $L_4$}
            child {
            node (000000) [term, label=below:{\scriptsize $\ell_1$}] {}
            edge from parent node [auto,inner sep = 1pt,swap,pos=0.6] {\scriptsize $e$}
            }
            child {
            node (000001) [term, label=below:{\scriptsize $\ell_2$}] {}
            edge from parent node [auto,inner sep = 1pt,pos=0.6] {\scriptsize $\bar{e}$}
            }
          edge from parent node [auto,inner sep = 1pt,swap,pos=0.65] {\scriptsize $c$}
          }
          child {
          node (00001) [term, label=below:{\scriptsize $\ell_3$}] {}
          edge from parent node [auto,inner sep = 1pt] {\scriptsize $\bar{c}$}
          }
      edge from parent node [auto,inner sep = 1pt,swap,pos=0.65] {\scriptsize $b$}
    }
    child {
      node (0001) [term, label=below:{\scriptsize $\ell_4$}] {}
      edge from parent node [auto,inner sep = 1pt] {\scriptsize $\bar{b}$}
      }
    edge from parent node [auto,inner sep = 1pt,swap,pos=0.65] {\scriptsize $a$}
  }
  child {
    node (001) [term, label=below:{\scriptsize $\ell_5$}] {}
    edge from parent node [auto,inner sep = 1pt,pos=0.8] {\scriptsize $\bar{a}$} 
  }
  edge from parent node [auto,inner sep=1pt,swap] {$\frac{1}{3}$}
}
child {
  node (01) [circ] {\scriptsize $R_1$}
  child {
     node (010) [term, label=below:{\scriptsize $r_5$}] {}
    edge from parent node [auto,inner sep = 1pt,swap,pos=0.8] {\scriptsize $a$}
  }
  child {
    node (011) [sq] {\scriptsize $R_2$}
    child {
      node (0110) [term, label=below:{\scriptsize $r_4$}] {}
      edge from parent node [auto,inner sep = 1pt,swap] {\scriptsize $b$}
    }
    child {
      node (0111) [circ] {\scriptsize $R_3$}
        child{
          node (01110) [term, label=below:{\scriptsize $r_3$}] {}
          edge from parent node [auto,inner sep = 1pt,swap] {\scriptsize $d$}
            }
        child {
          node (01111) [sq1] {\scriptsize $R_4$}
              child{
              node (011110) [term, label=below:{\scriptsize $r_2$}] {}
              edge from parent node [auto,inner sep = 1pt,swap,pos=0.6] {\scriptsize $e$}
                }
              child {
              node (011111) [term, label=below:{\scriptsize $r_{1}$}] {}
              edge from parent node [auto,inner sep = 1pt,pos=0.6] {\scriptsize $\bar{e}$}
              }
          edge from parent node [auto,inner sep = 1pt,pos=0.65] {\scriptsize $\bar{d}$}
          }
      edge from parent node [auto,inner sep = 1pt,pos=0.65] {\scriptsize $\bar{b}$}
      }
    edge from parent node [auto,inner sep = 1pt,pos=0.65] {\scriptsize $\bar{a}$} 
  }
  edge from parent node [auto,inner sep=1pt,] {$\frac{2}{3}$}
}
;
\end{scope}

\draw [dashed, thick, red, in=150,out=30] (00) to (01);
\draw [dashed, thick, blue, in=150,out=30] (000) to (011);
\draw [dashed, thick, brown, in=150,out=30] (00000) to (01111);
\end{tikzpicture}
\end{wrapfigure}

\noindent where  $\histP{I}{i}$ is the history of player~$i$ leading to any node in the information set~$I$.

Consider the 3-player game shown on the right, where players 1, 2, and 3 correspond to the blue, gray, and red nodes, respectively. The chance node transitions to $L_1$ or $R_1$ with probabilities $\frac{1}{3}$ and $\frac{2}{3}$, respectively. We denote the information sets by $I_1 = \{L_1,R_1\}$, $I_2 = \{L_2,R_2\}$ and $I_4 = \{L_4,R_4\}$ as indicated in the game tree. Since nodes $L_3$ and $R_3$ belong to singleton information sets, we refer to their respective information sets simply as $L_3$ and $R_3$. Each leaf node is labeled with its utility vector, denoted $\ell_1$, \ldots, $\ell_5$ on the left subtree and $r_1$, \ldots, $r_{5}$ on the right subtree. For simplicity, we identify each leaf by its utility vector. As described above, the honest histories of the leaves correspond precisely to their paths in the game tree. For example:
\[
\hist{\ell_1} = ((I_1,a)\cdot(L_3,c),(I_2,b),(I_4,e))\qquad \hist{\ell_4} = ((I_1,a),(I_2,\bar{b}),\varepsilon)\qquad
\hist{r_4} = ((I_1,\bar{a}),(I_2,b),\varepsilon)
\]

To capture all possible deviations of player 1 after reaching $I_1$ and being recommended~$a$, we consider the deviation histories of the leaves that can be reached by instead playing $\bar{a}$. These are:
\[
\begin{array}{c}
\replaceH{I_1}{a}{\ell_5} = ((I_1,a),\varepsilon,\varepsilon) \qquad \replaceH{I_1}{a}{r_4} = ((I_1,a),(I_2,b),\varepsilon) \qquad \replaceH{I_1}{a}{r_3} = ((I_1,a),(I_2,\bar{b}),\varepsilon)\\[1mm]
\replaceH{I_1}{a}{r_2} = ((I_1,a),(I_2,\bar{b}),(I_3,e)) \qquad \replaceH{I_1}{a}{r_{1}} = ((I_1,a),(I_2,\bar{b}),(I_3,\bar{e}))
\end{array}
\]

By considering both honest and deviation histories, we obtain a set $\Relevant$ of relevant histories of size at most $n \cdot N_I \cdot N_\ell \cdot |\Actions|$ where $N_\ell$ and $N_I$ are the number of leaves and information sets, respectively. We define a system of linear constraints of polynomial size that exactly captures the incentive constraints of an \EFCE. This system uses one variable for each relevant history in~$\Relevant$, where each such variable $p_{\hist{v}}$ for honest histories $\hist{v}$ represents the probability that the correlation plan induces the players to follow exactly the sequence of histories in the tuple.  However, for deviation histories~$\replaceH{I}{a}{v}$, these variables $p_{\replaceH{I}{a}{v}}$ are used to verify that  players have no incentive to deviate, conditioned to reaching information set~$I$ and being recommended~$a$. 

\begin{figure}[t]
\begin{center}
\renewcommand{\arraystretch}{1.2}
\begin{tabular}{|>{\centering\arraybackslash}m{1.3cm}|>{\centering\arraybackslash}m{1cm}|l|p{2.5cm}|}
    \cline{1-3}
    Honest
    & 
    All & 
    $\hist{v} = (\histP{v}{1}, \ldots, \histP{v}{n})$ & 
    \multicolumn{1}{c}{\parbox{2.5cm}{\centering When deviation}} 
    \\
    \hline
    \multirow{5}{*}{Deviation} &
    \EFCE &
    $(\histP{v}{1}, \ldots, \histP{v}{i-1}, \textcolor{green!60!black}{\histP{I}{i}\cdot(I,a)}, \histP{v}{i+1}, \ldots \histP{v}{n})$
    &
    \multirow{2}{*}{\parbox{2.5cm}{\emph{at $I$ after being recommended $a$}}}
    \\
    \cdashline{2-3}
    &
    \AFCE &
    $(\histP{v}{1}, \ldots, \histP{v}{i-1}, \textcolor{green!60!black}{\histP{I}{i}\cdot(I,a)\cdot h'}, \histP{v}{i+1}, \ldots \histP{v}{n})$
    &
    \\
    \cdashline{2-4}
    &
    \EFCCE &
    $(\histP{v}{1}, \ldots, \histP{v}{i-1}, \textcolor{green!60!black}{\histP{I}{i}}, \histP{v}{i+1}, \ldots \histP{v}{n})$
    &
    \multirow{2}{*}{\parbox{2.5cm}{\emph{at $I$}}}
    \\
    \cdashline{2-3}
    &
    \AFCCE &
    $(\histP{v}{1}, \ldots, \histP{v}{i-1}, \textcolor{green!60!black}{\histP{I}{i}} \cdot h', \histP{v}{i+1}, \ldots \histP{v}{n})$
    &
    \\
    \cdashline{2-4}
    &
    \NFCCE &
    $(\histP{v}{1}, \ldots, \histP{v}{i-1}, \textcolor{green!60!black}{\varepsilon}, \histP{v}{i+1}, \ldots \histP{v}{n})$
    &
    \emph{before beginning}\\
    \hline
\end{tabular}
\renewcommand{\arraystretch}{1}
\end{center}
\caption{
Our newly defined family of relevant histories is tailored to the equilibrium concept under consideration. The honest and deviation histories are only defined for leaves~$v$ in the game. For all equilibrium notions, each leaf $v$ is associated with a single honest history, but the deviation histories vary depending on the equilibrium concept. For each leaf $v$, 
(1)~in \EFCE and \AFCE, 
 there is one deviation history   for each pairs of $(I,a)$ of information set~$I$ in the path to~$v$ and each action $a$ available at $I$; 
(2)~in  \EFCCE and \AFCCE,  there is one deviation history for each information set~$I$ in the path to~$v$; and finally 
(3)  in \NFCCE, there is a single deviation history for each player $i$.
In the table, 
 $\histP{v}{j}$ denotes the history of player~$j$  to~$v$, and $\histP{I}{i}$ is the history of player~$i$  to  the information set~$I$. 
We further assume that the history of the leaf $v$ with respect  to player $i$ is of the form $\histP{v}{i} = \histP{I}{i} \cdot (I,a') \cdot h'$ for some action $a' \in \actions{I}$ and history~$h' \in \HistP{i}$.}
\label{fig:deviatehonest}
\end{figure}

\medskip

Going back to our example, write  $\ell_{i,j}$ and $r_{i,j}$  for the $j$-th coordinate of these vectors, which gives the payoff of the $j$-th player when reaching~$\ell_{i}$ and~$r_i$, respectively. The incentive constraints ensuring that player~1, conditioned to reaching $I_1$ and being recommended to play~$a$, has no incentive to deviate, are captured by the following three inequalities:
\begin{align}
\sum_{k=1}^{4} \frac{1}{3} \cdot \ell_{k,1}  \cdot p_{\hist{\ell_k}} + \frac{2}{3} \cdot r_{5,1} \cdot p_{\hist{r_5}} &\geq  \frac{2}{3}\cdot r_{4,1} \cdot p_{\replaceH{I_1}{a}{r_4}} + u^{R_3}\label{eq:incentive-honest vs deviation}\\
u^{R_3} &\geq \frac{2}{3}\cdot r_{3,1} \cdot p_{\replaceH{I_1}{a}{r_3}}\label{eq:incentive-max1}\\
u^{R_3} &\geq \frac{2}{3}\cdot r_{2,1} \cdot p_{\replaceH{I_1}{a}{r_2}} + \frac{2}{3}\cdot r_{1,1} \cdot p_{\replaceH{I_1}{a}{r_{1}}}\label{eq:incentive-max2}
\end{align}
where  the variable $u^{R_3}$ represents player 1’s optimal payoff at (their own) node~$R_3$. The left hand side of \eqref{eq:incentive-honest vs deviation} intuitively represents the expected payoff for player 1 (conditioned to reaching $I_1$ and being recommended to play~$a$) when all players follow the recommended strategy; in this case, only the leaves $\ell_1,\ell_2,\ell_3,\ell_4$ and $r_5$ can be reached. 
 
The right-hand side of \eqref{eq:incentive-honest vs deviation} represents the highest payoff player 1 can obtain by deviating, that is, by choosing action $\bar{a}$. This corresponds to the sum of the payoff from reaching $r_4$ and the best possible payoff player 1 can obtain after reaching~$R_3$. This value is in fact the maximum of the payoffs resulting from playing either $d$ or $\bar{d}$. To avoid introducing a max operator explicitly, we relax this computation using the two inequalities~\eqref{eq:incentive-max1} and~\eqref{eq:incentive-max2}. This relaxation is sound, as the value of $u^{R_3}$ is bounded above by the left-hand side of~\eqref{eq:incentive-honest vs deviation}. 

From the above, the incentive constraints for the players involve only a polynomial number of inequalities and variables. What remains is to address the first limitation in extending this approach to the more general setting with both chance nodes and $n$ players---namely, how to succinctly express a system of equations that captures a subset of~$\boldsymbol{p}_{\mu}=(p_{h})_{h\in \Relevant}$ among all correlation plans~$\mu$,  that contribute to expressing the optimal \EFCE. We note in passing that since computing an optimal \EFCE is \np-hard, it is infeasible to express all~$\boldsymbol{p}_{\mu}$, even with our newly defined family of relevant histories, with polynomially many equations (unless \ptime = \np). This motivates the need to focus on identifying  a small subset of ~$\boldsymbol{p}_{\mu}$ that are needed to expressing the optimal solution. We refer to such set of correlation plans as \textit{an optimality candidate set} for optimal \EFCE. 

We use Carathéodory's theorem to argue that there is always a small-size optimality candidate set.  To this aim, let $\Delta(\Sigma)$ denote the set of all distributions on strategy profiles~$\Sigma$. It suffices to show that the set $\polyCorr$ of all $\boldsymbol{p}_{\mu} \in \mathbb{R}^{\Relevant}$, with $\mu \in \Delta(\Sigma)$, is the image of the simplex $\Delta(\Sigma)$ under a linear map~$\kappa$. This implies that~$\polyCorr$ is convex, and its extreme points are the image of extreme points of~$\Delta(\Sigma)$. Recall that for all $d \in \mathbb{N}$ and all sets of points $A \subseteq \mathbb{R}^d$, Caratheodory's theorem dictates that any point in the convex hull of $A$ can be described as the convex combination of at most $d+1$ points from $A$. We deduce that for all $\mu \in \Delta(\Sigma)$ there exists another correlation plan~$\mu'$, whose support contain at most $k = |\Relevant|+1$ pure strategy profiles, and $\boldsymbol{p}_{\mu}=\boldsymbol{p}_{\mu'}$. 

Our procedure thus starts by guessing such a support $\bsigma_1,\ldots, \bsigma_k$ of pure strategy  profiles, and solving the \textsc{Threshold} problem over correlation plans on this support. We assign a variable $x_{\bsigma_j}$ to denote the probability of each $\bsigma_j$ in the (optimal) correlation plan, and we express  $\boldsymbol{p}_{\mu}$ using these $x_{\bsigma_j}$ variables. The resulting system of linear constraints is of polynomial size, yielding an \np procedure for deciding the \textsc{Threshold} problem for EFCE in arbitrary extensive-form games with chance moves and with the number of players given as input.

Our procedures for solving the \textsc{Threshold} problem for the various forms of correlated equilibria, namely NFCCE, EFCCE, AFCE, and AFCCE, all follow a similar structure. Each variant requires the corresponding definition given in~\Cref{fig:deviatehonest} from our newly defined family of relevant histories. In all cases, we rely on Carathéodory's theorem to ensure the existence of a small optimality candidate set. As explained in the case of \EFCE, after guessing the pure strategy profiles in the optimality candidate set, we can  formulate the \textsc{Threshold} problem into the feasibility of a  polynomial-size system of linear constraints. This leads to the following theorem:

\CorrelatedNPCom*

An alternative indirect proof of Theorem~\ref{thm:MaxCEisNP-complete} can be obtained by combining the mediator-augmented-game construction of \citet{DBLP:conf/sigecom/ZhangFCS22} with the small-support techniques of \citet{DBLP:conf/icml/ZhangFS23}. Although the latter studies zero-sum team-versus-team and team-versus-player games rather than Threshold-EFCE, it establishes the existence of small-support distributions that are realization-equivalent to a given strategy distribution, using Carathéodory's theorem. Applying an analogous compression argument to the mediator-augmented games of the former yields a NP upper bound for Threshold-EFCE. Our proof instead works directly in the original extensive-form game: we identify a polynomial-size family of relevant histories sufficient to encode both the equilibrium objective and the deviation constraints, and then apply Carathéodory's theorem to obtain a polynomial-support certificate.
\section{The \Threshold Problem for \AFCE and \AFCCE is \np-hard for pure 2-Player Games}
\label{sec-AFCENPhard}

\begin{figure}[t]
\begin{center}
\newcommand{\gadgetAFCE}[3]{
    \node[term,label=below:{\scriptsize $#2$},anchor=center,below left = 0.7cm and 0.4cm of #1.center] (v1) {};
\node[term,label=below:{\scriptsize $#3$},anchor=center,below right = 0.7cm and 0.4cm of #1.center] (v2) {};
\draw[-latex] (#1) edge node[auto,swap] {$\top$} (v1);
\draw[-latex] (#1) edge node[auto] {$\bot$} (v2);
}
\scalebox{0.8}{
\begin{tikzpicture}
    \tikzset{clause/.style={rectangle, draw,fill=red!20}}
    \tikzset{var/.style={circle, draw,fill=cyan!15}}
    \tikzset{term/.style = {circle,draw,inner sep = 1.5,fill=black}}

\node[clause,anchor=center] (R) at (0,0) {$R$};

\coordinate[left = 4.55cm of R.center,anchor=center] (Tmid);

\node[term,label=below:{\scriptsize $\textcolor{red}{-1}/\textcolor{blue}{2}$},anchor=center,below left = 1cm and 4.5cm of R.center] (T) {};

\node[clause,below = 1cm of R.center,anchor=center] (C2) {$C_2$};
\node[clause,anchor=center,left = 2cm of C2.center] (C1) {$C_1$};
\node[clause,anchor=center,right = 2cm of C2.center] (C3) {$C_3$};

\node[right = 3cm of C3.center,anchor=west] (Clause1) {clause $C_1$: $x$};
\node[below = 0.5cm of Clause1.west,anchor=west] (Clause2) {clause $C_2$: $\bar{x} \vee y$};
\node[below = 0.5cm of Clause2.west,anchor=west] (Clause3) {clause $C_3$: $\bar{x} \vee \bar{y}$};
\node[below = 0.5cm of Clause3.west,anchor=west] (Formula) {formula $\varphi = C_1 \wedge C_2 \wedge C_3$};

\draw (R) -- (Tmid);
\path (Tmid) edge[-latex]  node[auto,swap] {$\mathtt{end}$} (T);
\draw[-latex] (R) edge node[auto,swap] {$c_1$} (C1) ;
\draw[-latex] (R) edge node[auto,swap] {$c_2$} (C2) ;
\draw[-latex] (R) edge node[auto]  {$c_3$} (C3) ;

\node[var,below left = 1.5cm and 2cm of C1.center,anchor=center,inner sep=2pt] (X1) {$N_{x,1}$};
\node[var,below left = 1.5cm and 0cm of C1.center,anchor=center,inner sep=2pt] (X2) {$N_{\bar{x},2}$};
\node[var,below left = 1.5cm and -2cm of C1.center,anchor=center,inner sep=2pt] (X3) {$N_{\bar{x},3}$};

\node[var,below right = 1.5cm and 2cm of C2.center,anchor=center,inner sep=2pt] (Y1) {$N_{y,1}$};
\node[var,below right = 1.5cm and 4cm of C2.center,anchor=center,inner sep=2pt] (Y2) {$N_{\bar{y},2}$};

\draw[dashed,thick, blue] (X1) -- (X2);
\draw[dashed,thick, blue] (X2) -- (X3);
\draw[dashed,thick, blue] (Y1) -- (Y2);

\draw[-latex] (C1) edge node[auto,swap,pos=0.3,inner sep=1pt] {$x\vphantom{\overline{y}}$} (X1);
\draw[-latex] (C2) edge node[auto,swap,pos=0.3,inner sep=1pt] {$\bar{x}\vphantom{\overline{y}}$} (X2);
\draw[-latex] (C3) edge node[auto,swap,pos=0.3,inner sep=1pt] {$\bar{x}\vphantom{\overline{y}}$} (X3);

\draw[-latex] (C2) edge node[auto,pos=0.3,inner sep=1pt] {$y\vphantom{\overline{y}}$} (Y1);
\draw[-latex] (C3) edge node[auto,pos=0.3,inner sep=1pt] {$\overline{y}$} (Y2);

\gadgetAFCE{X1}{\textcolor{red}{-1},\textcolor{blue}{2}}{\textcolor{red}{0},\textcolor{blue}{0}}

\gadgetAFCE{X2}{\textcolor{red}{0},\textcolor{blue}{0}}{\textcolor{red}{-1},\textcolor{blue}{2}}

\gadgetAFCE{X3}{\textcolor{red}{0},\textcolor{blue}{0}}{\textcolor{red}{-1},\textcolor{blue}{2}}

\gadgetAFCE{Y1}{\textcolor{red}{-1},\textcolor{blue}{2}}{\textcolor{red}{0},\textcolor{blue}{0}}

\gadgetAFCE{Y2}{\textcolor{red}{0},\textcolor{blue}{0}}{\textcolor{red}{-1},\textcolor{blue}{2}}

\end{tikzpicture}
}
\end{center}
\caption{Game~$G_\varphi$ constructed from $\phi$ in the proof of {\np}-hardness for \AFCE and \AFCCE.
}
\label{fig:AFCE hardness}
\end{figure}

In this section, we address an open question posed by~\citet{von2008extensive} concerning the hardness of the \textsc{Threshold} problem for \AFCE in two-player games without chance nodes. We prove \np-hardness by a reduction from the \textsc{3SAT} problem, using a construction inspired by the same paper.
To outline the core ideas  behind the reduction, we first present an illustrative example.

Consider the CNF formula 
$\varphi = x \wedge (\bar{x} \vee y) \wedge (\bar{x} \vee \bar{y})$ over variables~$\{x,y\}$. 
Given such CNF formula, the reduction constructs a game $G_\varphi$; \Cref{fig:AFCE hardness} shows the construction for our example.
The resulting game involves two players: the $\varphi$-player (shown in blue), who aims to satisfy the formula~$\varphi$, and a \emph{spoiler} player (shown in red), whose goal is to falsify it.

The game begins with the spoiler  choosing either a clause $c_i$ of $\varphi$ or a special action $\mathtt{end}$, which terminates the game immediately. If a clause is chosen, the spoiler  then selects a literal from that clause. At this step, the spoiler  has perfect information. Importantly, all literal actions corresponding to the same variable lead to the same information set for the $\varphi$-player. For example, in \Cref{fig:AFCE hardness}, the $\varphi$-player has two information sets, one for $x$ and one for $y$. 
Upon reaching one of its information set, the $\varphi$-player selects an assignment ($\top$ or $\bot$) for the corresponding variable. 

Intuitively, since all literals of a variable lead to the same information set, the $\varphi$-player does not directly know which clause or literal was chosen by the spoiler player. If the assignment chosen by the $\varphi$-player satisfies the literal chosen by the spoiler then the $\varphi$-player receives a payoff of $2$ whereas the spoiler receives $-1$; otherwise, both players receive $0$.
Similarly, if the spoiler plays $\mathtt{end}$ and terminates the game then it receives $-1$ as payoff and the $\varphi$-player receives $2$. Observe that the social welfare at any leaf is at most $1$.

\AFCEhard*

\begin{proof}[Proof sketch.]

Below, we sketch an argument  showing that that  $\varphi$ is satisfiable if and only if there is an \AFCE  in $G_\varphi$ with  social welfare~$1$. 

When $\varphi$ is satisfiable by some assignment $\theta$, then the following defines a pure \AFCE $\CorrPlan_\theta$ with social welfare~$1$: 
    at root node $R$, $\CorrPlan_\theta$ recommends that the spoiler plays $\mathtt{end}$;
    at each clause $c_i$, the spoiler is recommended to play the literal in $c_i$ that is satisfied by $\theta$; and
    the $\varphi$-player is recommended to follow the assignment $\theta$ exactly.
Observe that the spoiler can only deviate at the root node, for example by selecting a clause~$c_i$. Indeed, according to the \AFCE definition, once a deviation has occurred, the spoiler must follow the recommended strategy for the remainder of the game. Since by construction, the chosen literal in every clause is  satisfied by $\theta$, the resulting payoff for the spoiler after playing $c_i$ would still be $-1$, giving it no incentive to deviate. The $\varphi$-player has also no incentive to deviate as it always receives the maximal payoff of $2$.

Conversely, assume now that $\varphi$ is not satisfiable. Assume by contradiction that an \AFCE (or \AFCCE) correlation plan $\correp$ with social welfare $1$ exists. Therefore, every strategy profile $\bsigma$ in the support of $\CorrPlan$ must lead to a terminal with payoff $(\textcolor{red}{-1},\textcolor{blue}{2})$. Each such strategy profile $\bsigma$ corresponds to a complete assignment $\theta$ of $\varphi$. 
Since $\varphi$ is not satisfiable, there must exist a clause $c_i$ such that $\theta \not\models \bsigma(C_i)$. In other words, by deviating unconditionally at the root and playing $c_i$, the spoiler  will fetch $0$ as payoff, which will overall lead to a payoff strictly greater than $-1$. This violates the \AFCE (or \AFCCE) condition. Hence, there can't be an \AFCE (or \AFCCE) whose support contains only pure profiles reaching leaves with payoff $(\textcolor{red}{-1},\textcolor{blue}{2})$. As a consequence, any other AFCE (or \AFCCE) will not have social welfare $1$. We conclude by noting that when $\varphi$ is not satisfiable, the game $G_{\varphi}$ admits no pure \AFCE. Consequently, the reduction also establishes that the \textsc{Any} problem for pure \AFCE is \np-hard. Moreover, this hardness result extends to \AFCCE: when $\varphi$ is satisfiable, any pure \AFCE is also an \AFCCE.
\end{proof}
\section{Complexity of \Threshold and \Any for \NashE}
\label{sec:Nash}

To the best of our knowledge, the only previous results about the \Threshold and \Any problems for \NashE equilibria in extensive games derive from known hardness results on the normal-form case, where the \Threshold problem is \existreal-complete~\cite{BiloM16-Nash-ETR}, while the \Any problem is \fixp-complete for 3 or more players~\cite{etessami2010complexity} and \ppad-complete for two players~\cite{chen2006settling}. For extensive games, it follows that the \Threshold problem is \existreal-hard, while the \Any problem is \fixp-hard for 3 or more players and \ppad-hard for two players.

Nash equilibria differ from more general Correlated Equilibria in the way players select their strategies. Instead of a correlation plan where the mediator chooses at random from a distribution over the set of pure strategy profiles, each player plays a \emph{mixed strategy} that chooses a pure strategy at random from a distribution over their respective sets of strategies. A mixed strategy profile is an element from $\big(\Delta(\Strategies_1)\times\ldots\times\Delta(\Strategies_n)\big)$, which can also be seen as a correlation plan since it is a subset of $\Delta\big(\Strategies_1\times\ldots\times \Strategies_n\big)$.

In extensive-form games, another way of defining Nash equilibria is to change the definition of strategies themselves. Instead of choosing a single action from each information set, a \emph{behavioral strategy (profile)} chooses a distribution over the actions playable in that information set ; $\tau: \InfSet{} \rightarrow \Delta(\Actions)$. 

In  his seminal paper, \citet{Kuhn+1953+193+216} showed that the two kinds of strategy are \emph{payoff equivalent}.
A mixed or behavioral strategy profile is a \NashE equilibrium if no player can improve their payoff by switching to another (pure) strategy. In our proofs, we use behavioral strategies, since they are able to represent mixed strategies with exponential support in polynomial size.

\subsection{\Threshold-\NashE in \existreal}
The idea behind showing the membership of \Threshold-\NashE in \existreal{} is to express the existence of a Nash equilibrium with expected value of the objective function at least $\lambda$ as the existence of a solution to a system of equations of polynomial size.

Firstly, a behavioral strategy profile can be expressed using variables $x_{I,a}$ for each action $a \in \actions{I}$.
\begin{align*}
\sum_{a \in \actions{I}} x_{I,a} = 1 &\qquad \forall i \in \{1, \ldots, n\}, \forall I \in \InfSet{i}
x_{I,a} \geq 0&\qquad \forall i \in \{1, \ldots, n\}, \forall I \in \InfSet{i},\forall a \in \actions{I}
\end{align*}
The expected payoff of player $i$ is $U_i = \sum_{\node \in \Leaves} \payoffP{\node}{i}\, \Reach{\node}\, \prod_{i=1}^n\, \prod_{(I,a) \in \histP{\node}{i}} x_{I,a}$. Given a behavioral strategy profile $\tau$, the probability of reaching a leaf $v$ following the behavioral strategies of all players except $i$ is $p_i(v) =  \Reach{v}\, \prod_{j \in \{1, \ldots, n\}\, \mid\, j\neq i} \, \prod_{(I,a) \in \histP{v}{j}} x_{I,a}$.

Using these values and the fact that players have perfect recall allows to characterize the best response for player $i$ to the strategy profile $\tau_{|-i}$ through a set of linear inequalities~\cite{KOLLER1996247}. The incentive constraint for player $i$ naturally corresponds to their payoff $U_i$ being greater than their best response $\bestpayoff{i}{}$, that is $U_i \geq \bestpayoff{i}{}$ for all $i \in \{1, \ldots, n\}$. And finally to ensure that the expected value of the objective function is at least the threshold $\lambda$ is given by the  constraint $\omega(U_1,\ldots,U_n) \geq \lambda$.

\NashExistsRCom*

\subsection{\Any-\NashE in \fixp}
\fixp, defined by \citet{etessami2010complexity} is the class of total search problems that can be reduced to finding a fixed point of a continuous function defined by an algebraic circuit. We prove inclusion in \fixp{} by constructing a function from the space of behavioral strategy profiles to itself, whose fixed points correspond to a Nash equilibrium. The function is defined as follows:
\begin{align*}
    f(\bprofile)(a)& = \frac{\bprofile(a)+\max\big(0, \delta(\bprofile,a)\big)}{1+\displaystyle\sum_{b|b \sim a}\max\big(0,\delta(\bprofile,b)\big)}\enspace,\\
\end{align*}
where $b \sim a$ signifies that $a$ and $b$ are playable from the same information set and $\delta(\bprofile,a)$ is the difference between the expected payoff of player $i$ in the profile $\bprofile$ and the expected payoff of player $i$ when playing action $a$ at $I$ (and everywhere else $\bprofile{}$ is followed). Essentially $\delta(\bprofile,a)$ is the benefit player $i$ would gain, if she deviates locally by playing only action $a$ at her information set $I$. 

A key property of the difference measure $\delta$ is that at $I$, its mean over the actions at $I$ is $0$.
\begin{restatable}{proposition}{propSumDelta}
\label{prop:sumdelta}
    Let $I$ be an information set and $\bprofile$ be a behavioral strategy profile. We have:
    \begin{align*}
        \sum_{a \in A(I)} \bprofile(a) \cdot \delta (\bprofile,a) &= 0 \enspace .
    \end{align*}
\end{restatable}

This consequently leads to the fact that the $\delta$ is always non-negative, i.e. players don't benefit by deviating locally at information sets. 
\begin{restatable}{lemma}{lemmaNegativeDelta}
    \label{prop:negative delta}
    Let $\bprofile$ be a fixed-point of $f$ and $a$ be an action. Then we have $\delta(\bprofile,a) \leq 0$.
\end{restatable}

Now, considering only local deviations at information set is not sufficient for Nash equilibria, since a player can deviate at several information sets at the same time. In the final step, we lift \Cref{prop:negative delta} to arbitrary deviations of individuals players, using backward induction. 

\AnyNashisFixPcomplete*
\section{Discussion and Open Questions}

Our results provide new insights into the representational complexity, or strategy complexity, of optimal correlated equilibria in extensive-form games, particularly when the number of players is part of the input. An important aspect of this complexity is how large the support of an equilibrium must be.

For equilibrium notions such as \EFCE, \EFCCE, \NFCCE, \AFCE, and \AFCCE, we show that the support size can be bounded polynomially without loss of expressiveness. More precisely, for any equilibrium of one of these types, there exists an equivalent equilibrium with support  consisting of a polynomial number of strategy profiles, which induces the same distribution over the leaf nodes and thus the same payoffs. This result enables NP-membership for the \textsc{Threshold} problem under these notions. Related to this, a pressing  question in line with the study of~\cite{farina2020polyPublicChance}, is the following:

\begin{itemize}
\item Which natural, tractable subclasses of games admit a polynomial-time solution to the \textsc{Threshold} problem for correlated equilibrium notions?
\end{itemize}

In stark contrast with the above notions, we showed that this support-size compressibility does not extend to \NFCE. We construct a family of games where optimal \NFCE necessarily require exponential-size support to realize their distribution over outcomes (no equivalent equilibrium exists over a small support). This distinction raises a deeper algorithmic question: 

\begin{itemize}
\item Could this exponential blowup in representation be a fundamental bottleneck to obtaining a polynomial-time algorithm for the \textsc{Any}-\NFCE problem?
\end{itemize}

Finally, the complexity landscape for \NFCE remains open in the setting of two-player games or games without change nodes. Our \pspace-hardness result for optimal \NFCE does not yet extend to these restricted cases. A natural first step would be to decide whether there are families of two-player games without chance where  optimal \NFCE requires exponential support:

\begin{itemize}
\item Can the general \textsc{Threshold} problem for \NFCE be solved in \pspace, or only in \exptime?
\item Does the \textsc{Threshold} problem for \NFCE remain \pspace-hard without chance node or with a fixed number of players?
\end{itemize}
\bibliographystyle{ACM-Reference-Format}
\bibliography{biblio}

\clearpage 

\appendix

\section*{Appendices} 
\addcontentsline{toc}{section}{Appendix}
\addtocontents{toc}{\protect\setcounter{tocdepth}{-1}}

\etocdepthtag.toc{mtappendix}
\etocsettagdepth{mtchapter}{none}
\etocsettagdepth{mtappendix}{subsection}
\tableofcontents


\section{Extended Preliminaries}

In this section, we formalize some notions required to define the different equilibria in extensive-form games with $n$ players. For readability, we may reintroduce some of the notions that were already presented in \Cref{sec:overview-extensive form games}. 

We denote by $\Node$ and $\Leaves \subseteq \Node$ the set of
nodes and leaves respectively
in our game tree. 
Given two information sets $I$ and $J$ belonging to the same player, we denote by $I \preceq J$ when a node in $J$ can be reached from some node in $I$, that is, when there exist $u \in I$ and $v \in J$ such that $u\rightarrow^* v$  with $\rightarrow^*$ being the transitive closure of the successor relation $\rightarrow$.

Moreover, we define the \emph{payoff} or \emph{utility function} $\payoff{\cdot}$, which associates to each leaf $\node \in \Leaves$ a vector in $\mathbb{Q}^n$, representing the payoffs of each player when the game reaches this leaf. Denote by $\payoffP{v}{i}$ the payoff associated to player $i$ at leaf $\node$, that is, $\payoffP{v}{i} = x_i$ when $\payoff{v} = (x_1,\ldots,x_n)$. In probability space on leaves, we use $U_i$ as the random variable corresponding to $u_i$. 

\subsection{History} 

A \emph{history of a player $i$} is a sequence of pairs of information sets $I \in \InfSet{i}$ and actions $a \in \actions{I}$. By convention, the empty sequence is denoted $\varepsilon$. We denote by $\HistP{i}$ the set of all histories of player $i$, and we define $\mathcal{H} = \HistP{1} \times \ldots \times \HistP{n}$. Most commonly, we will consider the \emph{history of a node $\node$ restricted to player $i$}, denoted by $\histP{\node}{i}$, as the sequence in $\HistP{i}$ of pairs of information sets and actions belonging to player $i$ encountered from the root to the node $\node$ in order.
In other words, if $\node_0 \xrightarrow{a_0} \node_1 \xrightarrow{a_1} \ldots \xrightarrow{a_\ell} \node$ where $j_1, \ldots, j_m$ is the increasing sequence of indices such that $\node_{j_k}$ belongs to some information set  $I_{j_k} \in \InfSet{i}$ then $\histP{\node}{i} = (I_{j_1},a_{j_1}) \cdot \ldots \cdot (I_{j_m},a_{j_m})$. When considering all the players, we denote by $\hist{\node} = (\histP{\node}{1},\ldots,\histP{\node}{n})$ 
the complete history of the node $\node$.\footnote{In some papers from the literature~\cite{DBLP:conf/aaai/FarinaBS20,farina2020polyPublicChance}, the history of a node is called the \emph{sequence} of the node, and is also sometimes directly identified by the last pair $(I,a)$ when assuming a \emph{perfect recall} game~\cite{DBLP:conf/sigecom/ZhangFCS22}.}

In this paper,  we assume that all players have \emph{perfect recall} and information sets group together nodes that are indistinguishable to the player who controls them. Accordingly, 
all nodes in the same information set must have the same history restricted to that player, that is, $\histP{\node}{i} = \histP{\node'}{i}$ for all information sets~$I \in \InfSet{i}$ with $i \in \{1, \ldots, n\}$ and $v,v' \in I$. 

Given an information set~$I \in \InfSet{i}$, 
we define $\histP{I}{i}$ as the history of player $i$ to that information set. Specifically,
$\histP{I}{i} = \histP{v}{i}$ for 
all nodes $v \in I$, as all such nodes   share the same player-$i$ history.

\subsection{Strategy Profile and Correlation Plans}

A \emph{pure strategy} of a player $i$ is a function $\sigma : \InfSet{i} \mapsto \Actions$
that assigns to each information set $I \in \InfSet{i}$ one of its available action, meaning that $\sigma(I) \in \actions{I}$. A \emph{(pure) strategy profile} is a tuple of pure strategies $\bsigma = (\sigma_1, \ldots, \sigma_n)$ where each $\sigma_i$ is a pure strategy of player $i$. We denote by $\Sigma$ and $\Sigma_i$ the sets of pure strategy profiles and pure strategies of player $i$, respectively. A \emph{correlated plan} $\CorrPlan$ is a probability distribution over the strategy profiles, and we denote by $\Delta(\Sigma)$ the set of all correlated plans.

Denote by $\SigmaP{i}$ the set of \emph{partial pure strategies of player $i$}, that is, the set strategies defined only on a subset of $\InfSet{i}$. Given a partial pure strategy~$\beta \in \SigmaP{i}$, with $i\in \{1,\ldots,n\}$,
denote by $\mathit{dom}(\beta)$  the domain of this partially defined function. Partial pure strategies are used to describe the deviation strategies of players; to this aim, 
given a pure strategy profile $\bsigma \in \Sigma$ and a partial pure strategy $\beta \in \SigmaP{i}$,  denote by $\sigma_i[\beta]$ the pure strategy in $\Sigma_i$ that plays $\beta$ when defined and $\sigma_i$ otherwise. Formally, for all $I \in \InfSet{i}$, 
\[
\sigma_i[\beta](I) = \begin{cases} \beta(I) \text{ when $I \in \mathit{dom}(\beta)$}\\
\sigma_i(I) \text{ otherwise.}
\end{cases}
\]
Write $\bsigma[\beta]$ for the pure strategy profile $(\sigma_1,\ldots,\sigma_{i-1},\sigma_i[\beta],\sigma_{i+1},\ldots, \sigma_n)$.

To define  variants of correlated equilibria and throughout our proofs, we use the following \emph{indicators}: given $\sigma \in \Sigma_i$ a pure strategy of player $i$, 
\begin{itemize}
    \item $\confPName{i}{\sigma}:\Node \rightarrow \{0,1\}$,  defined such that  
$\confP{i}{\sigma}{\node} = 1$ iff
 the path from the root to node~$\node$ is consistent with~$\sigma$.
Formally,
\[
    \confP{i}{\sigma}{\node} = 1 \mbox{ iff } \forall (I,a) \in \histP{\node}{i}, \; (I \in \mathit{dom}(\sigma) \Rightarrow \sigma(I) = a) \,.
\] 
\item $\confPName{i}{\sigma}: \HistP{i} \rightarrow \{0,1\}$, defined 
such that 
$\confP{i}{\sigma}{h} = 1$ iff the history~$h$ is consistent with~$\sigma$. 
That is, \[\confP{i}{\sigma}{h} = 1 \text{ iff } \forall (I,a) \in h ,\; (I \in \mathit{dom}(\sigma) \Rightarrow \sigma(I) = a)\,.\] 
\item $\confPName{i}{\sigma}: \bigsqcup_{i=1}^n\Sigma_i \to \{0,1\}$, defined such that $\confPName{i}{\sigma} (\alpha)=1$ iff $\sigma=\alpha$. 
\end{itemize}
We note that  $\confP{i}{\sigma}{\node} = \confP{i}{\sigma}{\histP{\node}{i}}$.
Extending to joint strategy profiles~$\bsigma$, write $\confName{\bsigma}$ for the indicator naturally defined by 
\begin{itemize}
    \item $\conf{\bsigma}{\node} = \prod_{\iota=1}^n \confP{i}{\sigma_i}{\node}$, which equals 1 iff the path from the root to node $v$ is consistent with $\bsigma$;
    \item $\conf{\bsigma}{\bh} = \prod_{\iota=1}^n \confP{i}{\sigma_i}{h_i}$, which equals 1 iff the history tuple $\bh$ is consistent with $\bsigma$;
    \item $\conf{\bsigma}{\alpha} = \sum_{\iota=1}^n \confP{i}{\sigma_i}{\alpha}$, that is, $\conf{\bsigma}{\alpha} =1$ iff $\alpha$ is one component of $\bsigma$.
\end{itemize}

Define $P_C:\Node \to [0,1]$ as the function that assigns to each node~$v\in \Node$ the product of the probabilities on the outgoing edges of the chance nodes along the path from the root to $\node$. Formally,  
$\Reach{\node} = \prod_{i=1}^k c_i$ if $c_1, \ldots, c_k$ are the  probabilities of the outgoing edges of chance nodes along  this path. If there are no chance nodes on the path ($k = 0$), we set  $\Reach{\node} = 1$. Recall that $U_i$ is the random variable of corresponding to the utility of the $i$-th player. Then 
\[
\ExpBasic{i}{\CorrPlan} = \sum_{\bsigma \in \Sigma}  \; \sum_{\node \in \Leaves} \CorrPlan(\bsigma)\, \Reach{\node}\,\payoffP{\node}{i}\,\conf{\bsigma}{\node}\,,
\]
defines the player $i$'s expected payoff under the correlation plan~$\mu$. 
By linearity of expectation, we can compute the  expected value of the objective function under the correlation plan $\CorrPlan$, which we denoted $\ExpOmega{\CorrPlan}{\omega}$.

\section{Correlated Equilibria in Extensive-Form Games}
\label{sec:equilibria}

In this paper, we study 6 notions of correlated equilibrium in extensive-form games, in addition to Nash equilibrium: namely, \NFCE, \EFCE, and \AFCE, along with their respective coarse variants—\NFCCE, \EFCCE, and \AFCCE. In the remainder of this section, we formally define these equilibria, starting with Nash equilibrium.

For convenience, we fix  the following parameters throughout this section:
a correlation plan $\CorrPlan \in \Delta(\Sigma)$, a pure strategy profile $\bsigma \in \Sigma$, and a leaf node~$v \in \Leaves$. 
    

\subsection{Normal Form (Coarse) Correlated Equilibrium.}

In \NFCE, the mediator privately recommends a complete strategy to each player before the game begins. Let $\alpha$ be such a recommended strategy to the $i$-th player, with $i\in \{1,\ldots,n\}$.

Define $\CondNFCE{\CorrPlan}{\alpha}{\bsigma}{\node}$ as the \emph{probability that the mediator draws~$\bsigma$ under $\CorrPlan$ and player $i$ reaches~$\node$, conditional on this player being recommended~$\alpha$}. Formally,  
\[
\CondNFCE{\CorrPlan}{\alpha}{\bsigma}{\node} = \conf{\bsigma}{\node}\,\Reach{\node}\,\frac{\CorrPlan(\bsigma) \, \conf{\bsigma}{\alpha} }{\displaystyle\sum_{\bsigma' \in \Sigma} \CorrPlan(\bsigma')\,\conf{\bsigma'}{\alpha}}
\]
Observe that $\CondNFCE{\CorrPlan}{\alpha}{\bsigma}{\node}$ forms a distribution over $\Sigma \times \Leaves$.
Accordingly, the expected payoff of player $i$ under~$\CorrPlan$, conditional on this player being recommended $\alpha$ is given by
\[
\ExpNFCE{i}{\CorrPlan}{\alpha} = \sum_{\bsigma \in \Sigma} \sum_{\node \in \Leaves} \CondNFCE{\CorrPlan}{\alpha}{\bsigma}{\node} \, \payoffP{\node}{i}
\] 
In what follows, we adopt analogous definitions and maintain consistent notation for expected payoffs across different conditional events, omitting repeated definitions for brevity.
To express the expected payoff when player $i$ aims to deviate by playing strategy $\beta$,
write~$\CondNFCEDeviation{\CorrPlan}{\alpha}{\beta}{\bsigma}{\node}$ as the \emph{probability that the mediator draws~$\bsigma$ under $\CorrPlan$ and player $i$ reaches~$\node$, conditioned on player $i$ receiving the recommendation $\alpha$ and deciding to deviate by playing~$\beta$}. Formally, for all $\beta \in \Sigma_i$,
\[
\CondNFCEDeviation{\CorrPlan}{\alpha}{\beta}{\bsigma}{\node} = \conf{\bsigma[\beta]}{\node} \, \Reach{\node} \, \frac{\CorrPlan(\bsigma)\,\conf{\bsigma}{\alpha}}{\displaystyle\sum_{\bsigma' \in \Sigma} \CorrPlan(\bsigma')\,\conf{\bsigma'}{\alpha}}
\]
This equation intuitively indicates that instead of considering the node that can be reached by $\bsigma$ (i.e., $\conf{\bsigma}{\node}$), we should consider the nodes that can be reached by $\bsigma[\beta]$ (i.e., $\conf{\bsigma[\beta]}{\node}$).

A correlation plan $\CorrPlan$ is an \NFCE when  for all player $i$, for all strategies $\sigma_i$, $\beta \in \Sigma_i$, 
\begin{equation}
\ExpNFCE{i}{\CorrPlan}{\sigma_i} \geq \ExpNFCEDeviation{i}{\CorrPlan}{\sigma_i}{\beta}\tag{NFCE}\label{eq:NFCE}
\end{equation}

In \NFCCE, the player must decide to deviate before receiving  its recommended strategy. This can be expressed by considering only the expected payoff over~$\mu$. Define \[\CondNFCCEDeviation{\CorrPlan}{\beta}{\bsigma}{v} = \CorrPlan(\bsigma)\,\Reach{v}\,\conf{\bsigma[\beta]}{v}\]
as the \emph{probability that the mediator draws~$\bsigma$ under $\CorrPlan$ and player $i$ reaches~$\node$, conditional on player $i$ deciding to deviate by playing $\beta$}. 
A correlation plan $\CorrPlan$ is a \NFCCE when, for all player $i$, for all strategies $\beta \in \Sigma_i$,
\begin{equation}
\ExpNFCCE{i}{\CorrPlan} \geq \ExpNFCCEDeviation{i}{\CorrPlan}{\beta}\tag{NFCCE}\label{eq:NFCCE}
\end{equation}


\subsection{Extended Form (Coarse) Correlated Equilibrium.}

In \EFCCE and \EFCE, the recommended action for a given information set is revealed to the player only when that information set is reached during play.
Recall that in \EFCE, the player can decide to deviate after it received its recommendation whereas in \EFCCE, it must decide before receiving the recommendation.

Below, we use two more indicator functions in formal definitions:
\[\confIn{I}{v}: \Node \rightarrow \{0,1\}  \qquad \qquad \text{ and} \qquad \qquad \conf{\bsigma}{I \rightarrow a} : \actions{I} \to \{0,1\}\,\]
where $\confIn{I}{v} = 1$ iff $v$ is a descendant of some node in $I$, and 
$\conf{\bsigma}{I \rightarrow a} = 1$ iff $\bsigma$ recommends action $a$ at information set reaching $I$, that is,  $\sigma_i(I) = a$ when $I \in \InfSet{i}$.

Fix the information set~$I \in \InfSet{i}$ belonging to the $i$-th player and an action~$a\in \actions{I}$ that is recommended to this player at $I$.  
Define $\CondEFCCE{\CorrPlan}{I}{\bsigma}{v}$ as the \emph{probability that the mediator draws $\bsigma$ under $\CorrPlan$ and player $i$ reaches $v$, conditional on player $i$ reaching~$I$}, that is,
\[
\CondEFCCE{\CorrPlan}{I}{\bsigma}{v} = \frac{\CorrPlan(\bsigma)\,\Reach{v} \conf{\bsigma}{v}\, \confIn{I}{v} }{\sum_{\bsigma' \in \Sigma} \sum_{v'\in I} \CorrPlan(\bsigma')\,\Reach{v'}\, \conf{\bsigma'}{v'}}
\]
holds. 
Similarly, define  $\CondEFCE{\CorrPlan}{I}{a}{\bsigma}{v}$  where the conditional event additionally includes the player receiving the recommendation to play action $a$ at information set $I$. Again, 
\begin{equation}
\label{eq:conditional corr I a}
\CondEFCE{\CorrPlan}{I}{a}{\bsigma}{v} = \frac{\CorrPlan(\bsigma)\,\Reach{v} \conf{\bsigma}{v}\, \conf{\bsigma}{I \rightarrow a}\, \confIn{I}{v}}{\sum_{\bsigma' \in \Sigma} \sum_{v'\in I} \CorrPlan(\bsigma')\,\Reach{v'} \conf{\bsigma'}{v'}\, \conf{\bsigma'}{I \rightarrow a}}\,.
\end{equation}

Since  the player may decide to deviate when reaching~$I$, its deviating strategy $\beta$ may affect the information sets reachable from $I$. To express this, let   $\Sigma^{I \preceq}_i$ be the set of partial pure strategies defined on $\{ J \in \InfSet{i} \mid I \preceq J \}$. 

Write $\CondEFCCEDeviation{\CorrPlan}{I}{\beta}{\bsigma}{v}$ as the \emph{probability that the mediator draws~$\bsigma$ under $\CorrPlan$ and player $i$ reaches~$\node$, conditional on this player reaching~$I$ and deciding to deviate by playing $\beta$}. 
Then 
\begin{equation}
\label{eq:deviation corr I}
\CondEFCCEDeviation{\CorrPlan}{I}{\beta}{\bsigma}{v} = \frac{\CorrPlan(\bsigma)\,\Reach{v} \conf{\bsigma[\beta]}{v}\, \confIn{I}{v} }{\sum_{\bsigma' \in \Sigma} \sum_{v'\in I} \CorrPlan(\bsigma')\,\Reach{v'}\, \conf{\bsigma'}{v'}}
\end{equation}
Similarly, define $\CondEFCEDeviation{\CorrPlan}{I}{a}{\beta}{\bsigma}{v}$ where the conditional event additionally includes the player receiving the recommendation to play $a$ at $I$.
Formally, for all $\beta \in \Sigma^{I \preceq}_i$, 
\begin{equation}
\label{eq:deviation corr I a}
\CondEFCEDeviation{\CorrPlan}{I}{a}{\beta}{\bsigma}{v} = \frac{\CorrPlan(\bsigma)\,\Reach{v}\, \conf{\bsigma[\beta]}{v}\, \conf{\bsigma}{I \rightarrow a}\, \confIn{I}{v}}{\sum_{\bsigma' \in \Sigma} \sum_{v'\in I} \CorrPlan(\bsigma')\,\Reach{v'} \conf{\bsigma'}{v'}\, \conf{\bsigma'}{I \rightarrow a}}
\end{equation}

A correlation plan $\CorrPlan$ is in \EFCCE when
for all players $i$, for all $I \in \mathcal{I}_i$ and for all $\beta \in \Sigma_i^{I \preceq}$,
\begin{equation}
\ExpEFCCE{i}{\CorrPlan}{I} \geq \ExpEFCCEDeviation{i}{\CorrPlan}{I}{\beta}\tag{EFCCE}\label{eq:EFCCE}
\end{equation}
Similarly, a correlation plan $\CorrPlan$ is in \EFCE when for all players $i$, for all $I \in \mathcal{I}_i, a \in A(I)$ and for all $\beta \in \Sigma_i^{I \preceq}$,
\begin{equation}
\ExpEFCE{i}{\CorrPlan}{I}{a} \geq \ExpEFCEDeviation{i}{\CorrPlan}{I}{a}{\beta}\tag{EFCE}\label{eq:EFCE}
\end{equation}


\subsection{Agent Form (Coarse) Correlated Equilibrium}

The formal definition of AFCE is in fact almost identical to EFCE except that we need to limit the choice of deviating strategies $\beta$ to only the information set that was reached. Hence $\CorrPlan$ is an \AFCE when for all players $i$, for all $I \in \InfSet{i}$, $a,a' \in \actions{I}$, defining $\beta = \{ I \mapsto a'\}$, we have
\begin{equation}
\ExpEFCE{i}{\CorrPlan}{I}{a} \geq \ExpEFCEDeviation{i}{\CorrPlan}{I}{a}{\beta}\tag{AFCE}\label{eq:AFCE}
\end{equation}
We note  that in the definition of AFCE, by defining $\beta = \{ I \mapsto a'\}$, the strategy profile $\bsigma[\beta]$ enforces that player $i$ only deviates at information set $I$ by playing $a'$ but will keeps following the recommendations for all other information sets. 
The definition of \AFCCE is analogous:  $\CorrPlan$ is an \AFCCE when for all players $i$, for all $I \in \InfSet{i}$, $a' \in \actions{I}$, defining $\beta = \{ I \mapsto a'\}$, we have
\begin{equation}
\ExpEFCCE{i}{\CorrPlan}{I} \geq \ExpEFCCEDeviation{i}{\CorrPlan}{I}{\beta}\tag{AFCCE}\label{eq:AFCCE}
\end{equation}
\section{The \textsc{Threshold} problem for \NFCE is \pspace-hard}
\label{app:MaxNFCEisPSpace-hard}

In this section, we present the full proof of our main theorem, based on the reduction described in \Cref{fig:MaxNFCEisPSpace-hard}.
\NFCEPSPACEh*

\subsection{Definition of the reduction}

Let $\Phi$ be a fully quantified boolean formula in prenex normal form:
$$\Phi = Q_1 x_1\, Q_2 x_2\, \ldots Q_n x_n\, \phi\ ,$$
where $Q_i \in \{\exists,\forall\}$ and $\phi = \bigvee_{k=1}^m t_k$ is a formula over the variables $x_1,\ldots,x_n$ in disjunctive normal form.

The game $G_\Phi$ involves $y+2$ players, where $y$ is the number of \emph{universal} variables:
\begin{itemize}
    \item a single \emph{assignment player} (Player~$\assignP$);
    \item a single \emph{formula player} (Player~$\formP$);
    \item a \emph{universal player} for each universal variable $x_i$ (Player~$\univP{i})$.
\end{itemize}

It is a tree of height four, with a single Chance node at the root and all leaves at depth $3$. The non-chance nodes can be grouped as 
a series of ``gadgets'' each associated with a variable $x_i$ and containing three controlled nodes. The Chance node leads to all the gadgets with uniform probability\footnote{Any arbitrary probability distribution would work here, so the precision of the Chance node is not an issue.}. In each gadget, the first node belongs to \Passign, who can play either $x_i$ or $\overline{x_i}$. These actions lead to each of the two other nodes (which we call $x_i$ and $\overline{x_i}$ when the gadget is clear from context). These nodes belong to another single player (either \Pform or one of the Players \univP{j}, with $j \geq i$) and lie in the same information set. After that second player takes an action, the game ends with one of the following payoffs:
\begin{itemize}
    \item $\bm{0}$: all players receive $0$ and the objective value is $0$;
    \item $\bm{1}$: all players receive $0$ and the objective value is $1$;
    \item $\forall_j^+$: player $\univP{j}$ receives $+1$, all other players receive $0$, and the objective function is $0$;
    \item $\forall_j^-$: player $\univP{j}$ receives $-1$, all other players receive $0$, and the objective function is $0$.
\end{itemize}

\textbf{(1)~Compatibility gadget $\bm{C_i}$ for variable $\bm{x_i}$.} The second player is \Pform. They have an action for each term of $\phi$. If the term $t$ they choose is compatible with the literal $\ell_i$ chosen by \Passign, \emph{i.e.} $\overline{\ell_i} \notin t$, the payoff is $\bm{1}$; otherwise, it is $\bm{0}$.
\begin{figure}[H]
\scalebox{1}{
\begin{tikzpicture}
    \gCompatibility{i}{0,0}
    \node[overlay,left=1.8 of CFxi.center,anchor=center] (CFx1) {};
    \node[overlay,right=1.8 of CFnxi.center,anchor=center] (CFx4) {};
    \draw[dashed] (CFx1) -- (CFxi) -- (CFnxi) {};
    \draw[dashed] (CFnxi) edge node [above,pos=.9] {$\scriptsize{f}$} (CFx4){};
    \node[overlay,above=.7 of Ci] (Cj){};
    \draw[dashed] (Cj) edge node [left,pos=.4] {$\scriptsize{\assignP}$} node [right,pos=.4] {$\scriptsize{A_i}$}(Ci);
    \draw[dashed] ([xshift=+2.6cm,yshift=.65cm]Ci.center) -- ([xshift=+1.8cm,yshift=.65cm]Ci.center);
    \draw[-latex] ([xshift=+1.8cm,yshift=.65cm]Ci.center) -- ([xshift=+1cm,yshift=.65cm]Ci.center) -- (Ci);
\end{tikzpicture}
}
\caption*{Compatibility gadget $C_i$}
\end{figure}

\textbf{(2)~Knowledge gadget $\bm{_jK_i}$ for Player~$\bm{\univP{j}}$ and variable~$\bm{x_i}$, where $j>i$} The second player is \Puniv{j}. They have two actions $x_i$ and $\overline{x_i}$. If they choose the same action that \Passign chose before, the payoff is $\bm{1}$. Otherwise, it is $\bm{0}$.
\begin{figure}[H]
\scalebox{1}{
\begin{tikzpicture}
    \gKnowledge{i}{j}{0,0}
        \node[overlay,above=.7 of Kixj] (Kl){};
        \draw[dashed] (Kl) edge node [left,pos=.4] {$\scriptsize{\assignP}$} node [right,pos=.4] {$\scriptsize{A_i}$}(Kixj);
        \node[overlay,below=1.5 of Kixj] (Km){};
        \draw[overlay,dashed] (Km) -- (Kixj){};
        \node[left=2.5cm of Kixj]{};
        \node[right=2.5cm of Kixj]{};
        \draw[dashed] ([xshift=+2.6cm,yshift=.65cm]Kixj.center) -- ([xshift=+1.8cm,yshift=.65cm]Kixj.center);
        \draw[-latex] ([xshift=+1.8cm,yshift=.65cm]Kixj.center) -- ([xshift=+1cm,yshift=.65cm]Kixj.center) -- (Kixj);
\end{tikzpicture}
}
\caption*{Knowledge gadget $_jK_i$}
\end{figure}

\textbf{(3)~Doubt gadget $\bm{D_i}$ for universal variable $\bm{x_i}$} The second player is \Puniv{i}. They have three actions $x_i$, $\overline{x_i}$, and $\checkmark$. If they choose $\checkmark$, the payoff is $\bm{1}$. Otherwise, if they choose the same literal as \Passign, the payoff is $\forall_i^+$; if they choose the opposite literal, the payoff is $\forall_i^-$.
\begin{figure}[H]
\scalebox{1}{
\begin{tikzpicture}
    \gUncertainty{i}{0,0}
    \node[overlay,above=.7 of Uxi] (Uj){};
    \draw[dashed] (Uj) edge node [left,pos=.4] {$\scriptsize{\assignP}$} node [right,pos=.4] {$\scriptsize{A_i}$}(Uxi);
    \node[overlay,below=1.5 of Uxi] (Uk){};
    \draw[overlay,dashed] (Uxi) -- (Uk){};
    \node[left=2.5cm of Uxi]{};
    \node[right=2.5cm of Uxi]{};
    \draw[dashed] ([xshift=+2.6cm,yshift=.65cm]Uxi.center) -- ([xshift=+1.8cm,yshift=.65cm]Uxi.center);
    \draw[-latex] ([xshift=+1.8cm,yshift=.65cm]Uxi.center) -- ([xshift=+1cm,yshift=.65cm]Uxi.center) -- (Uxi);

\end{tikzpicture}
}
\caption*{Doubt gadget $\bm{D_i}$}
\end{figure}

\subsection{Useful notions to prove the correctness of the reduction}

\paragraph{Valid and invalid pairs} Let $\mathbb{B}^{n}$ be the set of assignments of $x_1,\ldots,x_{n}$ and $T$ be the set of terms of $\phi$. A pair $(\theta, t)$ in \astree is \emph{valid} if $\theta \models t$ and \emph{invalid} if $\theta \not \models t$.

\paragraph{Explicit proof}
    An \emph{explicit proof} of $\Phi$ is a subset $P$ of $\astree$, such that:
    \begin{itemize}
        \item all pairs in $P$ are valid;
        \item if $x_i$ is a universal variable and $\ell_1,\ldots,\ell_i$ appears in $P$, then $\ell_1,\ldots,\ell_{i-1},\overline{\ell_i}$ appears in $P$;
    \end{itemize}

    An explicit proof of $\Phi$ is \emph{minimal} if:
    \begin{itemize}
        \item if $x_i$ is an existential variable and $\ell_1,\ldots,\ell_i$ appears in $P$, then $\ell_1,\ldots,\ell_{i-1},\overline{\ell_i}$ does \emph{not} appear in $P$.
    \end{itemize}

\begin{proposition}
    \label{prop:minimalproof}
    A formula $\Phi$ is true if and only if $\Phi$ has a minimal explicit proof $P$. 
\end{proposition}

\paragraph{Good strategy profile}
    Let $p = (\ell_1, \ldots, \ell_{n}, t)$ be a valid pair of $\astree$. The \emph{good strategy profile of $p$}, denoted $\profile^p$ is defined as follows:
    \begin{align*}
        \forall 1 \leq i \leq n, \profile^p(A_i) = \ell_i & \qquad \text{ strategy for Player }\assignP\\
         \left.
        \begin{aligned}
        \forall 1 \leq j \leq n \mid x_j \text{ is universal }, \forall 1 \leq i < j, \profile^p(_jK_i) & = \ell_i\\
        \forall 1 \leq j \leq n \mid x_j \text{ is universal }, \profile^p(D_j) & = \checkmark\\
        \end{aligned}
        \right\} &\qquad  \text{ strategy for Player }\univP{j}\\
        \profile^p(C)  = t & \qquad \text{ strategy for Player }\formP\\
    \end{align*}

\begin{proposition}
\label{prop:invalidprofile}
        Let $\profile$ be a strategy profile. If $\profile$ is the good strategy profile of a valid pair, then the expectation under $\profile$ of the objective function is $1$; otherwise, it is strictly less than $1$.
\end{proposition}

\begin{proof}
     It follows from the definition that the expected value for a good strategy profile is $1$ in all gadgets, while there is at least one gadget with expected value $0$ for bad strategy profiles. As, regardless of the strategy profile, the probability of reaching any of the gadgets in $\GPhi$ is strictly positive, \Cref{prop:invalidprofile} follows.
\end{proof}

\subsection{The main proof}

In the following two propositions, we will show that a QBF $\Phi$ is true if and only if there is a \NFCE in $\GPhi$ with expected objective value $1$. 

\begin{proposition}
\label{lem:trueimplieseq1}
Let $\Phi$ be a true quantified boolean formula. There is an \NFCE in $\GPhi$ with expected objective value $1$.
\end{proposition}

\begin{proof}
    Let $P = p_1,\ldots,p_{z}$ be a minimal explicit proof of $\Phi$. We define $\CorrPlan$ as the uniform correlation plan over the good strategy profiles $\profile^{p}$. It follows from~\Cref{prop:invalidprofile} that $\expect_\CorrPlan(\Omega) = 1$.

    Let us now show that $\CorrPlan$ is an \NFCE. The assignment and formula players have a payoff of $0$ in all leaves, so they have no reason to deviate. Likewise, the variable players have no reason to deviate in the Knowledge gadgets since their payoff is the same regardless of their choice. 
    
    Let us consider whether a deviation in a Doubt gadget $D_i$ could be profitable for Player~$\univP{i}$. Let $\alpha = \ell_1,\ldots,\ell_{i-1},\checkmark$ be a recommendation given to \Puniv{i} in $\CorrPlan$ and let us consider the potential deviation where \Puniv{i} plays $x_i$ in $D_i$. 
    
    As $\CorrPlan$ contains only good profiles, all the pairs where $\Puniv{i}$ gets that recommendation start with $\ell_1,\ldots,\ell_{i-1}$. It follows that the deviation is only profitable if the probability that a pair's profile starts with $\ell_1,\ldots,\ell_{i-1},x_i$ is strictly greater than the probability that they received $\ell_1,\ldots,\ell_{i-1}, \overline{x_i}$. 
    
    However, by definition of $\CorrPlan$, these two probabilities are equal, so this deviation is not profitable. A symmetric reasoning shows that the other possible deviation, to $\overline{x_i}$, is not profitable either.

    Therefore, $\CorrPlan$ is a \NFCE and Proposition~\ref{lem:trueimplieseq1} follows.
\end{proof}

\begin{proposition}
\label{lem:eq1impliestrue}
Let $\Phi$ be a quantified boolean formula and $\CorrPlan$ be a \NFCE in  $\GPhi$ with objective value $1$. Then $\Phi$ is true.
\end{proposition}

\begin{proof}
    It follows from \Cref{prop:invalidprofile} that all profiles in the support of $\CorrPlan$ are good profiles of valid pairs of $\astree$. Let $P$ be the set of pairs appearing in $\CorrPlan$. We show that it is an explicit proof of $\Phi$. Let $x_i$ be a universal variable and let $p = \ell_1,\ldots,\ell_i, \ldots,\ell_{n},t$ be a pair in $P$. 
    
    By definition of a good profile, \Puniv{i} receives the strategy $\ell_1,\ldots,\ell_{i-1}, \checkmark$ when $\CorrPlan$ selects $p$. 
    
    As $\CorrPlan$ is a \NFCE, the deviation where \Puniv{i} plays $\ell_i$ in $D_i$ must not be profitable for \Puniv{i}. It follows that the probability that $\CorrPlan$ selects a pair starting with $\ell_1,\ldots,\ell_i$ is equal to the probability that it selects a pair starting with $\ell_1,\ldots,\overline{\ell_i}$. 
    
    It follows that there is a pair $q= \ell_1,\ldots,\ell_{i-1},\overline{\ell_i},\ell'_{i+1},\ldots,\ell'_{n}$ in $P$.
    
    Therefore, $P$ is a proof of $\Phi$ and \Cref{lem:eq1impliestrue} follows.
\end{proof}

\NFCEPSPACEh*

\begin{proof}
Direct from \Cref{lem:eq1impliestrue,lem:trueimplieseq1}.
\end{proof}

\color{black}

\section{The \textsc{Threshold} Problem for \EFCE is in \NP}
\label{sec:EFCE upper bound}


\subsection{Relevant Histories}
\label{sec:relevant histories}

A key novelty of our  procedure is the definition of a family of small sets of relevant histories that are sufficient to express the different kinds of equilibrium through a system of linear constraints. 
As indicated in~\Cref{sec:Overview EFCE}, in the notion of relevant history first introduced in~\cite{von2008extensive} for two-players game, $\bh = (h_1,h_2)$ would be considered relevant when either one of them was the empty history (history of the root) or if a path from a root to some node would pass through the 
last information sets in the histories $h_1$ and $h_2$.
For $n$ player games, the notion was extended in~\cite{DBLP:conf/aaai/FarinaBS20} by considering  $(h_1,\ldots,h_n)$ pairwise-relevant tuples of histories, i.e. where for all $i\neq j$, $(h_i,h_j)$ is relevant. 

\begin{wrapfigure}[14]{l}{0.28\textwidth}
\tikzset{
triangle/.style = {regular polygon,regular polygon sides=3,draw,inner sep = 2},
circ/.style = {circle,fill=cyan!15,draw,inner sep = 2},
term/.style = {circle,draw,inner sep = 1.5,fill=black},
sq/.style = {rectangle,fill=gray!20, draw, inner sep = 4},
sq1/.style = {rectangle,fill=red!20, draw, inner sep = 4},
dashNode/.style = {edge from parent/.style={dashed,draw}},
normNode/.style = {edge from parent/.style={draw,solid}}
}

\begin{tikzpicture}
\tikzstyle{level 1}=[level distance=7mm,sibling distance =14mm]
\tikzstyle{level 2}=[level distance=7mm,sibling distance=12mm]
\tikzstyle{level 3}=[level distance=7mm,sibling distance=8mm]
\tikzstyle{level 4}=[level distance=7mm,sibling distance=7mm]


\begin{scope}[->, >=stealth]
\node (0) [circ] {\scriptsize $v_1$}
child {
  node (00) [circ] {\scriptsize $v_2$}
  child {
    node (000) [circ] {\scriptsize $v_3$}
    child {
      node (0000) [circ] {\scriptsize $v_4$}
          child[dashNode,level distance=14mm,sibling distance=12mm] {
          node (00000) [circ] {\scriptsize $v_n$}
            child[normNode,level distance=7mm,sibling distance=7mm] {
            node (000000) [term, label=below:{\scriptsize $p_{n+1}$}] {}
            edge from parent node [auto,swap,inner sep=1pt] {\scriptsize $a_n$}
            }
            child[normNode,level distance=7mm,sibling distance=7mm] {
            node (000001) [term, label=below:{\scriptsize $p_n$}] {}
            edge from parent node [auto,inner sep=1pt] {\scriptsize $b_n$}
            }
          edge from parent node [auto,swap,inner sep=1] {}
          }
          child {
          node (00001) [term, label=below:{\scriptsize $p_4$}] {}
          edge from parent node [auto,inner sep=1] {\scriptsize $b_4$}
          }
      edge from parent node [auto,swap,inner sep=1] {\scriptsize $a_3$}
    }
    child {
      node (0001) [term, label=below:{\scriptsize $p_3$}] {}
      edge from parent node [auto,inner sep=1] {\scriptsize $b_3$}
      }
    edge from parent node [auto,swap,inner sep=1] {\scriptsize $a_2$}
  }
  child {
    node (001) [term, label=below:{\scriptsize $p_2$}] {}
    edge from parent node [auto,inner sep=1] {\scriptsize $b_2$} 
  }
  edge from parent node [auto,swap,inner sep=1] {\scriptsize $a_1$}
}
child {
    node (001) [term, label=below:{\scriptsize $p_1$}] {}
    edge from parent node [auto,inner sep=1] {\scriptsize $b_1$}
} 
;
\end{scope}
\end{tikzpicture}
\end{wrapfigure}

Unfortunately, this definition does not lead to a polynomial number of relevant histories, as illustrated in the example on the left. In this game, each node $v_i$ is played by the player $i$. Let us denote by $I_i$ the information set $\{ v_i\}$. Thus, for all players $i \in \{1, \ldots, n\}$, $\HistP{i}$ is composed only three histories $\emptyset$, $(I_i,a_i)$ and $(I_i,b_i)$. However, any tuples in 
\[
\HistP{1} \times \ldots \times \HistP{n}
\] are pairwise-relevant since the path from $v_1$ to $v_n$ goes through all information sets $I_1,\ldots, I_n$. Hence, this example produces $3^n$ pairwise-relevant histories.

In this work, we restrict the notion of relevant histories by
defining two types of relevant histories, \emph{honest} and \emph{deviation} histories, which will be used to compute the various expected payoffs in the equilibrium. For instance, in Equation \eqref{eq:EFCE}, $\ExpEFCE{i}{\CorrPlan}{I}{a}$ corresponds to the average payoff of player $i$ when all players follow their recommendations. Hence, it suffices to consider the honest histories of each leaf in the game $\hist{v}=(\histP{v}{1},\ldots,\histP{v}{n})$. We denote by $\RelevantH$ the set $\{ \hist{\node} \mid \node \in \Leaves \}$ of all honest relevant histories.

The right hand side of \Cref{eq:EFCE}, that is, $\ExpEFCEDeviation{i}{\CorrPlan}{I}{a}{\beta}$ requires the tuples of histories where the player $i$, reaching some information set $I$ and being recommended some action $a \in \actions{I}$, decides to deviate by playing $b \in \actions{I}$. Hence for the players $j \neq i$, the tuples will contain the history of the leaves $v$ in  the subgame that follows after player $i$ takes action $b$. Formally, given an information set $I \in \InfSet{i}$ of a player $i \in \{1, \ldots, n\}$, and $a \in \actions{I}$, we denote by 
\[
  \replaceH{I}{a}{v} = (\histP{v}{1}, \ldots, \histP{v}{i-1}, \histP{I}{i}\cdot(I,a), \histP{v}{i+1}, \ldots \histP{v}{n})\,.
\] 
The \emph{deviation histories for \EFCE} are thus defined as follows
\[
\RelevantDEFCE = \{ \replaceH{I}{a}{v} \mid i \in \{1, \ldots, n\}, I \in \InfSet{i}, a \in \actions{I}, v \in \Leaves\}
\]
with the set $\Relevant = \RelevantH \cup \RelevantDEFCE$ being the set of all relevant histories. 
Following the definition, we directly obtain that the number of relevant histories is polynomial in the size of our input game, that is $|\Relevant| \in O(n\,|\Node|\,|\Leaves|\,|\Actions|)$. 


\subsection{Linear Constraints}
\label{sec:EFCE linear constraints}
Consider a subset $S \subseteq \Sigma$ of pure strategy profiles. We define
a system of linear constraints, denoted $\System{S}$, consists of the following types of constraints: 
\begin{itemize}
    \item \eqref{eq:corr1} and \eqref{eq:corr2} for description of correlation plan,
    \item \eqrefCC[S,\Relevant]{eq:corr3} for description of relevant histories,
    \item \Cref{eq:EFCE-best payoff} for expected payoff without deviation,
    \item \Cref{eq:EFCE-best deviation payoff,eq:EFCE-relaxation} for the best deviation payoff, 
    \item \Cref{eq:EFCE-incentive constraint} for incentive constraints requiring that the expected payoff without deviation must be greater than the best deviation, and finally
    \item \Cref{eq:threshold} for
    the expected value of objective function is greater than the threshold.
\end{itemize}


\paragraph{Correlation plan and relevant histories.}
 Given players history tuple $\bh \in \Hist$, we associate a real-valued variable $\Xvar{\bh}$ representing the probability of reaching $\bh$. (That is, the probability that the play follows a path consistent with $\bh$.)

When considering correlation plans in an equilibrium, one should in theory consider those whose support range over $\Sigma$. However, for our procedure, we will show that it is sufficient to consider correlation plan with a small size in order to determine the optimal \EFCE. Therefore, for correlation plans defined on a
subset $S \subseteq \Sigma$ of pure strategies, for a subset $H \subseteq \Hist$, the linear constraints defining the variables $\Xvar{\bh}$ with $\bh \in H$       and the variables $x_\bsigma$ with $\bsigma \in S$, are given below.
\begin{align}
0 \leq x_{\bsigma} \leq 1 &\qquad \forall \bsigma \in S\tag{C1($S$)}\label{eq:corr1}\\
\sum_{\bsigma \in S} x_{\bsigma} = 1\tag{C2($S$)}\label{eq:corr2}\\
\Xvar{\bh} = \sum_{\bsigma \in S} x_{\bsigma}\,\conf{\bsigma}{\bh} &\qquad \forall \bh \in H\tag{C3($S,H$)}\label{eq:corr3}
\end{align} 
Notice that \eqref{eq:corr1}, \eqref{eq:corr2} and \eqref{eq:corr3} have a number of variables and linear constraints bounded by $|H|+|S|+1$. These equations are generic and can be applied to all equilibria. In particular, in the case of \EFCE, we will consider the equations where $H = \Relevant$.


\paragraph{Expected payoff without deviation.}
In addition to the variables $\Xvar{\bh}$, we also introduce some variables $\UvarA{I}{a}$ to represent the 
expected payoff of player $i$ when it reaches the information set $I \in \InfSet{i}$, is recommended the action $a$ and no player deviates. The linear constraints linking the variables~$\UvarA{I}{a}$ with the variables $\Xvar{\bh}$ are given below. For all $I \in \InfSet{i}$, for all $a \in \actions{I}$,
\begin{equation}
  \UvarA{I}{a} = \sum_{v \in \Leaves} \payoffP{v}{i}\,\Reach{v}\,\Xvar{\hist{v}}\,\confIn{I,a}{v}\label{eq:EFCE-best payoff}
\end{equation}
where $\confIn{I,a}{v}$ is the  indicator for whether $(I,a)$ appears in the history of $v$ (more formally, $\confIn{I,a}{v} = 1$ iff $(I,a) \in \hist{v}$). 
\Cref{eq:EFCE-best payoff} basically describes the term $\ExpEFCE{i}{\CorrPlan}{I}{a}$ from the left-hand side of the inequalities in \eqref{eq:EFCE}.

\paragraph{Best deviation payoff.}

We introduce two additional types of variables to describe the best expected payoff for players when they decide to deviate.
The variables $\Vvar{I}{a}{I'}$ represent the best expected payoff for player $i$ upon reaching the information set $I' \in \InfSet{i}$, after having deviated at the information set $I \in \InfSet{i}$ when recommended to play the action $a \in \actions{I}$. Similarly, the variables $\VvarA{I}{a}{I'}{a'}$ represent the best expected payoff of player $i$ when reaching the information set $I' \in \InfSet{i}$, and playing $a' \in \actions{I'}$ after having deviated at information set $I \in \InfSet{i}$ when recommended to play the action $a \in \actions{I}$. 

For all information sets $I,I' \in \InfSet{i}$ with  $i \in \{1, \ldots, n\}$ and $I \preceq I'$, for all actions $a \in \actions{I}$ and~$a' \in \actions{I'}$,   
\begin{align}
\VvarA{I}{a}{I'}{a'} &= \sum_{v \in \Leaves\, \mid\, \histP{v}{i} = h} \payoffP{v}{i}\,\Reach{v}\,\Xvar{\replaceH{I}{a}{v}} + \sum_{J \in \InfSet{i}\, \mid\, \histP{J}{i} = h} \Vvar{I}{a}{J}\label{eq:EFCE-best deviation payoff}\\
\Vvar{I}{a}{I'} &\geq \VvarA{I}{a}{I'}{a'}\label{eq:EFCE-relaxation}
\end{align}
where  $h = \histP{I'}{i}\cdot (I',a')$. 
Intuitively, when computing the best payoff of a player $i$ after reaching $I'$ and playing $a'$, we can split the leaves of the remaining subgame in two, whether the player $i$ still has to play or not. For the leaves reached without any involvement of player $i$, we can directly compute the expected payoff. This is the sum 
\[\sum_{v \in \Leaves\, \mid\, \histP{v}{i} = h} \payoffP{v}{i}\,\Reach{v}\,\Xvar{\replaceH{I}{a}{v}}\]
in \Cref{eq:EFCE-best deviation payoff}. For the leaves where the player $i$ is still involved, we look at the \emph{next} information sets of that player, that is $J \in \InfSet{i}$ such that $\histP{J}{i} = h$ and add the best payoff when reaching that information set. This is the sum $\sum_{J \in \InfSet{i}\, \mid\, \histP{J}{i} = h} \Vvar{I}{a}{J}$ in \Cref{eq:EFCE-best deviation payoff}. The computation of $\Vvar{I}{a}{I'}$ would normally correspond to taking the maximum over all variables  $\VvarA{I}{a}{I'}{a'}$ with $a' \in \actions{I'}$. To keep the constraint linear, we relax this computation by introducing the inequalities $\Vvar{I}{a}{I'} \geq \VvarA{I}{a}{I'}{a'}$ for all $a' \in \actions{I'}$ (see \Cref{eq:EFCE-relaxation}). This relaxation is not problematic, as each $\Vvar{I}{a}{I'}$ is bounded by the variable $\UvarA{I}{a}$ as shown below.


\paragraph{Incentive constraints.}
The incentive constraints ensure that the expected payoff of every player without deviation is always greater than any of their best deviation payoffs, and that the expected value of the objective function is greater than the threshold $\lambda$. For all $i \in \{1, \ldots, n\}$, for all $I \in \InfSet{i}$, for all $a \in \actions{I}$,
\begin{align}
  \UvarA{I}{a} &\geq \Vvar{I}{a}{I}\label{eq:EFCE-incentive constraint}\\
  \sum_{v \in \Leaves} \omega(\payoffP{v}{1},\ldots,\payoffP{v}{n})\,\Reach{v}\,\Xvar{\hist{v}} &\geq \lambda\label{eq:threshold}
\end{align}

Our final procedure for solving the Threshold problem for EFCE consists of guessing a set $S \subseteq \Sigma$ with $|S| \leq |\Relevant| + 1$ and solving the system $\System{S}$.


\subsection{Correctness}

Some of the proposition presented here for the correctness of our procedure for \EFCE can be used to solve the \textsc{Threshold} problem for the other equilibria as well. For instance, to show that we can always guess a small support for the correlation plans, we rely on the Carathéodory theorem, stated below. Let $\mathit{conv}(A)$ denote the convex hull of $A$.

\begin{theorem}[Carathéodory]\label{thm:cara}
Let $A \subseteq \mathbb{R}^d$ be a finite set. Then for every $v \in conv(A)$, there exists some $A_v \subseteq A$ such that $| A_v | \leq d+1$ and $ v \in conv(A_v)$. 
\end{theorem}

The following proposition shows that the probability of reaching a history under any correlation plan can also be obtained under a correlation plan with small support.


\begin{proposition}
\label{prop:Cara}
Let $\CorrPlan \in \Delta(\Sigma)$ be a correlation plan. For all sets of histories $H \subseteq \Hist$, there exists a set $S \subseteq \Sigma$ of pure strategy profiles and a correlation plan $\CorrPlan' \in \Delta(\Sigma)$ such that $S = \supp(\CorrPlan')$, $|S| \leq |H| + 1$ and for all $\bh \in H$,
\[
    \sum_{\bsigma\in \Sigma} \CorrPlan(\bsigma) \conf{\bsigma}{\bh} = \sum_{\bsigma\in \Sigma} \CorrPlan'(\bsigma) \conf{\bsigma}{\bh}\,.
\]
\end{proposition}

\begin{proof}
Let $\delta_\bsigma$ denotes the Dirac distribution on the pure strategy profile $\bsigma \in \Sigma$; recall that such dirac distributions constitute a basis~$\mathcal B$ for $\mathbb{R}^{\Sigma}$.
We define the following linear map $\kappa : \mathbb{R}^{\Sigma} \rightarrow \mathbb{R}^{H}$ (with its action on the basis of $\mathbb{R}^{\Sigma}$), 
such that for all $\bh \in H$, 
\[\kappa(\delta_\bsigma)_{\bh} = \conf{\bsigma}{\bh}\,,\]
where   $\kappa(\delta_\bsigma)_{\bh}$ denotes the projection of the element $\kappa(\delta_\bsigma) \in \mathbb{R}^{H}$ onto the coordinate $\bh$. 
Fix a correlation plan~$\mu \in \Delta(\Sigma)$, and observe that the map $\kappa$ is  such that  
\[
\kappa(\CorrPlan)_{\bh} = \sum_{\bsigma \in \Sigma} \CorrPlan(\bsigma)\, \conf{\bsigma}{\bh}\,
\]
where $\bh \in H$. 
Denoting $\polyCorr = \kappa(\Delta(\Sigma))$, from the linearity of~$\kappa$ and convexity of~$\Delta(\Sigma)$ it follows that $\polyCorr$ is convex and equals to $\mathit{conv}(\kappa(\mathcal B))$. 
By \Cref{thm:cara}, since $\kappa(\CorrPlan) \in \polyCorr$, there exists $S \subseteq \Sigma$ of size at most~$|H|+1$
such that 
$\kappa(S) \subseteq \kappa(\mathcal B)$ and $\kappa(\CorrPlan) \in \mathit{conv}(\kappa(S)$.

Write $\kappa(\CorrPlan)$ as $\sum_{\bsigma\in S} \alpha_{\bsigma} \kappa(\delta_{\bsigma})$.
Define the correlation plan $\CorrPlan' \in \Delta(\Sigma)$ as 
$\CorrPlan'=\sum_{\bsigma\in S} \alpha_{\bsigma} \delta_{\bsigma}$. By construction, $\kappa(\CorrPlan) =\kappa(\CorrPlan')$, which  concludes the proof.
\end{proof}

Using the previous proposition with the set $H = \Relevant$, we can restrict attention to correlation plans with support size at most $|\Relevant|+1$. This enables an \NP-procedure:   guess such a support and  verify the existence of an \NFCE, as stated in the following main result of the section. To state the result precisely, we first recall that $\System{S}$ is the system of linear constraints consists of  \eqrefCA[S]{eq:corr1}, \eqrefCB[S]{eq:corr2}, \eqrefCC[S,\Relevant]{eq:corr3} and \Cref{eq:EFCE-best payoff,eq:EFCE-best deviation payoff,eq:EFCE-relaxation,eq:EFCE-incentive constraint,eq:threshold}. 
\begin{proposition}
\label{th:EFCE}
Given an extensive-form game with perfect recall, 
there exists a correlation plan $\CorrPlan$ in \EFCE with \[\ExpOmega{\CorrPlan}{\omega} \geq \lambda\] if and only if there exists $S \subseteq \Sigma$ with $|S| \leq |\Relevant|+1$ such that $\System{S}$ admits a solution.
\end{proposition}

To prove this proposition, we need to establish few intermediary results that relates to the linear constraints of the system. Throughout this section, we fix 
 $\theta$ to be an assignment over all variables of~$\System{S}$.

\paragraph{Correlation plan and relevant histories.} We first focus on the linear constraints defined in \Cref{sec:EFCE linear constraints}. 
When $\theta$ satisfies \eqref{eq:corr1} and \eqref{eq:corr2}, we say that it corresponds to  a correlation plan $\CorrPlan$ if $\theta(x_\bsigma) = \CorrPlan(\bsigma)$ holds.
By definition, any correlation plan $\CorrPlan$ with support~$S$ corresponds to a solution $\theta$ of \eqref{eq:corr1} and \eqref{eq:corr2} and vice-versa.
Let us now study \Cref{eq:corr3}.

\begin{restatable}{lemma}{propzvalue}
\label{prop:Zvar-best payoff}
Let $\theta$ corresponds to a correlation plan $\CorrPlan$.
It satisfies \eqref{eq:corr3} if and only if, for all history tuples~$\bh \in H$,
\begin{align}
\theta(\Xvar{\bh}) &= \sum_{\bsigma \in \Sigma}\,\CorrPlan(\bsigma) \, \conf{\bsigma}{\bh}  \tag{EqZ(H)}\label{eq:proof-correlation}
\end{align}
\end{restatable}

\begin{proof}
Observe that  $\theta(\Xvar{\bh})$ satisfies \eqref{eq:corr3} for all all tuple histories $\bh \in H$ if and only if \[\theta(\Xvar{\bh}) = \sum_{\bsigma \in S} \theta(x_\bsigma) \conf{\bsigma}{\bh}\,.\]
Since $S$ is the support of $\CorrPlan$, it follows that
\[
\theta(\Xvar{\bh}) = \sum_{\bsigma \in \Sigma} \theta(x_\bsigma) \conf{\bsigma}{\bh} = \sum_{\bsigma \in \Sigma} \CorrPlan(\bsigma) \conf{\bsigma}{\bh}\qedhere
\]
\end{proof}

\paragraph{Expected payoff without deviation.}
To express the correctness of \Cref{eq:EFCE-best payoff} computing the expected payoff without deviation, we denote by $\DenomEFCE$ the denominator of \Cref{eq:conditional corr I a}, that is:
\[
\DenomEFCE = \sum_{\bsigma' \in \Sigma} \sum_{v'\in I} \CorrPlan(\bsigma')\,\Reach{v'} \conf{\bsigma'}{v'}\, \conf{\bsigma'}{I \rightarrow a}
\]
The values of $\UvarA{I}{a}$ and $\Xvar{\bh}$ under~$\theta$ are given by the following lemma. 


\begin{restatable}{lemma}{propEFCEbestpayoff}
\label{prop:EFCE-best payoff}
Let $\CorrPlan$ be a correlation plan, and suppose that $\theta$ satisfies~\eqreftag[\RelevantH]{eq:proof-correlation-variables}. Then $\theta$ also fulfills \Cref{eq:EFCE-best payoff} if and only if, for all information sets~$I \in \InfSet{i}$ with $i\in \{1,\ldots,n\}$ and all actions $a \in \actions{I}$,
\[
\theta(\UvarA{I}{a}) = \DenomEFCE \cdot \ExpEFCE{i}{\CorrPlan}{I}{a}\,.
\]
\end{restatable}

\begin{proof}
Given  $I \in \InfSet{i}$ with $i \in \{1, \ldots, n\}$ and action~$a \in \actions{I}$, the assignment $\theta$ satisfies \eqref{eq:EFCE-best payoff} if and only if \[\theta(\UvarA{I}{a}) = \sum_{v \in \Leaves} \payoffP{v}{i}\,\Reach{v}\,\theta(\Xvar{\hist{v}})\,\confIn{I,a}{v}\,.\]
Observe that  $\conf{I,a}{v} = \conf{\bsigma}{I \rightarrow a}\, \confIn{I}{v}$ and $\conf{\bsigma}{\hist{v}} = \conf{\bsigma}{v}$ for all leaves $v$. Hence, 
\begin{align*}
\theta(\UvarA{I}{a}) &= \sum_{v \in \Leaves} \sum_{\bsigma \in \Sigma} \payoffP{v}{i}\,\Reach{v}\,\CorrPlan(\bsigma)\, \conf{\bsigma}{\hist{v}}\,\confIn{I,a}{v}\qquad \text{by \eqreftag[\RelevantH]{eq:proof-correlation-variables}} \\
&= \sum_{v \in \Leaves} \sum_{\bsigma \in \Sigma} \DenomEFCE \cdot \payoffP{v}{i}\,\CondEFCE{\CorrPlan}{I}{a}{\bsigma}{v}\\
 &= \DenomEFCE \cdot \ExpEFCE{i}{\CorrPlan}{I}{a}\qedhere
\end{align*}
\end{proof}

\paragraph{Best deviation payoff.} 
The next proposition establishes the relation between the values of the variables $\VvarA{I}{a}{I'}{a'}$, $\Vvar{I}{a}{I'}$ and the best deviation payoff. 
Given two histories $h,h'$, let us denote by $h \preceq h'$ when $h$ is a prefix of $h'$.
The proof of the following lemma relies on two technical results, \Cref{prop:generic best deviation,prop:EFCE-best deviation payoff-last}, which will also be reused in the proofs of the upper bound for other types of equilibrium. Their statement and proofs can be bound in \Cref{sec:technical lemmas}. 


Let $\CorrPlan$ be a correlation plan. Let $\theta^*$ be an assignment over all variables of~$\System{S}$ 
satisfying~ 
\eqreftag[\Relevant]{eq:proof-correlation-variables}.
Suppose that for all information sets $I,I' \in \InfSet{i}$  with $i\in \{1,\ldots,n\}$ and  $I \preceq I'$, and all $a \in \actions{I}$, for all $a' \in \actions{I'}$, the following two equations holds for $\theta^*$:
\begin{align*}
\theta^*(\VvarA{I}{a}{I'}{a'}) &= 
\max_{\beta \in \Sigma_i^{I\preceq}\,\mid\,\confP{i}{\beta}{h}=1}\, \sum_{\substack{\bsigma \in \Sigma\\v \in \Leaves \mid h \preceq \histP{v}{i}}} \DenomEFCE\,\payoffP{v}{i}\,\CondEFCEDeviation{\CorrPlan}{I}{a}{\beta}{\bsigma}{v}\\
\theta^*(\Vvar{I}{a}{I'}) &= \max_{\beta \in \Sigma_i^{I\preceq}\,\mid\,\confP{i}{\beta}{\histP{I'}{i}}=1}\, \sum_{\substack{\bsigma \in \Sigma\\v \in \Leaves \mid \confIn{I'}{v}=1}} \DenomEFCE\,\payoffP{v}{i}\,\CondEFCEDeviation{\CorrPlan}{I}{a}{\beta}{\bsigma}{v}
\end{align*}
where $h = \histP{I'}{i} \cdot (I',a')$.

\begin{restatable}{lemma}{propEFCEbestdeviationpayoff}
\label{prop:EFCE-best deviation payoff}
The following two propositions hold for $\mu$ and $\theta^*$:
\begin{enumerate}
\item $\theta^*$ 
is a solution of \Cref{eq:EFCE-best deviation payoff,eq:EFCE-relaxation};
\item $\theta \geq \theta^*$ for all solutions $\theta$ of \Cref{eq:EFCE-best deviation payoff,eq:EFCE-relaxation} and \eqreftag[\Relevant]{eq:proof-correlation-variables}, where~$\geq$
is component-wise comparison. 
\end{enumerate}
\end{restatable}

\begin{proof}
Let  $I \in \InfSet{i}$ with $i \in \{1, \ldots, n\}$ and $a \in \actions{I}$. 
Define  functions $f : \Node \rightarrow \mathbb{R}$ and $g : \Sigma_i^{I\preceq} \times \Node \rightarrow \mathbb{R}$ as follows:
\begin{align*}
f(v) &= \sum_{\bsigma \in \Sigma}\,\payoffP{v}{i}\,\Reach{v}\,\CorrPlan(\bsigma) \, \conf{\bsigma}{\replaceH{I}{a}{v}}\\
g(\beta,v) &= \sum_{\bsigma \in \Sigma} \DenomEFCE\,\payoffP{v}{i}\,\CondEFCEDeviation{\CorrPlan}{I}{a}{\beta}{\bsigma}{v}
\end{align*}

We aim to apply \Cref{prop:generic best deviation}. The first two conditions on the function on $g$ hold by its definition. The third condition holds thanks to \Cref{prop:EFCE-best deviation payoff-last}. The result is thus directly obtained by application of \Cref{prop:generic best deviation} and \eqreftag[\Relevant]{eq:proof-correlation-variables}.
\end{proof}

We can finally prove the main proposition of this section.

\begin{proof}[Proof of \Cref{th:EFCE}]
We start with the left implication. Let $\theta$ be a solution of $\System{S}$. Since $\theta$ satisfies both~\eqrefCA[S]{eq:corr1} and \eqrefCB[S]{eq:corr2}, it corresponds to a correlation plan $\CorrPlan$. Moreover with $\theta$ satisfying \eqrefCC[S,\Relevant]{eq:corr3}, we deduce from \Cref{prop:Zvar-best payoff} that $\theta$ satisfies \eqreftag[\Relevant]{eq:proof-correlation}.

We show that the correlation plan $\mu$ is an \EFCE such that $\ExpOmega{\CorrPlan}{\omega} \geq \lambda$. Let $I \in \InfSet{i}$ be an information set with  $i \in \{1, \ldots, n\}$ and $a \in \actions{I}$ be an action. 
By the fact that $\theta$ satisfies \Cref{eq:EFCE-incentive constraint} and by \Cref{prop:EFCE-best payoff}, we get that \[\DenomEFCE \cdot \ExpEFCE{i}{\CorrPlan}{I}{a} \geq \theta(\Vvar{I}{a}{I})\,.\] 

Notice that $\confP{i}{\beta}{\histP{I}{i}}=1$ for all partial strategies~$\beta \in \Sigma^{I\preceq}_i$. The following follows from $\theta$ satisfying~\Cref{eq:EFCE-best deviation payoff,eq:EFCE-relaxation} and \Cref{prop:EFCE-best deviation payoff},
\[
\DenomEFCE \cdot \ExpEFCE{i}{\CorrPlan}{I}{a} \geq \theta(\Vvar{I}{a}{I}) \geq \max_{\beta \in \Sigma_i^{I\preceq}}\, \sum_{\substack{\bsigma \in \Sigma\\v \in \Leaves \mid \confIn{I}{v}=1}} \DenomEFCE\,\payoffP{v}{i}\,\CondEFCEDeviation{\CorrPlan}{I}{a}{\beta}{\bsigma}{v}
\]
Furthermore, remark that $\CondEFCEDeviation{\CorrPlan}{I}{a}{\beta}{\bsigma}{v} \neq 0$ implies $\confIn{I}{v} = 1$. Thus, we obtain that
\[
\ExpEFCE{i}{\CorrPlan}{I}{a} \geq \max_{\beta \in \Sigma_i^{I\preceq}} \ExpEFCEDeviation{i}{\CorrPlan}{I}{a}{\beta}
\]
and so $\CorrPlan$ is an \EFCE. Finally, as $\theta$ satisfies \Cref{eq:threshold}, by \eqreftag[\Relevant]{eq:proof-correlation}, the threshold condition~$\ExpOmega{\CorrPlan}{\omega} \geq \lambda$ is guaranteed.

\medskip

Let us now prove the right implication. Consider a correlation plan $\CorrPlan$ in \EFCE such that $\ExpOmega{\CorrPlan}{\omega} \geq \lambda$. Let $\theta_\CorrPlan$ be the assignment induced by $\CorrPlan$ which satisfies \eqrefCA[\Sigma]{eq:corr1} and \eqrefCB[\Sigma]{eq:corr2}. We can extend $\theta_\CorrPlan$ such that it satisfies \eqreftag[\Relevant]{eq:proof-correlation}. Similarly, we can extend $\theta_\CorrPlan$ so that
\begin{align*}
\theta_\CorrPlan(\UvarA{I}{a}) &= \DenomEFCE \cdot \ExpEFCE{i}{\CorrPlan}{I}{a} \qquad \forall i \in \{1, \ldots, n\}, \forall I \in \InfSet{i}, \forall a \in \actions{I}\\
\theta_\CorrPlan(\VvarA{I}{a}{I'}{a'}) &= \theta^*(\VvarA{I}{a}{I'}{a'})\\
\theta_\CorrPlan(\Vvar{I}{a}{I'}) &= \theta^*(\Vvar{I}{a}{I'})
\end{align*}
with $\theta^*$ as defined at the beginning of this section.
Since
$\theta^*(\Vvar{I}{a}{I}) = \DenomEFCE \cdot \max_{\beta \in \Sigma_i^{I\preceq}} \ExpEFCEDeviation{i}{\CorrPlan}{I}{a}{\beta}$, it follows from \Cref{prop:EFCE-best deviation payoff,prop:EFCE-best payoff}  that
$\theta_{\CorrPlan}$ satisfies \Cref{eq:EFCE-best payoff,eq:EFCE-best deviation payoff,eq:EFCE-relaxation,eq:EFCE-incentive constraint,eq:threshold}. Now this does not directly conclude the proof as the support of $\CorrPlan$ currently range over $\Sigma$. However, by \Cref{prop:Cara}, there exist $S \subseteq \Sigma$ with $|S| \leq |\Relevant|+1$ and a correlation plan $\CorrPlan'$ such that $S = \supp(\CorrPlan')$ and for all $\bh \in \Relevant$,
\begin{equation}
\label{eq:proof-same Z value}
\sum_{\bsigma\in \Sigma} \CorrPlan(\bsigma) \conf{\bsigma}{\bh} = \sum_{\bsigma\in \Sigma} \CorrPlan'(\bsigma) \conf{\bsigma}{\bh}
\end{equation}
Denote by $\theta_{\CorrPlan'}$ the assignment induced by $\CorrPlan'$ which satisfies \eqrefCA[S]{eq:corr1} and \eqrefCB[S]{eq:corr2}. 
By \Cref{eq:proof-same Z value}, we can therefore extend $\theta_{\CorrPlan'}$ such that \[\theta_{\CorrPlan'}(x) = \theta_{\CorrPlan}(x)\] for all variables $x$ in the system $\System{S}$ other than the variables $x_\bsigma$.
Hence, as $\theta_{\CorrPlan'}$ satisfies \eqreftag[\Relevant]{eq:proof-correlation}, we deduce from \Cref{prop:Zvar-best payoff} that $\theta_{\mu'}$ satisfies \eqrefCC[S,\Relevant]{eq:corr3}. Since \Cref{eq:EFCE-best payoff,eq:EFCE-best deviation payoff,eq:EFCE-relaxation,eq:EFCE-incentive constraint,eq:threshold} does not depend on the variables $x_\bsigma$, and we already showed that $\theta_{\CorrPlan}$ satisfies these equations, we  conclude that so does $\theta_{\CorrPlan'}$. This implies  that $\theta_{\CorrPlan'}$ satisfies $\System{S}$ as well.
\end{proof}

\section{The \textsc{Threshold} Problem for \NFCCE, \EFCCE, \AFCE, \AFCCE is in \NP}

The procedures for solving the \textsc{Threshold} problem for \NFCCE, \EFCCE, \AFCE, \AFCCE are very similar to the one for solving \EFCE presented in \Cref{sec:EFCE upper bound}. The main differences will be in the definition of relevant deviation histories and an adaptation of the system of linear constraints. 


\subsection{The Linear Constraints for \EFCCE}

In \EFCCE, players must decide to deviate before being recommended an action. This is reflected in the definition of deviation histories: given a player $i \in \{1, \ldots, n\}$ and an information set $I \in \InfSet{i}$, we denote by 
\[
  \replaceHI{I}{v} = (\histP{v}{1}, \ldots, \histP{v}{i-1}, \histP{I}{i}, \histP{v}{i+1}, \ldots \histP{v}{n})\,.
\] 
The deviation histories for \EFCCE are thus defined as $$\RelevantDEFCCE = \{ \replaceHI{I}{v} \mid i \in \{1, \ldots, n\}, I \in \InfSet{i}, v \in \Leaves \}$$ inducing $\RelevantEFCCE = \RelevantH \cup \RelevantDEFCCE$ to be the set of relevant histories for \EFCCE.
With this, our system of linear constraints $\System{S}$ for \EFCCE will contain \eqrefCA[S]{eq:corr1}, \eqrefCB[S]{eq:corr2}, \eqrefCC[S,\RelevantEFCCE]{eq:corr3} and \eqref{eq:threshold} to describe the correlation plan, the relevant histories and the threshold constraint. 

To describe the expected payoff without deviation, we introduce a variable $\Uvar{I}$ to represent the \emph{average expected payoff of player $i$ when it reaches the information set $I \in \InfSet{i}$ and no player deviates.} The linear constraints linking $\Uvar{I}$ with the variables $\Xvar{\hist{v}}$ are given below.  For all $I \in \InfSet{i}$ with $i \in \{1, \ldots, n\}$, 
\begin{equation}
\Uvar{I} = \sum_{v \in \Leaves} \payoffP{v}{i}\,\Reach{v}\,\Xvar{\hist{v}}\,\confIn{I}{v}\label{eq:EFCCE-best payoff}
\end{equation}

The description of the best deviation introduces the variables $\VIvar{I}{I'}$ representing the \emph{best expected payoff of player $i$ conditional on reaching the information set $I' \in \InfSet{i}$ after having deviated at information set~$I \in \InfSet{i}$}. 
Similarly, we introduce $\VIvarA{I}{I'}{a'}$ where the conditional event additionally includes the player playing $a' \in \actions{I'}$.
Formally, for all information sets~$I,I' \in \InfSet{i}$ with $i \in \{1, \ldots, n\}$ and  $I \preceq I'$, for all actions $a'\in \actions{I'}$,
\begin{align}
\VIvarA{I}{I'}{a'} &= \sum_{v \in \Leaves\, \mid\, \histP{v}{i} = h} \payoffP{v}{i}\,\Reach{v}\,\Xvar{\replaceHI{I}{v}} + \sum_{J \in \InfSet{i}\, \mid\, \histP{J}{i} = h} \VIvar{I}{J}\label{eq:EFCCE-best deviation payoff}\\
\VIvar{I}{I'} &\geq \VIvarA{I}{I'}{a'}\label{eq:EFCCE-relaxation}
\end{align}
where $h = \histP{I'}{i} \cdot (I',a')$.

Finally, the incentive constraint naturally derive from $\Uvar{I}$ and $\VIvar{I}{I}$ as follows.  For all information sets~$I \in \InfSet{i}$ with $i \in \{1, \ldots, n\}$,  
\begin{equation}
\Uvar{I} \geq \VIvar{I}{I}\label{eq:EFCCE-incentive constraint}
\end{equation}

The correctness of our procedure is given in the proposition below. To state the result precisely, we define $\System{S}$ the system of linear constraints consisting of \eqrefCA[S]{eq:corr1}, \eqrefCB[S]{eq:corr2}, \eqrefCC[S,\RelevantEFCCE]{eq:corr3} and \Cref{eq:EFCCE-best payoff,eq:EFCCE-best deviation payoff,eq:EFCCE-relaxation,eq:EFCCE-incentive constraint,eq:threshold}.


\begin{restatable}{proposition}{propEFCCEcorrectness}
\label{th:EFCCE}
Given an extensive-form game with perfect recall, there exists a correlation plan $\CorrPlan$ in \EFCCE with \[\ExpOmega{\CorrPlan}{\omega} \geq \lambda\] if and only if there exists $S \subseteq \Sigma$ with $|S| \leq |\RelevantEFCCE|+1$ such that $\System{S}$ admits a solution.
\end{restatable}

\begin{proof}
We fix $\theta$ to be an assignment of all variables of $\System{S}$. For all $I \in \InfSet{i}$ with $i \in \{1, \ldots, n\}$, we define $\DenomEFCCE = \sum_{\bsigma' \in \Sigma} \sum_{v'\in I} \CorrPlan(\bsigma')\,\Reach{v'}\, \conf{\bsigma'}{v'}$.
We start the proof by proving the following two claims.


\begin{claim}
\label{prop:EFCCE-best payoff}
Let $\CorrPlan$ be a correlation plan, and suppose that $\theta$ satisfies \eqreftag[\RelevantH]{eq:proof-correlation-variables}. Then $\theta$ also fulfills \Cref{eq:EFCCE-best payoff} if and only if, for all information sets $I \in \InfSet{i}$ with $i \in \{1, \ldots, n\}$, 
\[
\theta(\Uvar{I}) = \DenomEFCCE \cdot \ExpEFCCE{i}{\CorrPlan}{I}
\]
\end{claim}

\begin{proofclaim}
Given  $I \in \InfSet{i}$ with $i \in \{1, \ldots, n\}$, the assignment $\theta$ satisfies \eqref{eq:EFCCE-best payoff} if and only if \[\theta(\Uvar{I}) = \sum_{v \in \Leaves} \payoffP{v}{i}\,\Reach{v}\,\theta(\Xvar{\hist{v}})\,\confIn{I}{v}\,.\]
Observe that $\conf{\bsigma}{\hist{v}} = \conf{\bsigma}{v}$ for all leaves $v$. Hence, 
\begin{align*}
\theta(\Uvar{I}) &= \sum_{v \in \Leaves} \sum_{\bsigma \in \Sigma} \payoffP{v}{i}\,\Reach{v}\,\CorrPlan(\bsigma)\, \conf{\bsigma}{\hist{v}}\,\confIn{I}{v}\qquad \text{by \eqreftag[\RelevantH]{eq:proof-correlation-variables}} \\
&= \sum_{v \in \Leaves} \sum_{\bsigma \in \Sigma} \DenomEFCCE \cdot \payoffP{v}{i}\,\CondEFCCE{\CorrPlan}{I}{\bsigma}{v}\\
 &= \DenomEFCCE \cdot \ExpEFCCE{i}{\CorrPlan}{I}
\end{align*}
This concludes the proof of our first claim.
\end{proofclaim}

For the computation of the best deviation payoff, take $\CorrPlan$ a correlation plan. Let $\theta^*$ be an assignment over all variables of~$\System{S}$ 
satisfying~\eqreftag[\RelevantEFCCE]{eq:proof-correlation-variables}. Suppose that for all information sets $I,I' \in \InfSet{i}$  with $i\in \{1,\ldots,n\}$ and  $I \preceq I'$, and for all $a' \in \actions{I'}$, the following two equations holds for $\theta^*$:
\begin{align*}
\theta^*(\VIvarA{I}{I'}{a'}) &= 
\max_{\beta \in \Sigma_i^{I\preceq}\,\mid\,\confP{i}{\beta}{h}=1}\, \sum_{\bsigma \in \Sigma}\,\sum_{v \in \Leaves \mid h \preceq \histP{v}{i}} \DenomEFCCE\,\payoffP{v}{i}\,\CondEFCCEDeviation{\CorrPlan}{I}{\beta}{\bsigma}{v}\\
\theta^*(\VIvar{I}{I'}) &= \max_{\beta \in \Sigma_i^{I\preceq}\,\mid\,\confP{i}{\beta}{\histP{I'}{i}}=1}\, \sum_{\bsigma \in \Sigma}\,\sum_{v \in \Leaves \mid \confIn{I'}{v}=1} \DenomEFCCE\,\payoffP{v}{i}\,\CondEFCCEDeviation{\CorrPlan}{I}{\beta}{\bsigma}{v}
\end{align*}


\begin{claim}
\label{prop:EFCCE-best deviation payoff}
The following two statements hold for $\CorrPlan$ and $\theta^*$:
\begin{enumerate}
\item $\theta^*$ is a solution of \Cref{eq:EFCCE-best deviation payoff,eq:EFCCE-relaxation};
\item $\theta \geq \theta^*$ for all solutions $\theta$ of \Cref{eq:EFCCE-best deviation payoff,eq:EFCCE-relaxation} and \eqreftag[\RelevantEFCCE]{eq:proof-correlation-variables}, where~$\geq$
is component-wise comparison. 
\end{enumerate}
\end{claim}

\begin{proofclaim}
Let $I \in \InfSet{i}$ with $i \in \{1, \ldots, n\}$. Define the functions $f : \Node \rightarrow \mathbb{R}$ and $g : \Sigma_i^{I\preceq} \times \Node \rightarrow \mathbb{R}$ as follows:
\begin{align*}
f(v) &= \sum_{\bsigma \in \Sigma}\,\payoffP{v}{i}\,\Reach{v}\,\CorrPlan(\bsigma) \, \conf{\bsigma}{\replaceHI{I}{v}}\\
g(\beta,v) &= \sum_{\bsigma \in \Sigma} \DenomEFCCE\,\payoffP{v}{i}\,\CondEFCCEDeviation{\CorrPlan}{I}{\beta}{\bsigma}{v}
\end{align*}

We aim to apply \Cref{prop:generic best deviation}. The first two conditions on the function on $g$ hold by its definition. The third condition holds thanks to \Cref{prop:EFCCE-best deviation payoff-last}. The result is thus directly obtained by application of \Cref{prop:generic best deviation} and \eqreftag[\RelevantEFCCE]{eq:proof-correlation-variables}.
\end{proofclaim}

We can finally prove the main result.
The proof is almost identical to the proof of \Cref{th:EFCE}. We will rely of course on the relevant histories $\RelevantEFCCE$ instead of $\Relevant$. Similarly, we rely on  \Cref{prop:EFCCE-best payoff,prop:EFCCE-best deviation payoff} instead of \Cref{prop:EFCE-best payoff,prop:EFCE-best deviation payoff}. Note that in the case of \EFCCE, we have that
\[
\theta_{min}(\VIvar{I}{I}) = \max_{\beta \in \Sigma_i^{I\preceq}\,\mid\,\confP{i}{\beta}{\histP{I}{i}}=1}\, \sum_{\substack{\bsigma \in \Sigma \\ v \in \Leaves \mid \confIn{I}{v}=1}} \DenomEFCCE\,\payoffP{v}{i}\,\CondEFCCEDeviation{\CorrPlan}{I}{\beta}{\bsigma}{v} = \DenomEFCCE\cdot \max_{\beta \in \Sigma_i^{I\preceq}} \ExpEFCCEDeviation{i}{\CorrPlan}{I}{\beta}
\]
as $\confP{i}{\beta}{\histP{I}{i}}=1$ for all $\beta \in \Sigma_i^{I\preceq}$.
\end{proof}


\subsection{The Linear Constraints for \NFCCE}

The linear constraints for \NFCCE can be seen as a restricted subset of those for \EFCCE. In \EFCCE, a player may deviate at any information set they reach, provided the deviation occurs before receiving a recommendation. In contrast, \NFCCE requires the player to commit to deviating at the outset of the game; before any recommendation is given and without knowing which information set they will encounter. As a result, the constraints for \NFCCE naturally correspond to the subset of \EFCCE constraints that pertain solely to the initial information set.

The set of \NFCCE relevant deviation histories for deviation focuses on the information sets $I$ with $\histP{I}{i} = \varepsilon$. In other word, we denote \[\replaceHEpsilon{i}{v} = (\histP{v}{1}, \ldots, \histP{v}{i-1}, \varepsilon, \histP{v}{i+1}, \ldots \histP{v}{n})\] and define \[\RelevantNFCCE = \RelevantH \cup \{ \replaceHEpsilon{i}{v} \,\mid\, i \in \{1, \ldots, n\}, v \in \Leaves\}\] to be the set of relevant histories for \NFCCE. Our system of linear constraints $\System{S}$ will naturally contain the equations \eqrefCA[S]{eq:corr1}, \eqrefCB[S]{eq:corr2}, \eqrefCC[S,\RelevantNFCCE]{eq:corr3} and \eqref{eq:threshold} to describe the correlation plan, the relevant histories and the threshold constraint.

To describe the \emph{expected payoff without deviation of a player $i$}, we introduce variables $\UvarP{i}$ and define the following constraints: for all  $i \in \{1, \ldots, n\}$,  
\begin{equation}
\UvarP{i} = \sum_{v \in \Leaves} \payoffP{v}{i}\,\Reach{v}\,\Xvar{\hist{v}}\label{eq:NFCCE-best payoff}
\end{equation}

The linear constraints describing the best deviation payoff introduce the variables $\VEvar{I'}$ representing the \emph{best expected payoff of player $i$ when reaching the information set $I' \in \InfSet{i}$ after having deviated at the beginning of the game}. We also consider the variables $\VEvar{h}$ with $h$ being a prefix of the honest history of some leaves. Formally, for all information sets $I' \in \InfSet{i}$ with $i \in \{1, \ldots, n\}$, for all histories~$h \in \{ h' \mid h' \preceq \histP{v}{i}, v \in \Leaves\}$, 
\begin{align}
\VEvar{h} &= \sum_{v \in \Leaves \mid \histP{v}{i} = h} \payoffP{v}{i}\,\Reach{v}\,\Xvar{\replaceHEpsilon{i}{v}} + \sum_{J \in \InfSet{i} \mid \histP{J}{i} = h} \VEvar{J}\label{eq:NFCCE-best deviation payoff}\\
\VEvar{I'} &\geq \VEvar{\histP{I'}{i}\cdot(I',a')}\label{eq:NFCCE-relaxation}
\end{align}
Finally, the incentive constraint is derived from $\UvarP{I}$ and $\VEvar{\varepsilon}$ as follows: for all $i \in \{1, \ldots, n\}$,
\begin{equation}
\UvarP{i} \geq \VEvar{\varepsilon}\label{eq:NFCCE-incentive constraint}
\end{equation}

The correctness of our procedure is given in the proposition below.  To state the result precisely, we define $\System{S}$ the system of linear constraints consisting of \eqrefCA[S]{eq:corr1}, \eqrefCB[S]{eq:corr2}, \eqrefCC[S,\RelevantNFCCE]{eq:corr3} and \Cref{eq:NFCCE-best payoff,eq:NFCCE-best deviation payoff,eq:NFCCE-relaxation,eq:NFCCE-incentive constraint,eq:threshold}.

\begin{restatable}{proposition}{propNFCCEcorrectness}
\label{th:NFCCE}
Given an extensive-form game with perfect recall, there exists a correlation plan~$\CorrPlan$ in \NFCCE with \[\ExpOmega{\CorrPlan}{\omega} \geq \lambda\] if and only if there exists $S \subseteq \Sigma$ with $|S| \leq |\RelevantNFCCE|+1$ such that $\System{S}$ admits a solution.
\end{restatable}

\begin{proof}
We fix $\theta$ to be an assignment of all variables o $\System{S}$. We start the proof by proving the following two claims.
     

\begin{claim}
\label{prop:NFCCE-best payoff}
Let $\CorrPlan$ be a correlation plan, and suppose that $\theta$ satisfies \eqreftag[\RelevantH]{eq:proof-correlation-variables} holds. Then $\theta$ fulfills \Cref{eq:NFCCE-best payoff} if and only if, for all $i \in \{1, \ldots, n\}$,
\[
\theta(\UvarP{i}) = \ExpNFCCE{i}{\CorrPlan}
\]
\end{claim}

\begin{proofclaim}
Given $i \in \{1, \ldots, n\}$, the assignment $\theta$ satisfies \eqref{eq:NFCCE-best payoff} if and only if 
\[\theta(\UvarP{i}) = \sum_{v \in \Leaves} \payoffP{v}{i}\,\Reach{v}\,\theta(\Xvar{\hist{v}})\,.\] 
Observe that $\conf{\bsigma}{\hist{v}} = \conf{\bsigma}{v}$ for all leaves $v$. Hence, 
\begin{align*}
\theta(\UvarP{i}) &= \sum_{v \in \Leaves} \sum_{\bsigma \in \Sigma} \payoffP{v}{i}\,\Reach{v}\,\CorrPlan(\bsigma)\, \conf{\bsigma}{\hist{v}}\qquad \text{by \eqreftag[\RelevantH]{eq:proof-correlation-variables}} \\
&= \ExpNFCCE{i}{\CorrPlan}
\end{align*}
This concludes the proof of our first claim.
\end{proofclaim}

For the computation of the best deviation payoff, take $\CorrPlan$ a correlation plan. Let $\theta^*$ be an assignment over all variables of~$\System{S}$ 
satisfying~\eqreftag[\RelevantNFCCE]{eq:proof-correlation-variables}. Suppose that for all information sets $I' \in \InfSet{i}$  with $i\in \{1,\ldots,n\}$ and for all $h$ prefix of $\histP{v}{i}$ for some leaf $v$, the following two equations holds for $\theta^*$:
\begin{align*}
\theta^*(\VEvar{h}) &= 
\max_{\beta \in \Sigma_i\,\mid\,\confP{i}{\beta}{h}=1}\, \sum_{\bsigma \in \Sigma}\,\sum_{v \in \Leaves \mid h \preceq \histP{v}{i}} \payoffP{v}{i}\,\CondNFCCEDeviation{\CorrPlan}{\beta}{\bsigma}{v}\\
\theta^*(\VEvar{I'}) &= \max_{\beta \in \Sigma_i\,\mid\,\confP{i}{\beta}{\histP{I'}{i}}=1}\, \sum_{\bsigma \in \Sigma}\,\sum_{v \in \Leaves \mid \confIn{I'}{v}=1} \payoffP{v}{i}\,\CondNFCCEDeviation{\CorrPlan}{\beta}{\bsigma}{v}
\end{align*}


\begin{claim}
\label{prop:NFCCE-best deviation payoff}
The following two propositions hold for $\CorrPlan$ and $\theta^*$:
\begin{enumerate}
\item $\theta^*$ is a solution of \Cref{eq:NFCCE-best deviation payoff,eq:NFCCE-relaxation};
\item $\theta \geq \theta^*$ for all solutions $\theta$ of \Cref{eq:NFCCE-best deviation payoff,eq:NFCCE-relaxation} and \eqreftag[\RelevantNFCCE]{eq:proof-correlation-variables}, where~$\geq$
is component-wise comparison. 
\end{enumerate}
\end{claim}

\begin{proofclaim}
Let $i \in \{1, \ldots, n\}$. Take $I = I_{root}$. Define the functions $f : \Node \rightarrow \mathbb{R}$ and $g : \Sigma_i^{I\preceq} \times \Node \rightarrow \mathbb{R}$ as follows:
\begin{align*}
f(v) &= \sum_{\bsigma \in \Sigma}\,\payoffP{v}{i}\,\Reach{v}\,\CorrPlan(\bsigma) \, \conf{\bsigma}{\replaceHEpsilon{i}{v}}\\
g(\beta,v) &= \sum_{\bsigma \in \Sigma} \DenomEFCCE\,\payoffP{v}{i}\,\CondEFCCEDeviation{\CorrPlan}{I}{\beta}{\bsigma}{v}
\end{align*}

We aim to apply \Cref{prop:generic best deviation}. The first two conditions on the function on $g$ hold by its definition. The third condition holds thanks to \Cref{prop:NFCCE-best deviation payoff-last}. The result is thus directly obtained by application of \Cref{prop:generic best deviation} and \eqreftag[\RelevantNFCCE]{eq:proof-correlation-variables}.
\end{proofclaim}

We can finally prove the main result.
The proof is almost identical to the proof of \Cref{th:EFCE}. We will rely of course on the relevant histories $\RelevantNFCCE$ instead of $\Relevant$. Similarly, we rely on  \Cref{prop:NFCCE-best payoff,prop:NFCCE-best deviation payoff} instead of \Cref{prop:EFCE-best payoff,prop:EFCE-best deviation payoff}. Note that in the case of \NFCCE, we have that
\[
\theta^*(\VEvar{\varepsilon}) = 
\max_{\beta \in \Sigma_i\,\mid\,\confP{i}{\beta}{\varepsilon}=1}\, \sum_{\bsigma \in \Sigma}\,\sum_{v \in \Leaves \mid h \preceq \histP{v}{i}} \payoffP{v}{i}\,\CondNFCCEDeviation{\CorrPlan}{\beta}{\bsigma}{v} = \max_{\beta \in \Sigma_i} \ExpNFCCEDeviation{i}{\CorrPlan}{\beta}
\]
as $\confP{i}{\beta}{\varepsilon}=1$ for all $\beta \in \Sigma_i$.
\end{proof}


\subsection{The Linear Constraints for \AFCE}

In \AFCE, when a player reaches the information set $I$ and being recommended $a$, it may decide to deviate at $I$ but is not allowed to deviate anymore for the remaining of the game.
This is reflected in the definition of deviation histories: given an information set $I \in \InfSet{i}$ with $i \in \{1, \ldots, n\}$, given $a \in \actions{I}$ and a leaf $v \in \Leaves$ such that $\confIn{I}{v} = 1$, we denote 
\[
  \replaceHA{I}{a}{v} = (\histP{v}{1}, \ldots, \histP{v}{i-1}, \histP{I}{i}\cdot(I,a)\cdot h', \histP{v}{i+1}, \ldots \histP{v}{n})
\] 
where $\histP{v}{i} = \histP{I}{i}\cdot (I,a') \cdot h'$ for some $a' \in \actions{I}$ and history $h'$. The set of relevant histories for \AFCE is then defined as follows:
\[
\RelevantAFCE = \RelevantH \cup \{ \replaceHA{I}{a}{v} \mid i \in \{1, \ldots, n\}, I \in \InfSet{i}, a \in \actions{I}, v \in \Leaves, \confIn{I}{v} = 1 \}
\]
Our system of linear constraints $\System{S}$ for \AFCE will contain \eqrefCA[S]{eq:corr1}, \eqrefCB[S]{eq:corr2}, \eqrefCC[S,\RelevantAFCE]{eq:corr3} and \eqref{eq:threshold} to describe the correlation plan, the relevant histories and the threshold constraint. 

To describe the expected payoff without deviation, we can reuse the linear constraints \eqref{eq:EFCE-best payoff} used in \EFCE since the decision to deviate is done at the same moment. 

The description of the best deviation is simplified compared to previous equilibrium concepts: once the player selects a deviating action, they must follow all subsequent recommendations without further opportunities to deviate. We introduce the variable $\VAvarA{I}{a}{a'}$ to describe the \emph{best expected payoff of player $i$ when reaching the information set $I \in \InfSet{i}$, being recommended $a$, and after deviating by playing $a'$.} Formally, for all players $i \in \{1, \ldots, n\}$, all information sets $I \in \InfSet{i}$ and all actions $a,a' \in \actions{I}$, 
\begin{align}
\VAvarA{I}{a}{a'} &= \sum_{v \in \Leaves\, \mid\, \histP{I}{i}\cdot(I,a') \preceq \histP{v}{i}} \payoffP{v}{i}\,\Reach{v}\,\Xvar{\replaceHA{I}{a}{v}}\label{eq:AFCE-best deviation payoff}
\end{align}

Finally, the incentive constraint naturally correspond to $\UvarA{I}{a}$ being greater than any $\VAvarA{I}{a}{a'}$, as indicated in the linear constraints below: for all $I \in \InfSet{i}$ with $i \in \{1, \ldots, n\}$, for all $a,a'\in \actions{I}$,
\begin{equation}
\UvarA{I}{a} \geq \VAvarA{I}{a}{a'} \label{eq:AFCE-incentive constraint}
\end{equation}

The correctness of our procedure is given in the proposition below. To state the result precisely, we define $\System{S}$ the system of linear constraints consisting of \eqrefCA[S]{eq:corr1}, \eqrefCB[S]{eq:corr2}, \eqrefCC[S,\RelevantAFCE]{eq:corr3} and \Cref{eq:EFCE-best payoff,eq:AFCE-best deviation payoff,eq:AFCE-incentive constraint,eq:threshold}. 


\begin{restatable}{proposition}{propAFCEcorrectness}
\label{th:AFCE}
Given an extensive-form game with perfect recall, there exists a correlation plan $\CorrPlan$ in \AFCE with \[\ExpOmega{\CorrPlan}{\omega} \geq \lambda\] if and only if there exists $S \subseteq \Sigma$ with $|S| \leq |\RelevantAFCE|+1$ such that $\System{S}$ admits a solution.
\end{restatable}

\begin{proof}
We fix $\theta$ to be an assignment of all variables of $\System{S}$. We start the proof by proving the following claim.


\begin{claim}
\label{prop:AFCE-best deviation payoff}
Let $\CorrPlan$ be a correlation plan, and suppose that $\theta$ satisfies \eqreftag[\RelevantH]{eq:proof-correlation-variables}. Then $\theta$ also fulfills \Cref{eq:AFCE-best deviation payoff} if and only if, for all information sets $I \in \InfSet{i}$ with $i \in \{1, \ldots, n\}$, for all $\forall a,a' \in \actions{I}$, 
\[
\theta(\VAvarA{I}{a}{a'}) = \DenomEFCE \cdot \ExpEFCEDeviation{i}{\CorrPlan}{I}{a}{\beta}
\]
with $\beta = \{ I \mapsto a'\}$.
\end{claim}

\begin{proofclaim}
We have
\begin{align*}
\theta(\VAvarA{I}{a}{a'}) &= \sum_{v \in \Leaves\, \mid\, \histP{I}{i}\cdot(I,a') \preceq \histP{v}{i}} \payoffP{v}{i}\,\Reach{v}\,\theta(\Xvar{\replaceHA{I}{a}{v}})\qquad \text{by \Cref{eq:AFCE-best deviation payoff}} \\
&= \sum_{v \in \Leaves}\, \sum_{\bsigma \in \Sigma}\, \payoffP{v}{i}\,\Reach{v}\,\conf{\bsigma}{\replaceHA{I}{a}{v}}\,\confIn{I,a'}{v}\qquad \text{by \eqreftag[\RelevantH]{eq:proof-correlation-variables}} 
\end{align*}
Observe that by definition, for all $v \in \Leaves$, $\conf{\bsigma}{\replaceHA{I}{a}{v}} = 1$ if and only if 
\begin{itemize}
\item for all $j \neq i$, $\confP{j}{\sigma_j}{v} = 1$; and
\item for all $(J,b) \in \histP{v}{i}$, $J \neq I$ implies $\sigma_i(J) = b$.  
\end{itemize}
Since $\beta = \{ I \rightarrow a'\}$, we also know that 
\[
\bsigma[\beta]_i = (\sigma_1,\ldots,\sigma_{i-1},\sigma_i[\beta],\sigma_{i+1},\ldots, \sigma_n)
\]
with $\sigma_i[\beta](I) = a'$ and for all $J \neq I$, $\sigma_i[\beta](J) = \sigma_i(J)$. Hence 
\[
\conf{\bsigma[\beta]_i}{v}\, \conf{\bsigma}{I \rightarrow a}\,\confIn{I}{v} = 1
\]
is equivalent to 
\begin{itemize}
\item $\sigma_i(I) = a$; and
\item $\confIn{I,a'}{v} = 1$; and
\item for all $j \neq i$, $\confP{j}{\bsigma}{v} = 1$; and 
\item for all $(J,b) \in \histP{i}{v}$, $J \neq I$ implies $\sigma_i[\beta](J) = \sigma_i(J) = b$.
\end{itemize}
This allows us to deduce that 
\[
\conf{\bsigma}{\replaceHA{I}{a}{v}}\,\confIn{I,a'}{v} = \conf{\bsigma[\beta]_i}{v}\, \conf{\bsigma}{I \rightarrow a}\,\confIn{I}{v}\,.
\]
We can therefore obtain that 
\begin{align*}
\theta(\VAvarA{I}{a}{a'}) &= \sum_{v \in \Leaves}\, \sum_{\bsigma \in \Sigma}\, \payoffP{v}{i}\,\Reach{v}\,\conf{\bsigma[\beta]_i}{v}\, \conf{\bsigma}{I \rightarrow a}\,\confIn{I}{v}\\
&= \sum_{v \in \Leaves}\, \sum_{\bsigma \in \Sigma}\, \payoffP{v}{i}\,\Reach{v}\,\CondEFCEDeviation{\CorrPlan}{I}{a}{\beta}{\bsigma}{v} \qquad \text{by \Cref{eq:deviation corr I a}}\\
&= \DenomEFCE \cdot \ExpEFCEDeviation{i}{\CorrPlan}{I}{a}{\beta}
\end{align*}
This concludes the proof of the claim.
\end{proofclaim}

We can finally prove the main result which follows the same pattern as the proof of \Cref{th:EFCE} in a simpler manner. We will rely of course on the relevant histories $\RelevantAFCE$ instead of $\Relevant$. Similarly, we rely on \Cref{prop:AFCE-best deviation payoff} instead of \Cref{prop:EFCE-best deviation payoff}. Note that the system of linear constraints in \AFCE is simpler. In particular, from \Cref{prop:AFCE-best deviation payoff}, we directly obtain that 
\[
\theta(\VAvarA{I}{a}{a'}) = \DenomEFCE \cdot \ExpEFCEDeviation{i}{\CorrPlan}{I}{a}{\beta}
\]
for all $i \in \{1, \ldots, n\}$, for all $I \in \InfSet{i}$ and all $a,a' \in \actions{I}$ with $\beta = \{ I \mapsto a'\}$. Thus an assignment $\theta_\CorrPlan$ that satisfies \eqreftag[\RelevantAFCE]{eq:proof-correlation} can be uniquely extended to an assignment that satisfies the linear constraints \Cref{eq:EFCE-best payoff,eq:AFCE-best deviation payoff,eq:AFCE-incentive constraint,eq:threshold}.
\end{proof}


\subsection{The Linear Constraints for \AFCCE}

In \AFCCE, a player may decide to deviate after reaching an information set but before receiving recommendation. Moreover, one it deviates, he is not allowed to deviate any more. This is reflected in the definition of deviation histories: given an information set $I \in \InfSet{i}$ with $i \in \{1, \ldots, n\}$ and a leaf $v \in \Leaves$ such that $\confIn{I}{v} = 1$, we denote 
\[
  \replaceHAC{I}{v} = (\histP{v}{1}, \ldots, \histP{v}{i-1}, \histP{I}{i} \cdot h', \histP{v}{i+1}, \ldots \histP{v}{n})
\] 
where $\histP{v}{i} = \histP{I}{i}\cdot (I,a) \cdot h'$ for some $a \in \actions{I}$ and history $h'$. The set of relevant histories for \AFCCE is then defined as follows:
\[
\RelevantAFCCE = \RelevantH \cup \{ \replaceHAC{I}{v} \mid i \in \{1, \ldots, n\}, I \in \InfSet{i}, v \in \Leaves, \confIn{I}{v} = 1 \}
\]
Our system of linear constraints $\System{S}$ for \AFCCE will contain \eqrefCA[S]{eq:corr1}, \eqrefCB[S]{eq:corr2}, \eqrefCC[S,\RelevantAFCCE]{eq:corr3} and \eqref{eq:threshold} to describe the correlation plan, the relevant histories and the threshold constraint. 

To describe the expected payoff without deviation, we can reuse the linear constraints \Cref{eq:EFCCE-best payoff} used in \EFCCE since the decision to deviate is done at the same moment. 

The description of the best deviation follows the simple linear constraints from \AFCE. We introduce the variable $\VACvar{I}{a}$ to describe the \emph{best expected payoff of player $i$ when reaching the information set $I \in \InfSet{i}$ and after deviating by playing $a$.} Formally, for all information sets $I \in \InfSet{i}$ with $i \in \{1, \ldots, n\}$ and all actions $a' \in \actions{I}$, 
\begin{align}
\VACvar{I}{a'} &= \sum_{v \in \Leaves\, \mid\, \histP{I}{i}\cdot(I,a') \preceq \histP{v}{i}} \payoffP{v}{i}\,\Reach{v}\,\Xvar{\replaceHAC{I}{v}}\label{eq:AFCCE-best deviation payoff}
\end{align} 

Finally, the incentive constraint naturally correspond to $\Uvar{I}$ being greater than any $\VACvar{I}{a'}$, as indicated in the linear constraint below: for all $\forall I \in \InfSet{i}$ with $i \in \{1, \ldots, n\}$, for all $a' \in \actions{I}$,
\begin{equation}
\Uvar{I} \geq \VACvar{I}{a'}\label{eq:AFCCE-incentive constraint}
\end{equation}

The correctness of our procedure is given in the theorem below. To state the result precisely, we define $\System{S}$ the system of linear constraints consisting of \eqrefCA[S]{eq:corr1}, \eqrefCB[S]{eq:corr2}, \eqrefCC[S,\RelevantAFCE]{eq:corr3} and \Cref{eq:EFCCE-best payoff,eq:AFCCE-best deviation payoff,eq:AFCCE-incentive constraint,eq:threshold}.


\begin{restatable}{proposition}{propAFCCEcorrectness}
\label{th:AFCCE}
Given an extensive-form game with perfect recall, there exists a correlation plan $\CorrPlan$ in \AFCCE with
There exists a correlation plan $\CorrPlan$ in \AFCCE with \[\ExpOmega{\CorrPlan}{\omega} \geq \lambda\] if and only if there exists $S \subseteq \Sigma$ with $|S| \leq |\RelevantAFCCE|+1$ such that $\System{S}$ admits a solution.
\end{restatable}

\begin{proof}
We fix $\theta$ to be an assignment of all variables of $\System{S}$. We start the proof by proving the following claim.


\begin{claim}
\label{prop:AFCCE-best deviation payoff}
Let $\CorrPlan$ be a correlation plan, and suppose that $\theta$ satisfies \eqreftag[\RelevantH]{eq:proof-correlation-variables}. Then $\theta$ also fulfills \Cref{eq:AFCCE-best deviation payoff} if and only if, for all information sets $I \in \InfSet{i}$ with $i \in \{1, \ldots, n\}$, for all $\forall a' \in \actions{I}$, 
\[
\theta(\VACvar{I}{a'}) = \DenomEFCCE \cdot \ExpEFCCEDeviation{i}{\CorrPlan}{I}{\beta}
\]
with $\beta = \{ I \mapsto a'\}$.
\end{claim}

\begin{proofclaim}
We have
\begin{align*}
\theta(\VACvar{I}{a'}) &= \sum_{v \in \Leaves\, \mid\, \histP{I}{i}\cdot(I,a') \preceq \histP{v}{i}} \payoffP{v}{i}\,\Reach{v}\,\theta(\Xvar{\replaceHAC{I}{v}})\qquad \text{by \Cref{eq:AFCCE-best deviation payoff}} \\
&= \sum_{v \in \Leaves}\, \sum_{\bsigma \in \Sigma}\, \payoffP{v}{i}\,\Reach{v}\,\conf{\bsigma}{\replaceHAC{I}{v}}\,\confIn{I,a'}{v}\qquad \text{by \eqreftag[\RelevantH]{eq:proof-correlation-variables}} 
\end{align*}
Observe that by definition, for all $v \in \Leaves$, $\conf{\bsigma}{\replaceHAC{I}{v}} = 1$ if and only if
\begin{itemize}
\item for all $j \neq i$, $\confP{j}{\sigma_j}{v} = 1$; and
\item for all $(J,b) \in \histP{v}{i}$, $J \neq I$ implies $\sigma_i(J) = b$.
\end{itemize}
Since $\beta = \{ I \rightarrow a'\}$, we also know that 
\[
\bsigma[\beta]_i = (\sigma_1,\ldots,\sigma_{i-1},\sigma_i[\beta],\sigma_{i+1},\ldots, \sigma_n)
\]
with $\sigma_i[\beta](I) = a'$ and for all $J \neq I$, $\sigma_i[\beta](J) = \sigma_i(J)$. Hence 
\[
\conf{\bsigma[\beta]_i}{v}\, \confIn{I}{v} = 1
\]
is equivalent to 
\begin{itemize}
\item $\confIn{I,a'}{v} = 1$; and
\item for all $j \neq i$, $\confP{j}{\bsigma}{v} = 1$; and
\item for all $(J,b) \in \histP{i}{v}$, $J \neq I$ implies $\sigma_i[\beta](J) = \sigma_i(J) = b$.
\end{itemize}
This allows us to deduce that 
\[
\conf{\bsigma}{\replaceHAC{I}{v}}\,\confIn{I,a'}{v} = \conf{\bsigma[\beta]_i}{v}\,\confIn{I}{v}\,.
\]
We can therefore obtain that 
\begin{align*}
\theta(\VACvar{I}{a'}) &= \sum_{v \in \Leaves}\, \sum_{\bsigma \in \Sigma}\, \payoffP{v}{i}\,\Reach{v}\,\conf{\bsigma[\beta]_i}{v}\, \confIn{I}{v}\\
&= \sum_{v \in \Leaves}\, \sum_{\bsigma \in \Sigma}\, \payoffP{v}{i}\,\Reach{v}\,\CondEFCCEDeviation{\CorrPlan}{I}{\beta}{\bsigma}{v} \qquad \text{by \Cref{eq:deviation corr I}}\\
&= \DenomEFCCE \cdot \ExpEFCCEDeviation{i}{\CorrPlan}{I}{\beta}
\end{align*}
This concludes the proof of the claim.
\end{proofclaim}

We can finally prove the main result which follows the same pattern as the proof of \Cref{th:EFCE} in a simpler manner. We will rely of course on the relevant histories $\RelevantAFCCE$ instead of $\Relevant$. Similarly, we rely on \Cref{prop:AFCCE-best deviation payoff,prop:EFCCE-best payoff} instead of \Cref{prop:EFCE-best deviation payoff,prop:EFCE-best payoff}. Note that the system of linear constraints in \AFCCE is simpler (as it was the case for \AFCE). In particular, from \Cref{prop:AFCCE-best deviation payoff}, we directly obtain that 
\[
\theta(\VACvar{I}{a'}) = \DenomEFCCE \cdot \ExpEFCCEDeviation{i}{\CorrPlan}{I}{\beta}
\]
for all $i \in \{1, \ldots, n\}$, for all $I \in \InfSet{i}$ and all $a' \in \actions{I}$ with $\beta = \{ I \mapsto a'\}$. Thus an assignment $\theta_\CorrPlan$ that satisfies \eqreftag[\RelevantAFCCE]{eq:proof-correlation} can be uniquely extended to an assignment that satisfies the linear constraints \Cref{eq:EFCCE-best payoff,eq:AFCCE-best deviation payoff,eq:AFCCE-incentive constraint,eq:threshold}.
\end{proof}


\section{Technical Lemmas}
\label{sec:technical lemmas}

The next lemma focuses on the computation of the best deviation payoff in \EFCE, and in particular on the first sum in \Cref{eq:EFCE-best deviation payoff}. As previously mentioned in \Cref{sec:EFCE linear constraints}, this sum corresponds to the payoff obtained from the leaves where the player $i$ has nothing more to play after reaching $I'$ and playing $a'$. This is thus valid for all  deviation strategies $\beta \in \Sigma_i^{I\prec}$ such that \[\confP{i}{\beta}{\histP{I'}{i} \cdot (I',a')} = 1\,.\]

\begin{restatable}{lemma}{propEFCEbestdeviationpayofflast}[EFCE - Best deviation payoff computation]
\label{prop:EFCE-best deviation payoff-last}
Let $\CorrPlan$ be a correlation plan, and suppose that $\theta$  satisfies\eqreftag[\Relevant]{eq:proof-correlation-variables}.
Let   $I,I' \in \InfSet{i}$  be information sets of the $i$-th player such that $I \preceq I'$. Let  $a \in \actions{I}$ and $a' \in \actions{I'}$ be actions, and 
  $\beta \in \Sigma_i^{I\preceq}$ be a partial strategy. Write
 $h = \histP{I'}{i} \cdot (I',a')$.  If   $\confP{i}{\beta}{h} = 1$, then
\[
\sum_{\substack{\bsigma \in \Sigma\\v \in \Leaves \mid \histP{v}{i} = h}}\, \payoffP{v}{i}\,\Reach{v}\,\CorrPlan(\bsigma)\,\conf{\bsigma}{\replaceH{I}{a}{v}} = 
\sum_{\substack{\bsigma \in \Sigma\\v \in \Leaves \mid \histP{v}{i} = h}}
\DenomEFCE\,\payoffP{v}{i}\,\CondEFCEDeviation{\CorrPlan}{I}{a}{\beta}{\bsigma}{v}
\]
\end{restatable}

\begin{proof}
Take $\bsigma \in \Sigma$ and $v \in \Leaves$ such that $\histP{v}{i} = h$.
By definition, 
\[\DenomEFCE\,\CondEFCEDeviation{\CorrPlan}{I}{a}{\beta}{\bsigma}{v} = \CorrPlan(\bsigma)\,\Reach{v} \conf{\bsigma[\beta]}{v} \confIn{I}{v} \conf{\bsigma}{I \rightarrow a}\,.
\]
Moreover by definition of $\replaceH{I}{a}{v}$ and by \Cref{eq:proof-correlation}, we know that 
\[
\theta(\Xvar{\replaceH{I}{a}{v}}) = \sum_{\bsigma \in \Sigma} \CorrPlan(\bsigma)\,\confP{i}{\sigma_i}{\histP{I}{i} \cdot(I,a)}\, \prod_{j\neq i} \confP{j}{\sigma_j}{v}
\]

Thus to prove the equality, it suffices to show that
\begin{equation}
  \confP{i}{\sigma_i}{\histP{I}{i} \cdot(I,a)} \prod_{j\neq i} \confP{j}{\sigma_j}{v} = \conf{\bsigma[\beta]}{v} \confIn{I}{v} \conf{\bsigma}{I \rightarrow a}
  \label{eq:same indicator}
\end{equation}
Since we assumed $\histP{v}{i} = h = \histP{I'}{i}\cdot(I',a')$, we directly have that 
\[\confP{i}{\sigma_i[\beta]}{\histP{I'}{i}\cdot(I',a')} = \confP{i}{\sigma_i[\beta]_i}{v}\,.\]
Moreover, as $I \preceq I'$, we also have $\confIn{I}{v} = 1$. 
Recall that by definition, \[\conf{\bsigma[\beta]}{v} = \confP{i}{\sigma_i[\beta]}{v}\prod_{j\neq i} \confP{j}{\sigma_i}{v}\,.\] Therefore, \Cref{eq:same indicator} is equivalent to 
\[
  \confP{i}{\sigma_i}{\histP{I}{i}\cdot(I,a)} = \confP{i}{\sigma_i[\beta]}{v} \conf{\bsigma}{I \rightarrow a}
\]
By definition $\confP{i}{\sigma_i}{\histP{I}{i}\cdot(I,a)} = \confP{i}{\sigma_i}{\histP{I}{i}}\, \conf{\bsigma}{I \rightarrow a}$. Moreover, as $\beta \in \Sigma_i^{I \preceq}$, we deduce that \[\confP{i}{\sigma_i[\beta]}{v} = \confP{i}{\sigma_i}{\histP{I}{i}}\, \confP{i}{\beta}{h}\,.\] As $\confP{i}{\beta}{h} = 1$ by hypothesis, we conclude that \Cref{eq:same indicator} holds.
\end{proof}


\begin{lemma}[EFCCE - Best deviation payoff computation]
\label{prop:EFCCE-best deviation payoff-last}
Let $\CorrPlan$ be a correlation plan and let $\theta$ be an assignment such that \eqreftag[\RelevantEFCCE]{eq:proof-correlation-variables} holds. For all $i \in \{1, \ldots, n\}$, for all $I,I' \in \InfSet{i}$, for all $a' \in \actions{I'}$, for all $\beta \in \Sigma_i^{I\preceq}$, denoting $h = \histP{I'}{i} \cdot (I',a')$, if $I \preceq I'$ and $\confP{i}{\beta}{h} = 1$ then
\[
\sum_{v \in \Leaves \mid \histP{v}{i} = h}\, \sum_{\bsigma \in \Sigma}\,\payoffP{v}{i}\,\Reach{v}\,\CorrPlan(\bsigma)\,\conf{\bsigma}{\replaceHI{I}{v}} = \sum_{\bsigma \in \Sigma}\,\sum_{v \in \Leaves \mid \histP{v}{i} = h} \DenomEFCCE\,\payoffP{v}{i}\,\CondEFCCEDeviation{\CorrPlan}{I}{\beta}{\bsigma}{v}
\]
\end{lemma}

\begin{proof}
Take $\bsigma \in \Sigma$ and $v \in \Leaves$ such that $\histP{v}{i} = h$.
By definition, 
\[\DenomEFCCE\,\CondEFCCEDeviation{\CorrPlan}{I}{\beta}{\bsigma}{v} = \CorrPlan(\bsigma)\,\Reach{v} \conf{\bsigma[\beta]}{v} \confIn{I}{v}\,.
\]
Moreover by definition of $\replaceHI{I}{v}$ and by \eqreftag[\RelevantEFCCE]{eq:proof-correlation-variables}, we know that 
\[
\theta(\Xvar{\replaceHI{I}{v}}) = \sum_{\bsigma \in \Sigma} \CorrPlan(\bsigma)\,\confP{i}{\sigma_i}{\histP{I}{i}}\, \prod_{j\neq i} \confP{j}{\sigma_j}{v}
\]

Thus to prove the equality, it suffices to show that
\begin{equation}
  \confP{i}{\sigma_i}{\histP{I}{i}} \prod_{j\neq i} \confP{j}{\sigma_j}{v} = \conf{\bsigma[\beta]}{v} \confIn{I}{v}
  \label{eq:same indicator EFCCE}
\end{equation}
Since we assumed $\histP{v}{i} = h = \histP{I'}{i}\cdot(I',a')$, we directly have that $\confP{i}{\sigma_i[\beta]}{\histP{I'}{i}\cdot(I',a')} = \confP{i}{\sigma_i[\beta]}{v}$. Moreover, as $I \preceq I'$, we also have $\confIn{I}{v} = 1$. 
Recall that by definition, $\conf{\bsigma[\beta]}{v} = \confP{i}{\sigma_i[\beta]}{v}\prod_{j\neq i} \confP{j}{\sigma_i}{v}$. Therefore, \Cref{eq:same indicator EFCCE} is implied by
\[
  \confP{i}{\sigma_i}{\histP{I}{i}} = \confP{i}{\sigma_i[\beta]}{v}
\]
As $\beta \in \Sigma_i^{I \preceq}$, we deduce that $\confP{i}{\sigma_i[\beta]}{v} = \confP{i}{\sigma_i}{\histP{I}{i}}\, \confP{i}{\beta}{h}$. As $\confP{i}{\beta}{h} = 1$ by hypothesis, we conclude that \Cref{eq:same indicator EFCCE} holds.
\end{proof}


\begin{lemma}[NFCCE - Best deviation payoff computation]
\label{prop:NFCCE-best deviation payoff-last}
Let $\CorrPlan$ be a correlation plan and let $\theta$ be an assignment such that \eqreftag[\RelevantNFCCE]{eq:proof-correlation-variables} holds. For all $i \in \{1, \ldots, n\}$, for all $I' \in \InfSet{i}$, for all $a' \in \actions{I'}$, for all $\beta \in \Sigma_i$, denoting $h = \histP{I'}{i} \cdot (I',a')$, if $\confP{i}{\beta}{h} = 1$ then
\[
\sum_{v \in \Leaves \mid \histP{v}{i} = h}\, \sum_{\bsigma \in \Sigma}\, \payoffP{v}{i}\,\Reach{v}\,\CorrPlan(\bsigma)\,\conf{\bsigma}{\replaceHEpsilon{i}{v}} = \sum_{\bsigma \in \Sigma}\,\sum_{v \in \Leaves \mid \histP{v}{i} = h} \payoffP{v}{i}\,\CondNFCCEDeviation{\CorrPlan}{\beta}{\bsigma}{v}
\]
\end{lemma}

\begin{proof}
Take $\bsigma \in \Sigma$ and $v \in \Leaves$ such that $\histP{v}{i} = h$.
By definition, 
\[\CondNFCCEDeviation{\CorrPlan}{\beta}{\bsigma}{v} = \CorrPlan(\bsigma)\,\Reach{v}\,\conf{\bsigma[\beta]}{v}\,.
\]
Moreover by definition of $\replaceHEpsilon{i}{v}$ and by \eqreftag[\RelevantNFCCE]{eq:proof-correlation-variables}, we know that 
\[
\theta(\Xvar{\replaceHEpsilon{i}{v}}) = \sum_{\bsigma \in \Sigma} \CorrPlan(\bsigma)\, \prod_{j\neq i} \confP{j}{\sigma_j}{v}
\]

Thus to prove the equality, it suffices to show that
\begin{equation}
  \prod_{j\neq i} \confP{j}{\sigma_j}{v} = \conf{\bsigma[\beta]}{v}
  \label{eq:same indicator NFCCE}
\end{equation}
Since we assumed $\histP{v}{i} = h$, we directly have that $\confP{i}{\sigma_i[\beta]}{h} = \confP{i}{\sigma_i[\beta]}{v}$. Moreover, as $\beta \in \Sigma$, we have $\sigma_i[\beta] = \beta$. By hypothesis, $\confP{i}{\beta}{h} = 1$ hence $\confP{i}{\sigma_i[\beta]}{v} = 1$. Finally, 
recall that by definition, $\conf{\bsigma[\beta]}{v} = \confP{i}{\sigma_i[\beta]}{v}\prod_{j\neq i} \confP{j}{\sigma_i}{v}$. Thus, we deduce that \Cref{eq:same indicator NFCCE} holds.
\end{proof}

The next proposition is the main proposition that will allow us to handle, at the same time, the systems of linear constraints for computing the best deviation payoff for the three equilibria \NFCCE, \EFCE and \EFCCE. 
Denote $I_{root}$ to be a virtual information belonging to all players, basically representing the beginning of the game. Hence, for all $i \in \{1, \ldots, n\}$ and $\histP{I_{root}}{i} = \varepsilon$ and all $I \in \InfSet{i}$, we will assume that $I_{root} \prec I$. 

\begin{lemma}[Generic best deviation payoff computation]
\label{prop:generic best deviation}
Consider a player $i \in \{1, \ldots, n\}$, an information set $I \in \InfSet{i} \cup \{ I_{root} \}$. Let 
\[H = \begin{cases}
    \{ \histP{v}{i} \mid v \in \Node, \histP{I}{i} \prec \histP{v}{i}\} & \text{if } I \in \InfSet{i}\\
    \{ \histP{v}{i} \mid v \in \Node\} & \text{otherwise } (\text{when } I = I_{root} )
\end{cases}\]
Let  functions $f : \Node \rightarrow \mathbb{R}$ and $g : \Sigma_i^{I\preceq} \times \Node \rightarrow \mathbb{R}$ be such that:
\begin{itemize}
\item for all $I' \in \InfSet{i}$,  all $a' \in \actions{I'}$ and $v\in \Leaves$, if $I \preceq I'$,  moreover if $\confP{i}{\beta}{\histP{I'}{i}\cdot(I',a')} \cdot \conf{I'}{v} = 1$ and $g(\beta,v) \neq 0$, then \[\histP{I'}{i}\cdot(I',a') \preceq \histP{v}{i}\]
\item for all $\beta,\beta' \in \Sigma_i^{I\preceq}$ and all $v \in \Leaves$, if $\confP{i}{\beta}{v} = \confP{i}{\beta'}{v}$ the \[g(\beta,v) = g(\beta',v)\]
\item for all $h \in H$ and  all $\beta \in \Sigma_i^{I \preceq}$, if $\confP{i}{\beta}{h} = 1$ then
\begin{equation}
\label{eq:max F and G}
\sum_{v \in \Leaves \mid \histP{v}{i} = h} f(v) = \sum_{v \in L \mid \histP{v}{i} = h} g(\beta,v)
\end{equation}
\end{itemize} 
Moreover, consider the follow linear constraints:
\begin{align}
\bestpayoff{i}{h} &= \sum_{v \in \Leaves \mid \histP{v}{i} = h} f(v) 
+ \sum_{J \in \InfSet{i} \mid \histP{J}{i} = h} \bestpayoff{i}{J} & \forall h \in H\label{eq:generic-bestpayoff-h}\\
\bestpayoff{i}{I'} &\geq  \bestpayoff{i}{\histP{I'}{i}\cdot(I',a')} & \forall I' \in \InfSet{i} \text{ s.t. } I \preceq I', \forall a' \in \actions{I'}\label{eq:generic-bestpayoff-I}
\end{align}
Let $\theta^*$ the assignment defined as:
\begin{align*}
\theta^*(\bestpayoff{i}{h}) &=  \max_{\beta \in \Sigma_i^{I\preceq} \mid \confP{i}{\beta}{h}=1} \, \sum_{v \in L\,\mid\,h \preceq \histP{v}{i}} g(\beta,v) &\forall h \in H
\\
\theta^*(\bestpayoff{i}{I'}) &=  \max_{\beta \in \Sigma_i^{I\preceq}\mid \confP{i}{\beta}{\histP{I'}{i}}=1} \, \sum_{v \in L\,\mid\,\conf{I'}{v}=1} g(\beta,v)& \forall I' \in \InfSet{i} \text{ s.t. } I \preceq I'
\end{align*}

The following two propositions hold:
\begin{enumerate}
\item $\theta^*$ is a solution of \Cref{eq:generic-bestpayoff-h,eq:generic-bestpayoff-I};
\item $\theta \geq \theta^*$ for all solutions $\theta$ of \Cref{eq:generic-bestpayoff-h,eq:generic-bestpayoff-I}, where~$\geq$ is component-wise comparison. 
\end{enumerate}

\end{lemma}

\begin{proof}
Let $h \in H$. Let $\beta \in \Sigma_i^{\preceq I}$. Let $J_1, \ldots, J_m$ be the information sets in $\InfSet{i}$ such that $\histP{J_k}{i} = h$ for all $k \in \{1, \ldots, m\}$. We have:
\begin{equation}
\label{eq:proof-split g beta}
\sum_{v \in \Leaves\, \mid\, h \preceq \histP{v}{i}} g(\beta,v)
= \sum_{v \in \Leaves\,|\, \histP{v}{i} = h} g(\beta,v) + \sum_{k=1}^m \,\sum_{v \in \Leaves\,| \confIn{J_k}{v} = 1}  g(\beta,v)
\end{equation}
By definition of $\theta^*$, take $\beta_1, \ldots, \beta_m \in \Sigma_i^{I\preceq}$ such that for all $k \in \{1, \ldots, m\}$, $ \confP{i}{\beta}{\histP{J_k}{i}}=1$ and
\[
\theta^*(\bestpayoff{i}{J_k}) = \sum_{v \in \Leaves\,| \confIn{J_k}{v} = 1} g(\beta_k,v)
\]
Due to the perfect recall assumption, we know that for all $j \neq k$, 
\[\{ J \in \InfSet{i} \mid J_k \preceq J\} \cap\{ J \in \InfSet{i} \mid J_j \preceq J\} = \emptyset \,.\] Hence, from the strategy $\beta_1, \ldots, \beta_m$, we build the strategy $\beta_{max}$ such that for all $J \in \InfSet{i}$, 
\begin{itemize}
\item for all $k \in \{1, \ldots, m\}$, if $J_k  \preceq J$ then $\beta_{max}(J) = \beta_k(J)$
\item if $I \preceq J$ and $J_k \not\preceq J$  for all $k \in \{1, \ldots, m\}$, then $\beta_{max}(J) = \beta_1(J)$
\end{itemize}
As all $\beta_k$, and in particular $\beta_1$, satisfy $\confP{i}{\beta_1}{h}=1$, we have by construction that $\confP{i}{\beta_{max}}{h}=1$.

Moreover, by hypothesis, we know that for all $\beta',\beta'' \in \Sigma_i^{I\preceq}$, for all $v \in \Node$,
if the condition 
$\confP{i}{\beta'}{v} = \confP{i}{\beta''}{v}$ holds we will have $g(\beta',v) = g(\beta'',v)$. Hence, for all $k \in \{1, \ldots, m\}$, for all $v \in \Leaves$ such that $\confIn{J_k}{v} = 1$, we have $g(\beta_k,v) = g(\beta_{max},v)$. Therefore, for all $k \in \{1, \ldots, m\}$,
\[
\sum_{v \in \Leaves \mid \confIn{J_k}{v} = 1} g(\beta_{max},v)
= \theta^*(\bestpayoff{i}{J_k})
\]
Therefore, with \Cref{eq:proof-split g beta,eq:max F and G}, we obtain:
\[
\sum_{v \in \Leaves\, \mid\, h \preceq \histP{v}{i}} g(\beta_{\max},v) = \sum_{v \in \Leaves\, \mid\, \histP{v}{i} = h} f(v) + \sum_{k=1}^m \, \theta^*(\bestpayoff{i}{J_k})
\]
Now, by definition of $\theta^*$, take $\beta_{max}' \in \Sigma_i^{I\preceq}$ such that $\confP{i}{\beta_{max}'}{h}=1$ and 
\[
\theta^*(\bestpayoff{i}{h}) = \sum_{v \in \Leaves \mid h \preceq \histP{v}{i}} g(\beta_{max}',v)\,.
\]
We directly have that $\theta^*(\bestpayoff{i}{h}) \geq \sum_{v \in \Leaves \mid h \preceq \histP{v}{i}} g(\beta_{\max},v)$. Moreover, as $h = \histP{J_k}{i}$ for all $k \in \{1, \ldots, m\}$ and $\confP{i}{\beta_{max}'}{h}=1$, we deduce that
\begin{align*}
 \theta^*(\bestpayoff{i}{h}) & = \sum_{v \in \Leaves \mid \histP{v}{i} = h} f(v) + \sum_{k=1}^m \,\sum_{v \in \Leaves\,| \confIn{J_k}{v} = 1}  g(\beta_{max}',v)\\
 &\leq \sum_{v \in \Leaves \mid \histP{v}{i} = h} f(v) + \sum_{k=1}^m \,\max_{\beta \in \Sigma_i^{I\preceq}\,\mid\,\confP{i}{\beta}{\histP{J_k}{i}}=1}\, \sum_{v \in \Leaves\,| \confIn{J_k}{v} = 1}  g(\beta,v)\\
  &\leq \sum_{v \in \Leaves \mid \histP{v}{i} = h} f(v) + \sum_{k=1}^m \,\theta^*(\bestpayoff{i}{J_k})\\
 &\leq \sum_{v \in \Leaves \mid h \preceq \histP{v}{i}} g(\beta_{\max},v)
\end{align*}
This allows us to deduce that:
\begin{equation}
\theta^*(\bestpayoff{i}{h}) = \sum_{v \in \Leaves \mid \histP{v}{i} = h} f(v) + \sum_{k=1}^m \,\theta^*(\bestpayoff{i}{J_k})
\label{eq:proof-best deviation payoff-sum}
\end{equation}
This equation allow us to show that $\theta^*$ satisfies the linear constraints \eqref{eq:generic-bestpayoff-h}.

\medskip

Let us now focus on the linear constraint \eqref{eq:generic-bestpayoff-I}. Let $I' \in \InfSet{i}$ such that $I \preceq I'$ and let $a' \in \actions{I'}$.
 We start by noticing that for all $\beta \in \Sigma_i^{I \preceq}$, if $\confP{i}{\beta}{\histP{I'}{i}}=1$ then there exists $a' \in \actions{I'}$ such that $\confP{i}{\beta}{\histP{I'}{i} \cdot (I',a')}=1$. Therefore, 
\[
\theta^*(\bestpayoff{i}{I'})
=
\max_{a' \in \actions{I'}} \max_{\beta \in \Sigma_i^{I\preceq}\,\mid\,\confP{i}{\beta}{\histP{I'}{i} \cdot (I',a')}=1}\, \sum_{v \in \Leaves \mid \confIn{I'}{v}=1} g(\beta,v)
\]
However, when $\confP{i}{\beta}{\histP{I'}{i} \cdot (I',a')}=1$, by hypothesis on the function $g$, amongst the leaves $v \in \Leaves$ such that $\confIn{I'}{v}=1$, only the ones where $\histP{I'}{i} \cdot (I',a')$ is a prefix of $\histP{v}{i}$ actually participate in the sum. Therefore, we have
\begin{align}
\theta^*(\bestpayoff{i}{I'}) &
=
\max_{a' \in \actions{I'}} \max_{\beta \in \Sigma_i^{I\preceq}\,\mid\,\confP{i}{\beta}{\histP{I'}{i} \cdot (I',a')}=1}\, \sum_{v \in \Leaves \mid \histP{I'}{i} \cdot (I',a') \preceq \histP{v}{i}} g(\beta,v)\notag\\
&= \max_{a' \in \actions{I'}} \theta^*(\bestpayoff{i}{\histP{I'}{i} \cdot (I',a')}) \label{eq:proof-best deviation payoff-relaxation}
\end{align}
This equation allow us to show that $\theta^*$ satisfies the linear constraints \eqref{eq:generic-bestpayoff-I}. We can therefore conclude that $\theta^*$ is solution of \Cref{eq:generic-bestpayoff-I,eq:generic-bestpayoff-h}.

\medskip

The proof of the second proposition is simply done by backwards induction on the tree using \Cref{eq:proof-best deviation payoff-sum,eq:proof-best deviation payoff-relaxation}. In particular, in the base case, we will have $\theta^*(\bestpayoff{i}{h}) = \theta(\bestpayoff{i}{h})$. Then, as $\theta(\bestpayoff{i}{I'}) \geq  \theta(\bestpayoff{i}{\histP{I'}{i}\cdot(I',a')})$ for all $a' \in \actions{I'}$, we deduce from \Cref{eq:proof-best deviation payoff-relaxation} and our inductive hypothesis that:
\[
\theta(\bestpayoff{i}{I'}) \geq \max_{a' \in \actions{I'}} \theta(\bestpayoff{i}{\histP{I'}{i}\cdot(I',a')}) \geq \max_{a' \in \actions{I'}} \theta^*(\bestpayoff{i}{\histP{I'}{i}\cdot(I',a')}) = \theta^*(\bestpayoff{i}{I'})
\]
Finally, as $\theta(\bestpayoff{i}{h}) = \sum_{v \in \Leaves \mid \histP{v}{i} = h} f(v) 
+ \sum_{J \in \InfSet{i} \mid \histP{J}{i} = h} \theta(\bestpayoff{i}{J})$, we deduce from \Cref{eq:proof-best deviation payoff-sum} and our inductive hypothesis that:
\[
\theta(\bestpayoff{i}{I'}) \geq \sum_{v \in \Leaves \mid \histP{v}{i} = h} f(v) 
+ \sum_{J \in \InfSet{i} \mid \histP{J}{i} = h} \theta^*(\bestpayoff{i}{J})\geq \theta^*(\bestpayoff{i}{I'})\qedhere
\]
\end{proof}


\section{The \textsc{Threshold} Problem for \AFCE and \AFCCE is \NP-hard even for 2-Player Games without Chance Node}
\label{app-AFCENPhard}

\begin{figure}[ht]
\begin{center}
\newcommand{\gadgetAFCE}[3]{
    \node[term,label=below:{\scriptsize $#2$},anchor=center,below left = 0.7cm and 0.4cm of #1.center] (v1) {};
\node[term,label=below:{\scriptsize $#3$},anchor=center,below right = 0.7cm and 0.4cm of #1.center] (v2) {};
\draw[-latex] (#1) edge node[auto,swap] {$\top$} (v1);
\draw[-latex] (#1) edge node[auto] {$\bot$} (v2);
}
\scalebox{0.8}{
\begin{tikzpicture}
    \tikzset{clause/.style={rectangle, draw,fill=red!20}}
    \tikzset{var/.style={circle, draw,fill=cyan!15}}
    \tikzset{term/.style = {circle,draw,inner sep = 1.5,fill=black}}

\node[clause,anchor=center] (R) at (0,0) {$R$};

\coordinate[left = 4.55cm of R.center,anchor=center] (Tmid);

\node[term,label=below:{\scriptsize $\textcolor{red}{-1}/\textcolor{blue}{2}$},anchor=center,below left = 1cm and 4.5cm of R.center] (T) {};

\node[clause,below = 1cm of R.center,anchor=center] (C2) {$C_2$};
\node[clause,anchor=center,left = 2cm of C2.center] (C1) {$C_1$};
\node[clause,anchor=center,right = 2cm of C2.center] (C3) {$C_3$};

\node[right = 3cm of C3.center,anchor=west] (Clause1) {clause $C_1$: $x$};
\node[below = 0.5cm of Clause1.west,anchor=west] (Clause2) {clause $C_2$: $\bar{x} \vee y$};
\node[below = 0.5cm of Clause2.west,anchor=west] (Clause3) {clause $C_3$: $\bar{x} \vee \bar{y}$};
\node[below = 0.5cm of Clause3.west,anchor=west] (Formula) {formula $\varphi = C_1 \wedge C_2 \wedge C_3$};

\draw (R) -- (Tmid);
\path (Tmid) edge[-latex]  node[auto,swap] {$\mathtt{end}$} (T);
\draw[-latex] (R) edge node[auto,swap] {$c_1$} (C1) ;
\draw[-latex] (R) edge node[auto,swap] {$c_2$} (C2) ;
\draw[-latex] (R) edge node[auto]  {$c_3$} (C3) ;

\node[var,below left = 1.5cm and 2cm of C1.center,anchor=center,inner sep=2pt] (X1) {$N_{x,1}$};
\node[var,below left = 1.5cm and 0cm of C1.center,anchor=center,inner sep=2pt] (X2) {$N_{\bar{x},2}$};
\node[var,below left = 1.5cm and -2cm of C1.center,anchor=center,inner sep=2pt] (X3) {$N_{\bar{x},3}$};

\node[var,below right = 1.5cm and 2cm of C2.center,anchor=center,inner sep=2pt] (Y1) {$N_{y,1}$};
\node[var,below right = 1.5cm and 4cm of C2.center,anchor=center,inner sep=2pt] (Y2) {$N_{\bar{y},2}$};

\draw[dashed,thick, blue] (X1) -- (X2);
\draw[dashed,thick, blue] (X2) -- (X3);
\draw[dashed,thick, blue] (Y1) -- (Y2);

\draw[-latex] (C1) edge node[auto,swap,pos=0.3,inner sep=1pt] {$x\vphantom{\overline{y}}$} (X1);
\draw[-latex] (C2) edge node[auto,swap,pos=0.3,inner sep=1pt] {$\bar{x}\vphantom{\overline{y}}$} (X2);
\draw[-latex] (C3) edge node[auto,swap,pos=0.3,inner sep=1pt] {$\bar{x}\vphantom{\overline{y}}$} (X3);

\draw[-latex] (C2) edge node[auto,pos=0.3,inner sep=1pt] {$y\vphantom{\overline{y}}$} (Y1);
\draw[-latex] (C3) edge node[auto,pos=0.3,inner sep=1pt] {$\overline{y}$} (Y2);

\gadgetAFCE{X1}{\textcolor{red}{-1},\textcolor{blue}{2}}{\textcolor{red}{0},\textcolor{blue}{0}}

\gadgetAFCE{X2}{\textcolor{red}{0},\textcolor{blue}{0}}{\textcolor{red}{-1},\textcolor{blue}{2}}

\gadgetAFCE{X3}{\textcolor{red}{0},\textcolor{blue}{0}}{\textcolor{red}{-1},\textcolor{blue}{2}}

\gadgetAFCE{Y1}{\textcolor{red}{-1},\textcolor{blue}{2}}{\textcolor{red}{0},\textcolor{blue}{0}}

\gadgetAFCE{Y2}{\textcolor{red}{0},\textcolor{blue}{0}}{\textcolor{red}{-1},\textcolor{blue}{2}}

\end{tikzpicture}
}
\end{center}
\caption{Schematic view of the game~$G_\varphi$ constructed in the {\NP}-hardness reduction from the formula~$\varphi =  \protect\underbrace{x}_{\text{clause } c_1} \bigwedge \protect\underbrace{(\, \bar{x} \, \vee \, y )}_{\text{clause } c_2} \bigwedge \protect\underbrace{(\, \bar{x} \, \vee \, \bar{y} )}_{\text{clause } c_3}$
}
\label{fig:AFCE hardness2}
\end{figure}

In this section, we present the details of our reduction from 3-SAT to the \textsc{Threshold} problem for \AFCE. 
Consider $\phi$, a 3-SAT formula over set of variables $X = \{x_1,\dots,x_n\}$ with $m$ clauses $c_1,\dots,c_m$ where $c_k =  (\ell_{k,1} \vee \ell_{k,2} \vee \ell_{k,3})$, i.e. 
\[
\phi = \bigwedge_{k=1}^m (\ell_{k,1} \vee \ell_{k,2} \vee \ell_{k,3})
\]

The game involves two players: the $\varphi$-player (shown in blue in \Cref{fig:AFCE hardness2}), who aims to satisfy the formula~$\varphi$, and a \emph{spoiler} player $S$ (shown in red), whose goal is to falsify it. We identify the spoiler as player $1$ and the $\phi$-player as player 2.

The game has a root node $R$ that belongs to spoiler player $S$. At node $R$, there are $m+1$ actions, $\actions{R} = \{\mathtt{end}\} \cup \{c_i\}_{i=1}^m$. Playing $\mathtt{end}$ at $R$ leads to a terminal node with payoffs $(-1,2)$ where $-1$ is the payoff for the spoiler and $2$ is the payoff for the $\varphi$-player. Otherwise for all $i \in \{1, \ldots, m\}$,  playing action $c_i$ leads to a node $C_i$ of player $S$, i.e. $R \xrightarrow{c_i} C_i$. Player $S$ has perfect information. At node $C_i$, player $S$ can choose a literal $\ell$ from the clause $c_i$, i.e. $\actions{C_i} = \{\ell \mid \ell \in c_i\}$. Playing $\ell$ goes to a node $N_{\ell,i}$ of $\varphi$-player. All nodes of $\phi$ with literals corresponding to the same variable $x \in X$ are in one information set $I_x$. Finally, at information set $I_x$, $\phi$ player has actions $\top$ or $\bot$, i.e. $\forall x \in X, \actions{I_x} = \{\top,\bot\}$.
From $N_{\ell,i}$, $\top$ or $\bot$ leads to a terminal node with payoff $(-1,2)$ if the literal is satisfied by the action; otherwise the payoff is $(0,0)$.

Any pure strategy profile $\bsigma$ gives a unique assignment $\theta_{\bsigma}$ given by $\theta_{\bsigma}(x) = \bsigma(I_x)$. Note that, as there are no chance nodes, each strategy profile $\bsigma$ leads to a unique terminal node from the root node $R$. 

\begin{lemma}
\label{lem:pure AFCE}
    If $\phi$ is satisfiable, there is a pure \AFCE(\AFCCE) with expected social value $1$ in $G_\phi$.
\end{lemma}

\begin{proof}
    Let $\theta$ be any satisfying assignment of $\phi$. The following pure strategy profile $\bsigma$, is an AFCE: $\bsigma(R) = \mathtt{end}$; for all $i \in \{1, \ldots, m\}$, $\bsigma(C_i) = \ell \in c_i$ such that $\theta \models \ell$; and for all $x \in X$, $\bsigma(I_x) = \theta(x)$. The only deviation that matters is at node $R$ by definition of \AFCE. But path conforming to $\bsigma$ from each node $C_i$ leads to node with payoff $(-1,2)$ which is not profitable for $S$ player. Hence this is an AFCE. Since an AFCE is also an AFCCE, this is also an AFCCE.
\end{proof}

\begin{lemma}
\label{lem:no satisfiable AFCE}
    If $\phi$ is not satisfiable, then there are no \AFCE(\AFCCE) with social welfare $1$ in $G_\phi$.
\end{lemma}

\begin{proof}
We show that there is no \AFCCE with social welfare $1$. Since any \AFCE is also an \AFCCE, this will also imply that there is no \AFCE with social welfare $1$.

Let $\CorrPlan{}$ be an \AFCCE and $\bsigma$ be a strategy profile in the support of $\correp$. 
As $\phi$ is not satisfiable, we have $\theta_\bsigma \not \models \phi$, meaning that there is a clause $c_i$ such that $\theta_\bsigma \not \models c_i$. Therefore, in $\bsigma$, the node $C_i$ leads to a leaf with payoff $(0,0)$. Hence, $\expect_{\mu[R \mapsto C_i]}(U_1) \geq -1 + \CorrPlan(\bsigma) > -1$. 
As $\mu$ is an \AFCCE, we have $\expect_{\mu}(U_1) \geq \expect_{\mu[R \mapsto C_i]}(U_1) > -1$, and as social welfare at any leaf is $-u_1$, it follows that $\expect_{\mu}(\omega) < 1$.

\end{proof}

\begin{corollary}
\label{cor:AFCE}
$G_\phi$ has a pure \AFCE (\AFCCE) if and only if $\phi$ is satisfiable. 
\end{corollary}

\begin{proof}
The left implication holds by applying \Cref{lem:pure AFCE}. For the right implication, consider a pure \AFCE or \AFCCE $\bsigma$. By construction, its social welfare is either $0$ or $1$. In the latter case, we conclude by \Cref{lem:no satisfiable AFCE}. The former case is in fact in contradiction with $\bsigma$ being an \AFCE or \AFCCE as the $\varphi$-player would deviate to obtain a payoff of $2$.
\end{proof}

\AFCEhard*

\begin{proof}
Direct by application of \Cref{cor:AFCE} and \Cref{lem:no satisfiable AFCE,lem:pure AFCE}.
\end{proof}
\section{Complexity of \textsc{Threshold} and \textsc{Any} for Nash}

\paragraph{Mixed and Behavioral Strategies}

A \emph{mixed strategy for player~$i$}  is a probability distribution over the pure strategies of player~$i$ with $\Delta(\Sigma_i)$ being the set of all such mixed strategies. 
A \emph{mixed strategy profile} is a tuple $\boldsymbol{\alpha} = (\alpha_1,\ldots,\alpha_n) \in \Delta(\Sigma_1) \times \ldots \times \Delta(\Sigma_n)$.
A \emph{behavioral strategy profile $\bprofile$} is a function $\tau : \InfSets \mapsto \Delta(\Actions)$ such that $\bprofile(I)(a) > 0 \Rightarrow a \in \actions{I}$.
A \emph{behavioral strategy for player~$i$} $\tau_i$ is the restriction of a behavioral strategy profile to $\InfSet{i}$.
We often write $\tau(a)$ instead of $\tau(I)(a)$.

A mixed strategy $\alpha_i \in \Delta(\Sigma_i)$ and a behavioral strategy $\tau_i$ of player $i$ of any kind (mixed or behavioral) are said to be \emph{realization equivalent} if for all leaves $v \in \Leaves$, the probability of reaching $v$ are same under the two strategies assuming all other actions are played on the path from root to $v$, that is
\[
\prod_{(I,a) \in \histP{v}{i}} \tau_i(I)(a) = \sum_{\sigma \in \Sigma_i} \alpha_i(\sigma) \, \confP{i}{\sigma}{\node}
\]

\begin{theorem}[Kuhn~{\cite{Kuhn+1953+193+216}}]
\label{thm:Kuhn}
In games with perfect recall, every mixed strategy of a player has a realization equivalent behavioral strategy (and vice versa).
\end{theorem}

A Nash equilibrium is a mixed  strategy profile $\boldsymbol{\alpha} = (\alpha_1,\ldots, \alpha_n)$ in which no player $i \in \{1, \ldots, n\}$ can achieve a higher expected payoff by deviating from their strategy $\alpha_i$, assuming all the other players play their respective strategies in the profile. Formally, $\boldsymbol{\alpha}$ is a Nash equilibrium when
\begin{equation}
\forall i \in \{1, \ldots, n\}, \forall \tau \in \Delta(\Sigma_i), \ExpBasic{i}{\boldsymbol{\alpha}} \geq \ExpBasic{i}{\boldsymbol{\alpha}[\tau]}
\tag{\NashE}\label{eq:NE}
\end{equation}
where $\boldsymbol{\alpha}[\tau] = (\alpha,\ldots, \alpha_{i-1},\tau,\alpha_{i+1},\ldots, \alpha_n)$.


As a result of \Cref{thm:Kuhn}, Nash equilibrium \eqref{eq:NE} can be redefined in term of behavioral strategy profile $(\tau_1,\ldots, \tau_n)$ without losing expressiveness as follows: for all players $i \in \{1, \ldots, n\}$ and all behavioral strategies $\tau'_i$ for player $i$, 
\[
\ExpP{i}{(\tau_1,\ldots, \tau_n)} \geq \ExpP{i}{(\tau_1,\ldots, \tau_{i-1},\tau'_i,\tau_{i+1},\ldots, \tau_n)} 
\]

\paragraph{Probability space for behavioral profiles}
\label{subsec:behavioralproba}
In the study of behavioral strategies and profiles, we consider the probability space on the leaves $\Leaves$. If $\bprofile$ is a behavioral strategy profile and $v_0$ the root node, we define the natural probability measure $\bproba$ as follows:
    \begin{align*}
        \proba_\bprofile(v_0)     & = 1                                 &  \\
        \proba_\bprofile(\node)      & = \proba_\bprofile(w) \cdot \bprofile(a)   & \text{if }w \xrightarrow{a} \node \text{ and } a\in\Actions\\
        \proba_\bprofile(\node)      & = \proba_\bprofile(w) \cdot x              & \text{if }w \xrightarrow{x} \node \text{ and } x \in \mathbb{Q}\\
    \end{align*}

We also define another measure, denoted $\proba_\bprofile^\node$, which represent the probability of reaching a node $w$ starting from the node $\node$: 

    \begin{align*}
        \proba_\bprofile^\node(\node)     & = 1                                 &  \\
        \proba_\bprofile^\node(u)      & = \proba_\bprofile^\node(w) \cdot \bprofile(a)   & \text{if }w \xrightarrow{a} u \text{ and } a\in\Actions\\
        \proba_\bprofile^\node(\node)      & = \proba_\bprofile^\node(w) \cdot x              & \text{if }w \xrightarrow{x} u \text{ and } x \in \mathbb{Q}\\
    \end{align*}

Note that for a leaf $\ell$, if $\proba_\bprofile(v) \neq 0, \proba_\bprofile^\node(\ell) = \proba_\bprofile(\ell \mid \node)$. However, $\proba_\bprofile^\node$ is well-defined even when $\proba_\bprofile(v) = 0$.

\subsection{The \textsc{Threshold} problem for Nash is in \texorpdfstring{$\exists\mathbb{R}$}{ETR}}
\label{sec:Nash ETR}

Given a threshold $\lambda \in \Rat$ and an objective function $\omega$, we translate the \textsc{Threshold} problem for Nash equilibria into a $\etr$ formula which will check the existence of a Nash equilibria with expected value of the objective function greater than $\lambda$.

Firstly, a behavioral strategy of each player $i$ can be expressed using variables $x_{I,a}$ for each action $a \in \actions{I}$ in an information set $I \in \InfSet{i}$ by satisfying the following constraints:
\begin{align}
\sum_{a \in \actions{I}} x_{I,a} = 1 &\qquad \forall i \in \{1, \ldots, n\}, \forall I \in \InfSet{i}\label{eq:Nash-behavioral-sum}\\
x_{I,a} \geq 0&\qquad \forall i \in \{1, \ldots, n\}, \forall I \in \InfSet{i},\forall a \in \actions{I}\label{eq:Nash-behavioral-positive}
\end{align}
The expected payoff of player $i$ is can be expressed using the variable $U_i$ by following the constraints:
\begin{align}
U_i = \sum_{\node \in \Leaves} \payoffP{\node}{i}\, \Reach{\node}\, \prod_{i=1}^n\, \prod_{(I,a) \in \histP{\node}{i}} x_{I,a}&\qquad \forall i \in \{1, \ldots, n\}\label{eq:Nash-expected payoff}
\end{align}

Given a behavioral strategy profile, the probability of reaching a leaf $v$ following the behavioral strategies of all player except $i \in \{1, \ldots, n\}$, denoted $p_i(v)$, is
\[
	p_i(v) =  \Reach{v}\, \prod_{j \in \{1, \ldots, n\}\, \mid\, j\neq i} \, \prod_{(I,a) \in \histP{v}{j}} x_{I,a}
\]
This allows us to characterize the best response for player $i$ to the strategy profile $\tau$ through a set of linear inequalities. To that purpose, we introduce some new variables $\bestpayoff{i}{I}$ (resp. $\bestpayoff{i}{(I,a)}$) representing the \emph{best payoff of player $i$ reaching $I \in \InfSet{i}$ (resp. and playing $a \in \actions{I}$)}. Note that in $\bestpayoff{i}{(I,a)}$, $(I,a)$ is the singleton history. We will also consider the variable $\bestpayoff{i}{\varepsilon}$ corresponding to the \emph{best payoff of player $i$ at the start of the game}.
As in our procedure for previous notion of equilibrium, theses variables can be computed bottom-up on the game tree via the following constraints for all players $i \in \{1, \ldots, n\}$, for all information sets $I \in \InfSet{i}$, for all histories $h \in \{ h' \preceq \histP{v}{i} \mid v \in \Leaves\}$,
\begin{align}
\bestpayoff{i}{h} &= \sum_{v \in \Leaves\, \mid\, \histP{v}{i} = h} \payoffP{v}{i}\, p_i(v) + \sum_{J \in \InfSet{i}\, \mid\,\histP{J}{i} = h} \bestpayoff{i}{J}\label{eq:Nash-best payoff I a}\\
\bestpayoff{i}{I} &\geq \bestpayoff{i}{\histP{I}{i}\cdot(I,a)}&\forall a \in \actions{I}\label{eq:Nash-best payoff I}
\end{align}
The incentive constraint naturally corresponds to $U_i$ being greater than $\bestpayoff{i}{\varepsilon}$, as given below.
\begin{equation}
  U_i \geq \bestpayoff{i}{\varepsilon} \qquad \forall i \in \{1, \ldots, n\} \label{eq:Nash-incentive constraints}
\end{equation}
Finally ensuring that the expected value of the objective function is at least the threshold $\lambda$ is given by the following constraint:
\begin{align}
	\omega(U_1,\ldots,U_n) \geq \lambda\label{eq:Nash-threshold}
\end{align}

\NashExistsRCom*

\begin{proof}[Sketch proof]
Consider the system $S$ of constraints formed by \Cref{eq:Nash-behavioral-sum,eq:Nash-behavioral-positive,eq:Nash-expected payoff,eq:Nash-best payoff I a,eq:Nash-best payoff I,eq:Nash-incentive constraints,eq:Nash-threshold}. Notice that $S$ is an this system has polynomially many constraints and variables. The proof follows by proving that $S$ admits a solution if and only if there exists a Nash equilibrium with expected value of objective function at least $\lambda$. 
The contraints in \cref{eq:Nash-best payoff I,eq:Nash-best payoff I a} are local optimality constraints similar to \cref{eq:EFCE-relaxation,eq:EFCE-best deviation payoff} and the proof follows analogous arguments.
\end{proof}

\subsection{The {\sc Any} problem for Nash is in \fixp}
\label{sec:NashFIXP}

Here we will provide the proof for \Cref{thm:AnyNashisFixP-complete}.

\AnyNashisFixPcomplete*

We define few notations necessary for this proof. 
\subsubsection*{First and next information sets and leaves}

We define the \emph{initial information sets of player~$i$} and \emph{initial leaves of player~$i$} as follows:
   \begin{align*}
       \texttt{Fst}_i &= \{I \in \mathcal{I}_i \mid \histP{I}{i} = \epsilon\}\enspace;\\
       \texttt{fst}_i &= \{\ell \in \Leaves \mid \histP{\ell}{i} = \epsilon\}\enspace.
   \end{align*}

Note that the initial sets of player~$i$ (leaves of player~$i$)
form a partition in the probability space. 
Furthermore, the probability of reaching one of these information sets or leaves is independent of player~$i$'s strategy. Likewise, we define the \emph{next information sets of player~$i$ after an action $\action \in \ActionsP{i}$} and \emph{next leaves of player~$i$ after an action $\action \in \ActionsP{i}$}  as follows:
    \begin{align*}
        \texttt{Nxt}_i(a) &= \{I \in \mathcal{I}_i \mid \histP{I}{i} \text{ ends with }\action\}\enspace;\\
        \texttt{nxt}_i(a) &= \{\ell \in \Leaves \mid \histP{\ell}{i} \text{ ends with }\action\}\enspace.
    \end{align*}

Again, the next sets and leaves of player~$i$ form a partition of the event ``$a$'', and the conditionnal probability of reaching one of these information sets or leaves from $a$ is independent of player~$i$'s strategy.
We say that an information set $I$ of player $i$ is \emph{reachable} under a partial behavioral strategy profile $\bprofile_{|-i}$ if there exists a strategy $\alpha$ such that $\mathbb{P}_{\bprofile_{|-i}[\alpha]}(I) > 0$.

Our next proposition states that the probability of reaching a specific vertex in an information set of a player $i$, conditional on reaching that information set, is independent of the strategy of player $i$. It is a direct consequence of the perfect recall hypothesis.

\begin{proposition}
    \label{prop:distribution}
    Let $\bprofile_{|-i}$ be a partial behavioral strategy profile excluding player~$i$ and $I$ be an information set of player~$i$ reachable under it. Then, for any two behavioral strategies $\alpha$ and $\beta$ for player~$i$, 
    $$\mathbb{P}_{\bprofile_{|-i}[\alpha]}(v \mid I) = \mathbb{P}_{\bprofile_{|-i}[\beta]}(v \mid I)$$

    \end{proposition}

\begin{proof}
    Let $\alpha$ be a strategy such that $\mathbb{P}_{\bprofile_{|-i}[\alpha]}(I) > 0$ and let $v$ be a vertex in $I$. We define $p_{\bprofile_{|-i}}(v)$ and $p_\alpha(I)$ as follows:
    \begin{align}
        p_{\bprofile_{|-i}}(v) & = \prod_{a \in \histP{v}{-i}}\bprofile_{|-i}(a)\notag\\
        p_{\alpha}(I) & = \prod_{a \in \histP{I}{i}}\alpha(a)\notag\\
    \intertext{By definition of $\proba_{\bprofile_{|-i}[\alpha]}$:}    
        \mathbb{P}_{\bprofile_{|-i}[\alpha]}(v) & =  \Reach{v} \cdot \prod_{a\in\hist{v}}\bprofile_{|-i}[\alpha](a)\notag\\
        \mathbb{P}_{\bprofile_{|-i}[\alpha]}(v) & =  \Reach{v} \cdot \prod_{a \in \histP{v}{-i}}\bprofile_{|-i}(a) \cdot \prod_{a \in \histP{v}{i}}\alpha(a)\notag\\
    \intertext{By the perfect recall hypothesis, we have $\histP{v}{i} = \histP{I}{i}$.}
        \mathbb{P}_{\bprofile_{|-i}[\alpha]}(v) & =  \Reach{v} \cdot p_{\bprofile_{|-i}}(v) \cdot \prod_{a \in \histP{I}{i}}\alpha(a)\notag\\
        \mathbb{P}_{\bprofile_{|-i}[\alpha]}(v) & =  \Reach{v} \cdot p_{\bprofile_{|-i}}(v) \cdot p_\alpha(I)\label{eqn:anynash1}\\
    \intertext{Summing over the vertices of $I$, we get:}
        \mathbb{P}_{\bprofile_{|-i}[\alpha]}(I) & = \sum_{w\in I}\big( \Reach{w} \cdot p_{\bprofile_{|-i}}(w) \cdot p_\alpha(I)\big)\notag\\
        \mathbb{P}_{\bprofile_{|-i}[\alpha]}(I) & = \bigg(\sum_{w\in I}\big( \Reach{w} \cdot p_{\bprofile_{|-i}}(w)\big)\bigg) \cdot  p_\alpha(I)\label{eqn:anynash2}\\
    \intertext{Since reaching $v$ implies reaching $I$, we have:}
        \mathbb{P}_{\bprofile_{|-i}[\alpha]}(v \mid I) & = \frac{\mathbb{P}_{\bprofile_{|-i},\alpha}(v)}{\mathbb{P}_{\bprofile_{|-i},\alpha}(I)}\notag\\
    \intertext{From (\ref{eqn:anynash1}) and (\ref{eqn:anynash2}), we get:}
        \mathbb{P}_{\bprofile_{|-i}[\alpha]}(v \mid I) & = \frac{\Reach{v} \cdot p_{\bprofile_{|-i}}(v) \cdot p_\alpha(I)}{\bigg(\displaystyle\sum_{w\in I}\big( \Reach{w}  \cdot p_{\bprofile_{|-i}}(w)\big)\bigg) \cdot  p_\alpha(I)}\notag\\
        \mathbb{P}_{\bprofile_{|-i}[\alpha]}(v \mid I) & = \frac{\Reach{v} \cdot p_{\bprofile_{|-i}}(v)}{\displaystyle\sum_{w\in I}\big(\Reach{w} \cdot p_{\bprofile_{|-i}}(w)\big)}\notag
    \end{align}

    Therefore, $\mathbb{P}_{\bprofile_{|-i},\alpha}(v \mid I)$ is independent from $\alpha$ and \Cref{prop:distribution} follows.
\end{proof}

\Cref{prop:distribution} allows us to define a probability distribution $\mathbb{D}_{\bprofile_{|-i}}^I$ over the vertices of $I$ such that, for any strategy $\alpha$ of player $i$:
    \begin{align*}
        \mathbb{D}_{\bprofile_{|-i}}^I(v) & = \mathbb{P}_{\bprofile_{|-i}[\alpha]}(v \mid I)\\
    \intertext{In turn, we define the \emph{probability measure of a profile $\bprofile$ starting from an information set $I$}:}
        \proba_\bprofile^I(\leaf) & = \sum_{v\in I} \left(\Distrib^{I}(v) \cdot \proba_\bprofile^v(\leaf)\right)
    \intertext{As well as the \emph{probability measure of a profile $\bprofile$ starting from an action $a$}:}
        \proba_\bprofile^a(\leaf) & = \sum_{v\in I} \left(\Distrib^{I}(v) \cdot \proba_\bprofile^{a(v)}(\leaf)\right)
    \end{align*}
Note that these probability measures are defined only when the information set $I$, or the information set from which $a$ is playable, is reachable under $\bprofile_{|-i}$.

\begin{remark}
\label{rem:frommid}
    If $\mathbb{P}_\bprofile(I) > 0$, we have 
    \begin{align*}
        \mathbb{E}_\bprofile^I(U_i) & = \mathbb{E}_\bprofile(U_i \mid I)\enspace;\\
        \intertext{likewise, if $\mathbb{P}_\bprofile(a) > 0$, we have:}\\
        \mathbb{E}_\bprofile^a(U_i) & = \mathbb{E}_\bprofile(U_i\mid a) \enspace.
    \end{align*}
\end{remark}

\label{def:delta}
    Let $\bprofile$ be a behavioral strategy profile and let $a\in \actions{I}$ be an action of player $i$. We define $\delta(\bprofile,a)$ as follows:
    \begin{align*}
        \delta(\bprofile,a) & = 0 \text{ if $I$ is not reachable under $\bprofile_{|-i}$}\\
        \delta(\bprofile,a) & = \mathbb{E}_{\bprofile}^a(U_i) - \mathbb{E}_\bprofile^I (U_i)
    \end{align*}

Next we show that the mean of the value $\delta$ of the actions playable from a given information set is $0$.

\propSumDelta*
\begin{proof}
    If $I$ is not reachable under $\bprofile_{|-i}$, then for all $a \in \actions{I}$, $\delta(\bprofile,a) = 0$ and Proposition~\ref{prop:sumdelta} ensues.
    Otherwise, we have:
        \begin{align*}
        \sum_{a \in \actions{I}} \bprofile(a) \cdot \delta (\bprofile,a)     & = \sum_{a \in \actions{I}} \bprofile(a) \cdot \delta (\bprofile,a)\\
                                                                    & = \sum_{a \in \actions{I}} \left(\bprofile(a) \cdot \big(\mathbb{E}_{\bprofile[a]}^I(U_i) - \mathbb{E}_\bprofile^I (U_i)\big)\right) \\
                                                                    & = \sum_{a \in \actions{I}} \left(\bprofile(a) \cdot \mathbb{E}_{\bprofile[a]}^I(U_i)\right) - \sum_{b \in \actions{I}} \left(\bprofile(b) \cdot \mathbb{E}_\bprofile^I (U_i)\right)\\
                                                                    & = \mathbb{E}_\bprofile^I(U_i) - 1 \cdot \mathbb{E}_\bprofile^I (U_i)\\
                                                                    & = 0
    \end{align*}
\end{proof}

Now we define the function, whose fixed point will determine Nash equilibria. For actions $a,b \in \Actions$, let $a \sim b$ denote, for some $I$, $a \in \actions{I}$ and $b \in \actions{I}$.
    We define the function $f$, from the set of behavioral strategy profiles $\bprofileSet$ to itself as:
    \begin{align*}
    f(\bprofile)(a)& = \frac{\bprofile(a)+\max\big(0, \delta(\bprofile,a)\big)}{1+\displaystyle\sum_{b|b \sim a}\max\big(0,\delta(\bprofile,b)\big)}\enspace,\\
\end{align*}

As the set of behavioral strategy profiles $\bprofileSet$ is compact and convex, it follows from Brouwer's fixed-point theorem~\cite{Brouwer1912} that $f$ has at least one fixed point.

\lemmaNegativeDelta*
\begin{proof}
    Assume, by contradiction, that $\delta(\bprofile,a) > 0$.
    \paragraph{Case 1: $\bprofile(a) = 0$}
    We have:
    \begin{align*}
        f(\bprofile)(a)   &= \frac{\bprofile(a)+\max\big(0, \delta(\bprofile,a)\big)}{1+\displaystyle\sum_{b \sim a}\max\big(0,\delta(\bprofile,b)\big)}\\
                    &= \frac{\delta(\bprofile,a)}{1+\displaystyle\sum_{b \sim a}\max\big(0,\delta(\bprofile,b)\big)}\\
                    &> 0\\
                    &\neq \bprofile(a)
    \end{align*}
    Therefore, $\bprofile$ is not a fixed-point of $f$.
    \paragraph{Case 2: $\bprofile(a) > 0$}
    By Proposition~\ref{prop:sumdelta}, we have:
    \begin{align*}
        \sum_{b \sim a} \bprofile(b) \cdot \delta (\bprofile,b) &= 0\enspace.
    \end{align*}
    Since $\bprofile(a) > 0$ and $\delta(\bprofile,a)>0$, there must exist $b \sim a$ such that $\bprofile(b) > 0$ and $\delta(\bprofile,b) < 0$. We have:
    \begin{align*}
     f(\bprofile)(b)    &= \frac{\bprofile(b)+\max\big(0, \delta(\bprofile,b)\big)}{1+\displaystyle\sum_{c \sim b}\max\big(0,\delta(\bprofile,c)\big)}\\
                        &= \frac{\bprofile(b)}{1+\displaystyle\sum_{c \sim b}\max\big(0,\delta(X,I,c)\big)}\\
                        &\leq \frac{\bprofile(b)}{1+\delta(\bprofile,a)}\\
                        &< \bprofile(b)\\
                        &\neq \bprofile(b)
    \end{align*}
    Therefore, $\bprofile$ is not a fixed-point of $f$.
\end{proof}

\begin{wrapfigure}[5]{r}{0.25\textwidth}
\begin{tikzpicture}
\tikzset{p0/.style={circle, draw,minimum size=15}}
\tikzset{payoff/.style={}}

\node[p0] (n0) at (2,2.6){};
\node[p0] (n1) at (3,1.8){};
\node[payoff] (n2) at (1,1.8){$1$};
\node[payoff] (n3) at (2,1){$1$};
\node[payoff] (n4) at (4,1){$0$};

\draw [-latex,dashed](n0) -- (n1);
\draw [-latex,thick](n0) -- (n2);
\draw [-latex,dashed](n1) -- (n3);
\draw [-latex,thick](n1) -- (n4);
\end{tikzpicture}
\end{wrapfigure}
We remark that, it is not true that every Nash equilibrium is necessarily a fix-point of $f$. In the single-player game depicted on the right, the thick lines constitute a Nash equilibrium, as the player has no reason to change their action in either vertex, but the only fixed-point of $f$ is the strategy where they play left on both vertices. But it is sufficient to show that every fixed-point is a Nash equilibrium.

\begin{proposition}
\label{prop:FixedPointIsNash}
    Let $\bprofile$ be a fixed-point of $f$. Then $\bprofile$ is a Nash equilibria.
\end{proposition}

\begin{proof}
    Let $\bprofile$ be a behavioral strategy profile which is not a Nash equilibrium, and $\alpha$ be a pure behavioural strategy for some player $i$ such that:
    \begin{align*}
        \mathbb{E}_{\bprofile[\alpha]} (U_i) & > \mathbb{E}_{\bprofile}(U_i)\\
    \intertext{As the Fst$_i$ and fst$_i$ form a partition of the plays, we get:} 
        \sum_{I\in\texttt{Fst}_i} \mathbb{P}_{\bprofile[\alpha]}(I)\cdot\mathbb{E}_{\bprofile[\alpha]}(U_i\mid I) + \sum_{\ell\in\texttt{fst}_i} \mathbb{P}_{\bprofile[\alpha]}(\ell)\cdot u_i(\ell) & > \sum_{I\in\texttt{Fst}_i} \mathbb{P}_\bprofile(I)\cdot\mathbb{E}_\bprofile(U_i\mid I) + \sum_{\ell\in\texttt{fst}_i} \mathbb{P}_\bprofile(\ell)\cdot u_i(\ell)\\
    \intertext{As $I$'s and $\ell$'s in the sums are initial for player $i$, we have $\mathbb{P}_{\bprofile[\alpha]}(I) = \mathbb{P}_{\bprofile}(I)$ and $\mathbb{P}_{\bprofile[\alpha]}(\ell) = \mathbb{P}_{\bprofile}(\ell)$.}
        \sum_{I\in\texttt{Fst}_i} \mathbb{P}_{\bprofile}(I)\cdot\mathbb{E}_{\bprofile[\alpha]}(U_i\mid I) + \sum_{\ell\in\texttt{fst}_i} \mathbb{P}_{\bprofile}(\ell)\cdot u_i(\ell) & > \sum_{I\in\texttt{Fst}_i} \mathbb{P}_\bprofile(I)\cdot\mathbb{E}_\bprofile(U_i\mid I) + \sum_{\ell\in\texttt{fst}_i} \mathbb{P}_\bprofile(\ell)\cdot u_i(\ell)\\
        \sum_{I\in\texttt{Fst}_i} \mathbb{P}_{\bprofile}(I)\cdot\mathbb{E}_{\bprofile[\alpha]}(X\mid I) & > \sum_{I\in\texttt{Fst}_i} \mathbb{P}_\bprofile(I)\cdot\mathbb{E}_\bprofile(U_i\mid I)\\
    \end{align*}

    Therefore, there must exist an information set $I$ for player $i$ such that:
    \begin{align*}
        \mathbb{P}_\bprofile(I) & > 0\\
        \mathbb{E}_{\bprofile[\alpha]}(U_i\mid I) & > \mathbb{E}_\bprofile(U_i\mid I)
    \end{align*}

    By Remark~\ref{rem:frommid}, we have :
    \begin{align*}
        \mathbb{E}^I_{\bprofile[\alpha]}(U_i) & > \mathbb{E}_\bprofile^I(U_i)
    \end{align*}

Let $J$ be an information set of player $i$ such that:
\begin{align}
    \mathbb{E}^J_{\bprofile[\alpha]}(U_i) & > \mathbb{E}_\bprofile^J(U_i)\label{defJ1}\\
    \intertext{and, for any action $a \in \actions{J}$, for any information set $K$ in $\texttt{Nxt}_i(a)$:}
    \mathbb{E}^K_{\bprofile[\alpha]}(U_i) & \leq \mathbb{E}_\bprofile^K(U_i)\label{defJ2}\\
    \intertext{By (\ref{defJ1}), we have:}
    \mathbb{E}^J_{\bprofile[\alpha]}(U_i)   & > \mathbb{E}_\bprofile^J(U_i)\notag\\
    \intertext{Extending the expectation over the possible actions playable in $J$, we get:}
    \sum_{a \in \actions{J}} \bprofile[\alpha](a) \cdot \mathbb{E}^{a}_{\bprofile[\alpha]} (U_i) & > \mathbb{E}_\bprofile^J(U_i)\notag\\
    \intertext{As $J$ belongs to player~$i$:}
    \sum_{a \in \actions{J}} \alpha(a) \cdot \mathbb{E}^{a}_{\bprofile[\alpha]} (U_i) & > \mathbb{E}_\bprofile^J(U_i)\notag\\
    \intertext{Extending the expectations over the next information set of player~$i$ reached after $a$, we get:}
    \sum_{a \in \actions{J}} \alpha(a) \cdot \Bigg(\sum_{K\in\texttt{Nxt}_i^a(J)}\proba^a_{\tau[\alpha]} (K) \cdot \mathbb{E}^{K}_{\bprofile[\alpha]} (U_i)\Bigg) & > \mathbb{E}_\bprofile^J(U_i)\notag\\
    \intertext{By (\ref{defJ2}), we have:}
    \sum_{a \in \actions{J}} \alpha(a) \cdot \Bigg(\sum_{K\in\texttt{Nxt}_i^a(J)}\proba^a_{\tau} (K) \cdot \mathbb{E}^{K}_{\bprofile} (U_i)\Bigg) & > \mathbb{E}_\bprofile^J(U_i)\notag\\
    \intertext{We can contract the inner term as an expectation under $\bprofile$ starting from $a$:}
    \sum_{a \in \actions{J}} \alpha(a) \cdot \mathbb{E}^{a}_{\bprofile} (U_i) & > \mathbb{E}_\bprofile^J(U_i)\notag\\
    \intertext{As the $\alpha(a)$ sum up to 1:}
    \sum_{a \in \actions{J}} \alpha(a) \cdot \mathbb{E}^{a}_{\bprofile} (U_i) & > \sum_{a \in \actions{J}} \alpha(a) \cdot \mathbb{E}_\bprofile^J(U_i)\notag\\
    \sum_{a \in \actions{J}} \alpha(a) \cdot \left( \mathbb{E}^{a}_{\bprofile} (U_i) - \mathbb{E}_\bprofile^J(U_i)\right)& > 0\notag\\
    \intertext{By definition of $\delta$:}
            \sum_{a \in \actions{J}} \alpha(a) \cdot \delta(\bprofile, a) & > 0\notag\\
    \end{align}

    Therefore, there is an action $a \in \actions{J}$ such that $\delta(\bprofile, a) > 0$. By \Cref{prop:negative delta}, $\bprofile$ is not a fix-point of $f$, and Proposition~\ref{prop:FixedPointIsNash} ensues.
    \end{proof}

\begin{corollary}
\label{cor:AnyNFCEisFixP}
    The \Any problem for \NFCE in extensive-form games is in \fixp.
\end{corollary}

\end{document}